%% file: main.tex
\documentclass[letterpaper,twocolumn,10pt]{article}
\usepackage{usenix}

\input{packages}

\newcommand{\SADRA}{\textsc{SADRA}}

\begin{document}
 
\date{}

\title{\Large \bf \SADRA{}: Sound Capability-based Access Control  System for Resource-Disaggregated Architectures} 

\author{
{\rm Hamed Rasifard}\\
Saarland University
\and
{\rm Amir Farahani Khojasteh}\\
CISPA Helmholtz Center for Information Security
 \and
 {\rm Hamed Nemati}\\
KTH Royal Institute of Technology
 \and
 {\rm Michael Backes}\\
CISPA Helmholtz Center for Information Security
}

\maketitle

\begin{abstract}
Resource disaggregation separates memory and accelerators from compute nodes and makes them remotely accessible.  This improves resource sharing, but also removes the local kernel from the resource-access path.  Under an untrusted host, compromised host software may use stale authority, exceed delegated authority, or reuse authority provisioned for another process.  Prior work identifies capability-based access control as well suited to these architectures.  Our systematization of twenty-two prior capability systems finds that none combines host-independent validation of process authority, authoritative enforcement at the resource, and revocation that remains effective while remote authorization state is stale.        

We present \textbf{\SADRA{}}, a distributed capability-based access-control system for resource-disaggregated architectures.  SmartNIC hardware isolated from host software independently checks every inter-node request at two points,  first at the compute node against the requesting process's authority and again at the resource against the current authoritative access state.  Linked process, compute, and resource capabilities allow these checks to use local state without coordination on the access path.  When distributed authority is revoked, \SADRA{} denies subsequent dependent accesses at the resource without waiting for remote nodes to update, while stale capability state is reclaimed separately.        

We prove capability safety, authority safety, revocation soundness, and strong isolation for a formal architectural model, and model-check the formalization with SPIN. Our FPGA SmartNIC prototype sustains 89.5\,Gbit/s aggregate throughput, within run-to-run variation of a non-enforcing baseline that peaks at 90--91\,Gbit/s. Cleanup of a 128-capability subtree completes in 528\,ns, while access denial does not depend on completion of that cleanup.
\end{abstract}

\section{Introduction}
Resource disaggregation separates memory and accelerators from compute nodes and pools them over a fabric, improving utilization and elasticity~\cite{keeton2015machine,asanovic2014firebox}.  Prior work identifies capability-based access control as well suited to resource-disaggregated architectures because capabilities combine resource designation,  fine-grained authority,  and delegation in references that remote enforcement points can validate directly~\cite{bresniker2019rack}. In conventional systems, the kernel mediates the resource-access path. However, disaggregation removes the kernel from that path,  while the fabric provides no equivalent trusted mediator, so compromised host software can use stale authority,  exceed delegated authority, or reuse authority provisioned for another process~\cite{rothenberger2021rdmark,taranov2022nevermore}.

The mediation gap is structural.  Complete mediation~\cite{saltzer1975protection} requires every access to be checked against current authority.  Once authority is delegated across nodes, authorization state becomes distributed, and enforcement points may observe revocation at different times.  An enforcement point with stale state may therefore continue accepting requests under revoked authority. The hard problem is denying those requests before distributed state converges.

A concrete example shows how this problem arises.  In a genomics pipeline over disaggregated memory~\cite{becker2020genomics},  a stage delegates read access to a downstream task and later completes its work; the system should then revoke that access (\Cref{fig:gen_ex}).  If requests are not checked against process-specific authority before fabric admission,  compromised host software can continue issuing requests under authority provisioned to the completed stage.  If they are not also checked against authoritative memory-side state before data access, a request admitted under stale compute-side state can still access the protected data after revocation.  The same boundary failures arise in multi-tenant inference over shared embedding tables~\cite{naumov2019dlrm} and kernel-page swapping backed by disaggregated memory~\cite{amaro2020can}.  Neither check alone suffices.  Enforcing both without host-side authorization or coordination on the access path raises three challenges.

The first challenge is restoring complete mediation while enforcing \emph{fine-grained subject authority} under an untrusted host.  A resource-side enforcement point alone can authoritatively decide access to the protected resource, but it cannot bind a request to process-specific authority or stop unauthorized traffic before fabric admission,  because request-origin information comes only from untrusted host.  Thus, fine-grained subject authority requires a second point,  outside host control,  that validates each request against the authority of its issuing process before fabric admission.  That compute-side point,  in turn,  cannot serve as the authoritative decision point for the remote resource.  Each requirement pins enforcement to a different boundary,  and no single placement satisfies both.

The second challenge is revocation.  Once authority is revoked,  one authoritative point must deny every dependent request.  Because authority may be re-delegated across compute nodes,  waiting for remote enforcement points to update or for the affected hierarchy to be reclaimed leaves the propagation-dependent window open.  A compute node may therefore continue issuing requests under stale authority; the resource-side enforcement point must reject them whether or not that node has observed the revocation.

The third challenge is simultaneously supporting the capability-system dimensions that matter in this setting.  We adapt the rack-scale capability design space of Bresniker et al.~\cite{bresniker2019rack} to resource disaggregation under an untrusted host and define six dimensions: inter-node mediation,  subject and object granularity,  delegation,  revocation,  and persistence of capability state across failures.  The difficulty lies in how these dimensions interact; process-granularity validation and authoritative resource decisions require different boundaries,  distributed delegation makes revocation sensitive to propagation,  and persistence must restore delegated authority and prior revocations across failures without reactivating stale authority.  We evaluate twenty-two capability systems against these dimensions (\S\ref{sec:systematization}) and find that none satisfies all six.

The systems most relevant to resource-disaggregated capability enforcement illustrate how enforcement placement,  revocation, and recovery create these gaps.  CEP~\cite{azriel2019memory} enforces only at the memory node and trusts the compute-side OS to assign channels to processes.  A compromised host can therefore reuse a channel provisioned for another process,  while no trusted compute-side point validates process-specific authority before fabric admission.  SemperOS~\cite{hille2019semperos} distributes enforcement across cooperating microkernels,  but cross-kernel revocation requires coordination and delegation-tree traversal.  White et al.~\cite{white2025enabling} provide version-based invalidation,  but their controllers run in trusted host software,  and recovery invalidates stale authority rather than restoring live capability state.  FractOS~\cite{vilanova2022slashing} combines distributed delegation with immediate owner-side revocation in heterogeneous disaggregated systems,  but its trusted software Controllers do not independently validate every inter-node resource request at both the compute and resource boundaries.  Each leaves at least one of these challenges unresolved.

We present \textbf{\SADRA{}},  a distributed capability-based access-control system for resource-disaggregated architectures. \SADRA{} provides host-independent mediation of remote resource access by enforcing requests at both ends of the access path.  Hardware controllers isolated from untrusted host software validate requests against process-granularity authority before fabric admission and against authoritative resource state before resource access.

\subfile{figure_genomeic_example}

Single-boundary enforcement is structurally insufficient: process-granularity authority must be validated before fabric admission,  while the authoritative access decision must be made before resource access.  \SADRA{}'s core insight is that these checks can be enforced independently using local state while forming a \textbf{two-stage enforcement invariant},  without coordination on the access path.  Only the \emph{compute controller} can validate process-granularity authority before fabric admission; only the \emph{resource controller} can make the authoritative access decision before resource access.  Achieving this invariant requires an authority representation that lets each controller perform its check using local state.

\SADRA{} realizes the invariant by decomposing authority into linked \emph{process},  \emph{compute},  and \emph{resource} capabilities,  $c_p \mapsto c_c \mapsto c_r$.  Resource capabilities anchor authoritative access rights at the resource controller.  Compute capabilities derive from that authority and represent the rights delegated to a compute controller.  Process capabilities derive from a compute capability and associate a subset of those rights with an individual process.  Each controller maintains its local delegation state as a rooted \emph{capability tree} and validates the portion of the chain it enforces against that state.  Delegation,  including inter-node delegation,  attenuates existing authority rather than creating independently rooted authority.

The revocation point depends on where authority remains rooted.  A compute controller locally revokes authority confined to its node.  Authority delegated across compute nodes remains rooted at the authoritative resource controller,  enabling \emph{resource-rooted revocation}; the controller fences the revoked subtree and denies every dependent request at the resource boundary,  regardless of whether remote compute controllers have observed the revocation.  Because every resource access crosses that boundary, delayed propagation cannot extend resource access.  A remote compute controller discovers stale re-delegated authority after a denied request,  then fences the stale subtree and reclaims it asynchronously.  The fence determines authorization; cleanup only reclaims state.

Under the untrusted-host model, \SADRA{} satisfies four security properties,  formalized in \S\ref{sec:sec_properties_soundness}: capability safety, authority safety, revocation soundness, and strong isolation.  We prove these properties for a formal architectural model using the assumption-commitment method~\cite{misra1981proofs,de2001concurrency} and model-check the formalization with SPIN.  The argument assumes neither trusted host software nor coordination on the access path.

We implement \SADRA{} on Xilinx Alveo U55C FPGA SmartNICs.  The compute and resource controllers are synthesized hardware logic on the NIC datapath that enforces authorization inline for every protected request in a memory-disaggregated cluster,  without host software on the enforcement path.  This paper makes the following contributions:

\begin{itemize}

\item We introduce a \textbf{two-stage enforcement invariant}: the compute controller validates process-granularity authority before fabric admission,  while the resource controller makes the authoritative access decision before resource access.  Neither check alone suffices,  and the invariant holds without coordination on the access path (\S\ref{sec:sys-design}).

\item We design a two-level revocation mechanism: node-local authority is revoked at its compute controller,  while authority delegated across nodes is fenced at the authoritative resource controller,  protecting the resource without waiting for remote-state propagation.  Decomposing authority into linked process,  compute,  and resource capabilities,  $c_p \mapsto c_c \mapsto c_r$,  keeps remotely delegated authority linked to the resource capability from
which it derives (\S\ref{sec:sys-design}).

\item We prove capability safety,  authority safety,  revocation soundness,  and strong isolation for a formal architectural model using the assumption-commitment method,  and model-check the formalization with SPIN (\S\ref{sec:sec_properties_soundness}).

 \item We implement \SADRA{} on Xilinx Alveo U55C FPGAs and evaluate it on a memory-disaggregated cluster. \SADRA{} sustains 89.5\,Gbit/s, effectively reaching the 90--91\,Gbit/s throughput ceiling of the remote-memory datapath.  Under a Meta-DLRM workload, it remains within 1.8\% of baseline inference throughput while processing up to 1.42 million protected requests/s,  producing outputs bit-identical to the non-enforcing baseline in all 48 paired runs. Cleanup of 128 capabilities completes in 528\,ns and,  for a fixed capability count, is unchanged as the number of subtrees varies from one to eight (\S\ref{sec:evaluation}).
\end{itemize}


\section{Background}
\label{sec:background}

\subsection{Object-Capability Model} 
\label{sec:capability_model}  

Building on the capability model~\cite{dennis1966programming, fabry1974capability}, the object-capability model~\cite{miller2003capability, rasifard2023seal} represents authority through unforgeable references called \emph{capabilities}. A capability binds an object reference to access rights, so a subject invokes an object only by presenting a valid capability. Thus, authority is not ambient in this model; a subject can exercise only the authority conveyed by the capabilities it holds. Delegation transfers no more authority than the delegator holds, allowing authority to be attenuated but not amplified; \SADRA{} formalizes this monotonic derivation relation (\S\ref{sec:sec_properties_soundness}).  

\subfile{figure_obj_cap}

This model handles revocation via indirection (\Cref{fig:cap-obj-mdl}). To grant $B$ access to object  $C$,  $A$ interposes a forwarding proxy $F$ and a revocation point $R$, then revokes the authority by disabling $R$. This indirection suffices when a single enforcement point mediates every access. When several enforcement points check delegated authority but learn of revocation at different times, a revoked capability may remain usable at one that has not yet observed it~\cite{levy1984capability,chase1994capability}. \SADRA{} closes this window by rooting revocation of authority delegated across nodes at the authoritative resource-side controller, which denies dependent requests at the resource boundary whether or not remote compute-side controllers have observed the revocation (\S\ref{sec:sys-design}).

\subsection{Disaggregated Architecture Model} 
\label{sec:disagg-model}  

In a resource-disaggregated architecture, compute nodes access remote memory and accelerators over a high-speed fabric. We consider kernel-bypass datapaths in which requests bypass the local kernel, enter the fabric, and reach the resource node without traversing a software stack. No software mediator lies on this access path, so access control must be enforced explicitly at the interconnect boundaries. Resource-side enforcement also cannot assume the presence of software: some resource nodes expose memory or accelerators through \emph{hardware-only controllers}, with no software stack at the boundary. Each compute and resource node connects to the fabric through a SmartNIC, whose isolated datapath provides an enforcement point outside host-software control.

In multi-tenant deployments, tenants share compute and resource nodes under a hypervisor or container runtime. This creates two trust boundaries: the compute-side boundary, where requests enter the fabric, and the resource-side boundary, where requests leave the fabric before reaching the protected resource. \S\ref{sec:threat-model} defines the trust model, including the untrusted-host assumption. These boundaries provide the setting for \SADRA{}'s enforcement points (\S\ref{sec:sys-design}).

\subsection{SmartNICs as Enforcement Points} 
\label{sec:SmartNICs_bkg}  

Programmable SmartNICs connect each node to the fabric and process datapath traffic at line rate~\cite{liu2019e3,lin2020panic}. Their enforcement logic and security-critical state can be isolated from host software, allowing the host to submit requests through defined interfaces without permitting it to modify that state or bypass datapath checks. This placement provides two structural properties of a reference monitor~\cite{anderson1972reference}: enforcement is invoked on every remote request and is tamper-resistant under the untrusted-host model.  

Deployed systems such as AWS Nitro~\cite{aws2024nitro} and NVIDIA BlueField~\cite{nvidia2023bluefield} use SmartNICs to place security-critical functions outside host software control. Implementing enforcement in programmable SmartNIC hardware also allows capability checks to run as fixed pipelines with bounded latency independent of host load. Prior FPGA work such as $\mu$Shell~\cite{chen2026mushell} implements capability enforcement in programmable hardware for memory access and inter-vFPGA communication.


\section{Threat Model} 
\label{sec:threat-model}  

\noindent\textbf{System Model.} 
We consider a multi-tenant, resource-disaggregated architecture in which compute nodes issue requests over a high-speed fabric to remote resource pools. Each compute and resource node connects to the fabric through a SmartNIC containing a \SADRA{} controller. Every inter-node request to a disaggregated resource traverses the compute controller before fabric admission and the resource controller before resource access.

\noindent\textbf{Adversary.} 
The adversary controls all host-resident software, including tenant processes, kernels, and hypervisors. It may submit arbitrary requests through the host--SmartNIC interface, including requests that omit a capability, modify capability fields, present stale capabilities, reuse capabilities issued to another process, or request operations beyond the permissions a capability conveys. Compromised nodes may also inject, delay, reorder, or drop traffic that they originate. The adversary cannot modify a controller's security-critical state, bypass its datapath enforcement, or forge the integrity protection on a capability or authenticated control-plane message.  

\noindent\textbf{Trust Assumptions.} 
Controllers communicate over authenticated, integrity-protected control-plane channels. The trusted computing base (TCB) comprises the SmartNIC controller logic, its cryptographic keys, and the on-device stores holding capability, delegation, and revocation state. Any firmware or management component that can modify this state is also part of the TCB. We assume that the TCB executes correctly and that the cryptographic primitives remain secure. Host software may interact with the TCB only through defined interfaces.  

\noindent\textbf{Security Goals.} 
\SADRA{} provides \textit{capability safety}, \textit{authority safety}, \textit{revocation soundness}, and \textit{strong isolation} properties, formalized in \S\ref{sec:sec_properties_soundness}. These properties ensure that no request relying on forged, stale, or revoked authority is accepted for resource access and that every accepted inter-node request has passed both controller checks.  

\noindent\textbf{Out of Scope.} We do not address physical attacks, timing or power side channels, compromise of TCB components, or denial-of-service attacks, and provide no availability guarantees under such attacks. Our security claims concern access-control correctness under untrusted hosts.


\section{Security Properties and Soundness}
\label{sec:sec_properties_soundness}

\subsection{Authority Model}  
\SADRA{} models principals in three classes: processes ($\mathbb{P}$), compute controllers ($\mathbb{CC}$), and resource controllers ($\mathbb{RC}$). It partitions resources into disjoint regions, each managed by one resource controller, and presents authority over a region through \emph{capabilities}: integrity-protected tokens that bind a principal, a resource region, and a permission set. Authorization depends on both the validity of the presented capability and the validity of its derivation chain (\S\ref{app:formal-model}).  

\SADRA{} represents delegated authority through linked capability trees maintained by the resource and compute controllers. For each resource region, the resource controller maintains a \emph{resource capability tree} rooted at the region's resource capability. A compute capability derives from exactly one resource capability and delegates a subset of its authority to a compute controller. Each compute controller maintains a local \emph{compute capability tree} rooted at a distinguished virtual root, which serves only as a structural anchor and presents no authority. The controller inserts its compute capabilities below this root and derives process capabilities from the corresponding compute capabilities for individual processes on its node. Thus, every process capability is linked to a compute capability and, through it, to the resource capability from which its authority derives.

Fencing a resource capability in the resource capability tree causes that capability and every capability derived from it to fail validation at the resource controller. After the fence is installed, a compute controller may still retain a compute capability derived from the fenced resource capability in its local state. However, the compute capability can no longer authorize resource access because the resource controller rejects every request that presents it.

\subsection{Two-Stage Enforcement Invariant}

\begin{restatable}[Two-Stage Enforcement]{invariant}{TWOSTENF}
\label{inv:two-stage}

\SADRA{} authorizes an inter-node resource-access request if and only if
both enforcement controllers accept the request.

The compute controller authenticates the process capability, binds it to
the hardware-observed requesting process, validates the
$c_p\mapsto c_c$ authority relation, checks the requested resource extent
and operation, and rejects the request if the corresponding compute
authority is covered by an active local fence.

After this check succeeds, the compute controller replaces the process
capability with the corresponding root-anchored compute capability and
forwards the request.

The resource controller authenticates the forwarded capability, verifies
its association with the corresponding resource capability and the
originating compute node, checks the requested resource extent and
operation, and rejects the request if the corresponding resource
authority is covered by an active resource-side fence.

Resource access is permitted only after both checks succeed.
\end{restatable}

This invariant realizes complete mediation for inter-node resource
access. The compute controller enforces process-specific authority before
fabric admission but cannot make the authoritative resource-access
decision because the resource capability tree and its active fences
reside at the resource controller. Conversely, the resource controller
receives a resource-controller-created, root-anchored compute capability
rather than the process capability and therefore does not independently
establish the process-specific authorization performed at the compute
controller.

The architectural authority chain is
\[
c_p\mapsto c_c\mapsto c_r.
\]
For resource access, the compute controller first validates
$c_p\mapsto c_c$, then forwards the root-anchored compute authority
associated with $c_c$ for the resource controller to validate against
$c_r$.

Both enforcement stages are therefore necessary. A request may pass the
compute-side check while failing the resource-side check, for example
when the compute controller still observes active local authority after
the corresponding resource authority has been fenced at the resource
controller. Such a request is denied before resource access.

Fence installation does not require removal of the corresponding
capability linkage. An active fence independently disables the affected
authority at that controller and causes subsequent requests depending on
that authority to fail its enforcement check.

\subsection{Security Properties} 
\label{sec:security:props}  

The following four properties hold under the threat model (\S\ref{sec:threat-model}). Formal statements and model details appear in Appendix~\ref{app:formal-model}; complete proofs and SPIN model-checking details are deferred to a companion technical report.

\noindent\textbf{Capability Safety.}
A process can acquire a valid capability only through resource allocation or legitimate delegation. Each valid process capability must form a chain $c_p \mapsto c_c \mapsto c_r$ ending at a resource capability in the authoritative resource capability tree.  

\noindent\textbf{Authority Safety.} 
Delegation never amplifies authority. A subject can delegate only rights it already holds, and every delegated capability covers a subset of its parent's permissions and resource extent.  

\noindent\textbf{Revocation Soundness.} For revocation within one compute node, the compute controller invalidates the capability and its descendants; all later compute-controller validations of those capabilities fail. For authority delegated across nodes, revocation takes effect when the resource controller atomically installs a fence on the corresponding resource capability. All later resource-controller validations of dependent capability paths then fail without waiting for propagation to remote compute controllers or for subtree cleanup. Cleanup follows this authoritative denial and does not delay resource protection.  

\noindent\textbf{Strong Isolation.}
When a process allocates a resource region in exclusive mode, \SADRA{} makes the resulting authority non-delegable, including by the allocating process. Since independent allocations are disjoint, no other process can acquire authority over the region through either delegation or a separate allocation.

\subsection{Soundness Argument}  

Using the assumption--commitment (A--C) method~\cite{misra1981proofs,de2001concurrency}, we model the compute and resource controllers as interacting state machines with local commitments. At the compute boundary, the controller accepts a request only after authenticating the process capability, binding it to the hardware-observed requesting process, validating the $c_p \mapsto c_c$ relation, establishing that the corresponding compute authority remains present in controller state and is not covered by an installed local fence, and checking the requested resource extent and operation. At the resource boundary, the controller independently accepts the forwarded request only after authenticating the root-anchored compute capability, validating its association with the corresponding resource capability and originating compute node, establishing that the corresponding resource authority remains present in authoritative resource state and is not covered by an installed resource-side fence, and checking the requested resource extent and operation. Because every protected inter-node resource-access request necessarily traverses both enforcement points under the system model of \S\ref{sec:threat-model}, these local commitments jointly realize the two-stage enforcement invariant.

The two-stage enforcement invariant combines with the capability-creation and delegation rules to establish capability and authority safety. Delegation monotonically attenuates resource extent and permissions, while capability authentication and controller-resident authority state prevent fabricated or reclaimed authority from authorizing access. Strong isolation follows from three properties: independent resource allocations are disjoint, the $\mathtt{exclusive}$ permission is preserved through the initial resource, compute, and process capability chain, and exclusive authority is non-delegable, including by the allocating process.

For revocation rooted in a compute subtree, the compute controller installs a local fence, causing subsequent compute-side checks for the fenced capability and its descendants to fail before structural cleanup. If that subtree contains authority re-delegated to other compute nodes, the corresponding resource-side revocation requests are forwarded to the authoritative resource controller, which installs a fence on each corresponding resource capability. Each such authority is denied at the resource boundary once its resource-side fence is installed. A compute-subtree revocation with no such inter-node descendants is the purely local case.

When a process directly revokes authority that it previously delegated across nodes, the compute controller resolves the process revocation handle to the corresponding compute revocation handle and forwards a copy of that handle to the authoritative resource controller. The compute controller does not install a compute fence on the compute revocation handle because the handle carries no usable authority. The resource controller authenticates the handle, verifies that its source-node binding matches the trusted origin of the revocation request, and installs the corresponding resource-side fence. Knowledge of a resource-side identifier alone does not authorize revocation.

After the resource controller rejects a request that presents stale authority, the compute controller that forwarded the request may install a local fence over the corresponding stale compute capability and its subtree. This containment reduces stale traffic and enables subsequent reclamation, but authoritative resource protection does not depend on its completion. More generally, fence installation determines the security effect at the controller where the fence is installed, while propagation, subtree traversal, and reclamation proceed separately from the resource-access path.      

The formal architectural model and supporting invariants appear in \S\ref{app:formal-model}. Complete proofs of capability safety, authority safety, revocation soundness, and strong isolation are deferred to a companion technical report.


\subfile{table_RW}

\section{Systematization}
\label{sec:systematization}

We systematize twenty-two prior capability systems along six design dimensions tailored to resource disaggregation under an untrusted host. We adapt the dimensions most relevant to our setting from the rack-scale capability design space of Bresniker et al.~\cite{bresniker2019rack} to resource disaggregation under an untrusted host. These dimensions complement the broader taxonomy of capability systems developed by Mauthe et al.~\cite{mauthe:usenix26}. Inter-node mediation, subject granularity, and object granularity characterize whether and at what granularity a system mediates the inter-node resource-access path, while delegation, revocation, and persistence characterize how authority can evolve securely. We classify each system as providing full, partial, or no support for each dimension. No prior system satisfies all six simultaneously.

\subsection{Design Dimensions} 
\label{sec:features-required}  

We define six dimensions and classify each system as providing \emph{full}, \emph{partial}, or \emph{no} support for each.  

\noindent\textbf{Inter-Node Mediation.} Under the untrusted-host model, complete mediation requires request validation before fabric admission and an authoritative access decision before resource access; no single enforcement point can simultaneously provide both. \emph{No}: the design does not mediate an inter-node resource-access path; shared-memory multiprocessor access alone does not qualify. \emph{Partial}: it mediates access at only one boundary or requires coordination between enforcement points on the access path. \emph{Full}: it independently mediates access at both the compute and resource boundaries; authority that never crosses a node boundary may be handled locally.  

\noindent\textbf{Subject Granularity.} The finest unit at which a system distinguishes principals. \emph{No}: the system does not distinguish subjects for access control. \emph{Partial}: authority is associated with a host, VM, tenant, code domain, or principal other than an individual process. \emph{Full}: authority is associated with an individual process.  

\noindent\textbf{Object Granularity.} The finest unit of resource state over which authority can be independently granted and constrained. \emph{No}: the system does not distinguish protected resource objects. \emph{Partial}: authority is limited to predefined objects or fixed-granularity units, such as pages, segments, typed objects, or fixed-size regions. \emph{Full}: authority can independently cover arbitrary sub-object ranges, including byte-granular ranges for byte-addressable resources.  

\noindent\textbf{Delegation.} Whether a principal can transfer a subset of its authority without amplification. \emph{No}: the system does not support authority transfer. \emph{Partial}: it supports transfer but does not enforce monotonic attenuation of permissions or resource extent. \emph{Full}: every delegated capability is restricted to a subset of its parent's permissions and resource extent.

\noindent\textbf{Revocation.} How a system invalidates granted authority. \emph{No}: it provides no revocation mechanism. \emph{Partial}: revocation is coarse-grained, depends on trusted host software, or requires distributed propagation before all dependent access is denied. \emph{Full}: an authoritative revocation action causes later validations of dependent authority to fail before resource access without waiting for propagation.  

\noindent\textbf{Persistence.} Whether the architecture provides recovery semantics that preserve valid authority and prior revocations across failures or reboots without reactivating stale authority. \emph{No}: the architecture does not define recovery of prior authority, or requires capabilities to be reacquired after failure. \emph{Partial}: the architecture preserves some authority state, but capability, delegation, or revocation state may be lost or recovered inconsistently. \emph{Full}: the architecture provides recovery of capability, delegation, and revocation state such that valid authority remains available without reprovisioning and stale authority cannot be reactivated.

\subsection{Systematization Results} 
\label{sec:systematization-results}  

Support across the twenty-two prior systems varies among the six dimensions (\Cref{tab:coveredCriteria}). Subject granularity, object granularity, and delegation receive at least partial support in 95\%, 100\%, and 95\% of the systems, respectively. Revocation is supported by 91\%, but only a small subset provides authoritative revocation independent of distributed-state convergence; the remaining systems rely on coordination, traversal, or coarse OS-mediated invalidation. Inter-node mediation and persistence are substantially less common. No prior system independently validates every inter-node request for resource access at both the compute and resource boundaries. The 45\% with inter-node enforcement mediate at only one boundary, leaving either process-granularity validation before fabric admission or authoritative validation before resource access absent. Persistence receives at least partial support in 27\% of the systems and full support in two.

No prior capability system realizes all six dimensions simultaneously. Inter-node designs commonly rely on coordination, delayed propagation, or centralized state, preventing revocation from taking effect at the authoritative resource boundary independently of distributed-state convergence. Memory-disaggregated designs come closest, but each lacks either dual-boundary enforcement or the required revocation and persistence properties. Thus, we use CEP~\cite{azriel2019memory}, White et al.~\cite{white2025enabling}, SemperOS~\cite{hille2019semperos}, and FractOS~\cite{vilanova2022slashing} as comparison baselines. CEP and White et al.\ target memory-disaggregated access, SemperOS provides the closest comparison for distributed delegation and revocation through software capability management, and FractOS provides the closest comparison for immediate owner-side revocation in disaggregated systems.

\subsection{State of the Art} 
\label{sec:state-of-the-art-comparison}  

\noindent\textbf{CEP.} 
CEP~\cite{azriel2019memory} enforces capabilities at the memory node, where a coprocessor mediates access after the compute node's MMU checks. A process is associated with its CEP context through compute-side fabric identifiers, so CEP relies on the application OS to preserve channel integrity. Compromised host software can therefore misuse a channel provisioned for another process. The prototype dedicates one servlet to each client process and relies on the client OS for fair arbitration, coupling controller resources to host-managed process channels. Revocation recursively traverses a derivation DAG, so its completion depends on the delegation history rather than on a fence checked before each dependent access.  

\noindent\textbf{SemperOS.} 
SemperOS~\cite{hille2019semperos} distributes capability management across independent $\mu$-kernels, each responsible for a group of processing elements. Cross-kernel delegation and revocation require inter-kernel messages. Revocation uses a two-phase mark-and-sweep traversal. A group-spanning revocation takes roughly three times as long as a local one, and although marked capabilities are immediately blocked from further exchange, subtree deletion waits for remote replies. SemperOS limits concurrent revocation work to two threads per kernel, but completion still depends on inter-kernel communication and traversal.  

\noindent\textbf{White et al.} 
White et al.~\cite{white2025enabling} use per-node controllers, per-process capability tables, and versioned references. Invalidating a guard makes later use of guarded capabilities fail without distributed traversal, and delegation is attenuating with explicit withdrawal. The controllers, however, execute as trusted host software rather than as enforcement points isolated from the host kernel, so the design does not provide independent dual-boundary enforcement under an untrusted host. After a controller failure, stale authority is invalidated and applications must reacquire capabilities; live capability, delegation, and revocation state is not restored.

\noindent\textbf{FractOS.} FractOS~\cite{vilanova2022slashing} provides distributed capabilities for memory and service objects in heterogeneous disaggregated systems. Trusted Controllers run as user-level Linux processes on CPUs or SmartNICs, and delegation may span Controllers. FractOS makes revocation effective at the owning Controller by invalidating the referenced object, with separately revocable objects organized in a revocation tree providing selectivity. Capabilities referring to invalidated objects are cleaned up later, so denial does not wait for distributed cleanup. Controller failure revokes the affected Processes' capabilities, and reboot counters detect stale capabilities rather than restore live authority. FractOS does not independently validate every inter-node resource request before fabric admission and again before resource access, so it does not provide dual-boundary enforcement under an untrusted host.

\noindent\textbf{Comparison Summary.} 
Consider a delegated remote-memory capability that is revoked while the host OS is compromised. A resource-boundary-only design cannot establish that a request was validated against the authority of its issuing process because request-origin information comes from untrusted host software. A compute-boundary-only design cannot ensure denial at the resource boundary after revocation if its local state is stale. Sound enforcement therefore requires process-granularity validation before fabric admission and an authoritative access decision before resource access.  

CEP enforces access in memory-side hardware, but trusts compute-side channel assignment and invalidates descendants through derivation-graph traversal. SemperOS distributes authority across kernels, but revocation depends on inter-kernel coordination and traversal. White et al.\ provide version-based invalidation, but retain trusted host software on the enforcement path and do not restore live authority after failure. FractOS provides distributed delegation and selective owner-side revocation without waiting for distributed cleanup, but its trusted software Controllers do not independently validate every inter-node resource request at both the compute and resource boundaries. None combines independent validation at both boundaries under the untrusted-host model with resource-boundary denial that does not wait for propagation to remote compute controllers. Adding a second checkpoint alone is insufficient; the controllers must validate linked authority from local state without coordinating on the access path. \SADRA{} realizes this through the two-stage enforcement invariant (\autoref{inv:two-stage}) and the capability chain $c_p \mapsto c_c \mapsto c_r$ (\S\ref{subsec:our-cap-model}). 

This design places trusted controller logic at both boundaries. It does not reduce the TCB; instead, it moves enforcement outside host software control. Once installed, the resource-controller fence denies dependent requests at the resource boundary, while compute-controller containment and subtree cleanup proceed separately.


\subfile{figure_systems}

\section{System Design}
\label{sec:sys-design}

\SADRA{} realizes the two-stage enforcement invariant (\S\ref{sec:sec_properties_soundness}) through hardware controllers at both boundaries: each inter-node request is validated against local authority state before fabric admission and before data access, with no synchronous coordination between controllers or with host software.

The design maps to the four security properties of \S\ref{sec:security:props}. The capability model (\S\ref{subsec:our-cap-model}) supports capability safety by making every valid capability depend on a resource capability at the resource controller, and authority safety by requiring each delegated capability to encode a subset of the delegator's rights. The controller architecture (\S\ref{subsec:dis-acc-controllers}) supports strong isolation by denying other processes access to exclusively allocated regions. Revocation soundness comes from the resource-side fence: a fixed-latency standing deny over the revoked subtree, independent of delegation depth, while compute-side containment and structural cleanup proceed separately and never delay data protection. The following subsections describe each component.

\subsection{Disaggregated System Model}
\label{subsec:abstract_model}

In a resource-disaggregated system, compute nodes host processes that issue requests, while resource nodes manage and serve hardware resources such as memory or accelerators (\Cref{fig:combined_system}(\subref{fig:rda_system})). Each protected inter-node request crosses two enforcement boundaries: the compute-side boundary before fabric admission and the resource-side boundary before access to the managed resource.

Each node attaches to the fabric through an enforcement point that mediates protected requests crossing its boundary. We instantiate this point as a SmartNIC,  but the model requires only boundary-resident logic whose security-critical state is isolated from host-software control. The host may interact with the enforcement point through defined interfaces, but cannot modify its authority state or bypass its datapath checks. In multi-tenant deployments, tenants may share compute and resource nodes under a hypervisor or container runtime; the authority model applies independently of that partitioning.

We formalize this structure with an abstract model $\mathcal{M} = (\mathbb{N}_r, \mathbb{N}_c, \mathbb{R}, \mathbb{P}, F)$, where $\mathbb{N}_r$ is the set of resource nodes, $\mathbb{N}_c$ is the set of compute nodes, $\mathbb{R}$ is the set of managed resources, $\mathbb{P}$ is the set of processes, and $F$ represents the inter-node fabric. Each resource node $N^j_r \in \mathbb{N}_r$ manages a subset of resources $R_j \subseteq \mathbb{R}$, while each compute node $N^k_c \in \mathbb{N}_c$ hosts a subset of processes $P_k \subseteq \mathbb{P}$. We use this model throughout the paper; its formalization appears in \Cref{app:formal-model}.

\subsection{SADRA Capability Model}
\label{subsec:our-cap-model}

\SADRA{} uses two types of hardware controllers. Compute controllers ($CC_k \in \mathbb{CC}$) reside on compute nodes and enforce authorization before fabric admission, while resource controllers ($RC_j \in \mathbb{RC}$) reside on resource nodes, maintain authoritative resource capability state, and enforce authorization before resource access. \S\ref{subsec:dis-acc-controllers} describes their architectural organization and interaction. The capability model defines the authority structures these controllers maintain and the relationships among them.       

\noindent\textbf{Capability Taxonomy.} \SADRA{} classifies capabilities by residency and scope into \textit{resource} capabilities ($\mathbb{C}_r$), \textit{compute} capabilities ($\mathbb{C}_c$), and \textit{process} capabilities ($\mathbb{C}_p$) (\Cref{fig:caps}). Resource and compute capabilities reside in controller-managed state that host software cannot directly modify. Process capabilities are protected tokens issued to applications and bound by the compute controller to individual processes; fabricated or altered tokens cannot authorize requests.       

Each $c_p$ binds authority to a process, preventing transfer between co-located processes. Each $c_r$ encodes precise access bounds, such as memory offsets or resource slices, supporting isolation below page or segment boundaries.

\noindent\textbf{Capability Chain Invariant.}
The three capability types form a linear chain of authority:
$
c_p \mapsto c_c \mapsto c_r.
$
\SADRA{} authorizes a request only when its presented capability validates against the applicable link of a valid chain rooted at a resource capability managed by a resource controller. The compute controller validates the $c_p \mapsto c_c$ link locally, while the resource controller validates the $c_c \mapsto c_r$ link against authoritative resource-side state. Each controller validates its portion independently.

Invalidating a link invalidates authority derived through that link. For resource-capability revocation, the resource controller first installs a fence that denies accesses through the revoked subtree at the authoritative resource boundary; compute-side containment and structural materialization then proceed separately. Delegation can only narrow authority: every derived capability encodes a subset of its parent's rights over a subset of its parent's resource region.

\subfile{figure_sadra_caps}
\subfile{figure_example_trees}

\subsubsection{Indicator Flags and Delegation Tracking}
\label{subsec:indicator-flag}

\SADRA{} tracks delegation relationships through a one-bit indicator flag in each compute and process capability descriptor. An unset flag denotes usable authority; a set flag denotes a revocation handle rather than a usable right. The flag preserves delegator-side revocation references for both local process-to-compute links $c_p \mapsto c_c$ and inter-node compute-to-resource links $c_c \mapsto c_r$.

When a process delegates authority, \SADRA{} creates the capability issued to the recipient and a separate capability retained for revocation control by the delegator. The retained capability has its indicator flag set and therefore acts as a revocation handle. The delegator's original usable capability remains unchanged. Hardware rejects attempts to use the handle for reads, writes, or further delegation. The handle carries no operational authority; it retains the linkage needed to identify the delegated authority and serves only as the reference through which the delegator initiates revocation.

\subsubsection{Capability Trees}
\label{subsec:capability-tree}

\SADRA{} organizes capabilities into rooted trees maintained by compute  and resource controllers. These trees represent delegation and revocation dependencies that controllers consult when validating protected requests.

Each resource controller roots its tree at a resource capability representing authority over a resource region. Before data access, it validates that the presented compute capability in a request maps to a resource-side capability under this root and that no active fence covers its delegation path. Each compute controller maintains a tree of compute capabilities derived from resource capabilities and tracks local delegation among processes. Process capabilities derive from compute capabilities and represent the authority exercised by applications.

The tree structure captures dependency relationships among delegated capabilities. For resource-capability revocation, the resource controller installs a fence at the authoritative resource boundary that denies accesses through the revoked subtree. This is a fixed-latency authorization check and takes effect independently of delegation depth. Structural cleanup of the affected subtree then proceeds separately and does not delay data protection. Because the fence already denies every access through the revoked subtree, not-yet-reclaimed capabilities cannot authorize access until cleanup removes them, so deferring removal does not affect authorization decisions.

\subsubsection{Capability Operations}
\label{subsec:cap-ops}

\SADRA{} manages resource access through three operations: allocation, delegation, and revocation. \Cref{fig:cap-trees1} illustrates the capability tree state at each step. We describe each operation for a representative process $p \in P_k$ accessing resource $r \in R_j$.      

\noindent\textbf{Resource Allocation.} 
$RC_j$ creates $c_r^1$, encodes the granted rights, and inserts it into its resource tree. It derives $c_c^1$ ($c_c^1 \mapsto c_r^1$) and sends it to $CC_k$ (\Cref{fig:cap-trees1}\ref{fig:tree-2}), which inserts it, derives $c_p^1$ ($c_p^1 \mapsto c_c^1$), and returns it to $p$ (\Cref{fig:cap-trees1}\ref{fig:tree-3}).

\noindent\textbf{Intra-Node Delegation.} 
$CC_k$ derives $c_c^2$ from $c_c^1$ with a subset of its rights over a subset of its resource region, derives $c_p^2$ for recipient $p'$, and issues a tagged capability to $p$ as a revocation handle (\Cref{fig:cap-trees1}\ref{fig:tree-4}). The compute controller performs this operation without contacting the resource controller.

\noindent\textbf{Intra-Node Revocation.} 
$p$ invokes its revocation handle. $CC_k$ invalidates $c_c^2$ in its tree, breaking the chain for all process capabilities that depend on it (\Cref{fig:cap-trees1}\ref{fig:tree-5}). Subsequent requests from $p'$ fail at the compute controller. This revocation is local to $CC_k$ and does not involve the resource controller.

\noindent\textbf{Inter-Node Delegation.} 
$RC_j$ creates a new resource capability granting a subset of the delegating capability's rights over a subset of its resource region. It then derives a compute capability for the recipient node and a compute revocation handle for the delegating node, and sends each to the corresponding compute controller (\Cref{fig:cap-trees1}\ref{fig:tree-6}). The recipient compute controller derives and issues a process capability to the recipient process, while the delegating compute controller creates a corresponding process revocation handle for the delegating process (\Cref{fig:cap-trees1}\ref{fig:tree-7}).

\noindent\textbf{Inter-Node Revocation.}
$p$ invokes its revocation handle. $CC_k$ resolves it to the corresponding compute revocation handle, removes that handle from its tree, and forwards the revocation request to $RC_j$
(\Cref{fig:cap-trees1}\ref{fig:tree-8}). $RC_j$ then installs a fence at the resource boundary, denying accesses through the revoked subtree (\Cref{fig:cap-trees1}\ref{fig:tree-9}).  Once installed, the resource controller performs a fixed-latency fence check for each request and denies access without waiting for descendant-state propagation.  A request may still pass compute-side validation if $CC_{k'}$ retains stale state, but fails at the resource boundary because the fence covers the subtree. On learning of the revocation or observing the denial, $CC_{k'}$ installs a local fence over the stale compute capability and derived subtree, then prunes that state asynchronously.

The formal model underlying these operations appears in \Cref{app:formal-model}.

\subsection{Persistent Capability State and Recovery} \label{subsec:persistent-capability-state}  \SADRA{}'s controller-maintained capability state provides the basis for persistent recovery. In conventional systems, authority state is often maintained by host software, requiring recovery to reconstruct authorization relationships from software-managed state. In contrast, \SADRA{} records delegation and revocation state explicitly in the capability trees maintained by the compute and resource controllers. Recovery therefore restores the authority state enforced by the trusted controllers rather than reconstructing it from an untrusted host.

Persistent controller state includes capability-tree relationships, controller-local authority identifiers and their allocation counters, delegation and revocation-control state, and active fences. The controller-local identifier counters are monotonically increasing, non-wrapping persistent state. A controller durably advances the corresponding counter before authority carrying a newly allocated $\mathit{cid}$ or $\mathit{rid}$ becomes usable. Recovery therefore cannot cause a reclaimed identifier to be reassigned to different authority.  

Restoring the capability trees and fence state preserves the authority relationships required by two-stage enforcement. Process capabilities remain associated with compute-side authority through $c_p\mapsto c_c$, while compute-side authority remains associated with resource-side authority through $c_c\mapsto c_r$. A process token retained by untrusted software does not regain authority merely because it survives a failure: it can authorize a request only if it resolves to active controller-resident authority and the request passes both enforcement stages.

\SADRA{} orders capability creation so that newly created authority becomes usable only after the controller state required to establish its authority chain has been durably committed. Incomplete capability creation is therefore not made usable after recovery. This applies to initial allocation as well as intra-node and inter-node delegation.

Revocation requires stronger recovery ordering because cleanup may span multiple controllers. A fence is durably established before capability state protected by that fence can be reclaimed. Consequently, a failure during local cleanup restores the fence before reclamation continues, so authority disabled before the failure does not become active again.

For a fenced compute subtree that contains inter-node revocation handles, recovery must also preserve outstanding resource-side revocation obligations. A compute revocation handle is retained until the authoritative resource controller confirms that the corresponding resource authority has been fenced or already reclaimed. Thus, a crash after installation of a local compute fence but before completion of remote revocation cannot erase the information needed to finish that revocation. After recovery, the compute controller replays every unresolved resource-side revocation before retiring the corresponding handle or reclaiming the fenced subtree.

For resource-rooted revocation, the resource-controller fence is the authoritative security-effect point for resource access. The resource controller durably commits this fence before acknowledging completion to the source compute controller. If either controller fails after the resource fence is installed, recovery restores the fence before affected resource authority can become active. If the source fails before receiving the acknowledgment, it may safely resend the retained revocation request; resource-side revocation is idempotent when the target is already fenced or has already been reclaimed.

After recovery, controllers restore committed capability and fence state, discard incomplete authority-creation transitions, and replay unresolved revocation transitions. The recovery invariant is therefore that recovery neither creates authority that was not previously committed nor reactivates authority already disabled by a committed fence. Pending inter-node revocations remain represented by their retained compute handles until the authoritative resource-side state has been secured.   

\subsection{SADRA Architectural Model}
\label{subsec:dis-acc-controllers}

\SADRA{} places hardware-resident controllers at the compute-side boundary before fabric admission and the resource-side boundary before data access to realize the two-stage enforcement invariant (\autoref{inv:two-stage}), as shown in~\Cref{fig:sadra_system}. A \textit{compute controller} resides on each compute node, while a \textit{resource controller} resides on each resource node. Both reside in SmartNIC hardware whose enforcement state is isolated from host-software control. Every protected inter-node request is validated at both points on the data path.  

\noindent\textbf{Compute Controller.} The compute controller performs two functions on the hardware datapath. First, \textbf{validation before fabric admission}: it checks that the requesting process holds a capability locally valid with respect to its delegation and revocation state, and drops requests that fail at the source. Second, \textbf{intra-node delegation}: it transfers attenuated authority between co-located processes ($p, p' \in P_k$) without remote-node involvement or network coordination.       

\noindent\textbf{Resource Controller.} The resource controller is the authoritative enforcement point for its local resource subset $R_j$ and maintains authoritative resource capability state. It handles resource allocation across the fabric and manages inter-node delegation and revocation when participating processes reside on different compute nodes. Before resource access, it independently validates every arriving request against the corresponding resource-side capability state and active fences. The resource controller rejects requests that fail.      

When the resource controller revokes a resource capability, it installs a fence that denies every access through the revoked subtree at the resource boundary. This denial takes effect independently of whether corresponding compute controllers have fenced or pruned stale local state.      

Each controller makes its authorization decision locally. For resource-capability revocation, the resource-side fence is the authoritative effect; structural cleanup and compute-side containment proceed separately. This organization removes coordination bottlenecks on the access path while preserving authoritative resource-side protection.


\section{Implementation} 
\label{sec:implementation}  

\SADRA{} implements the compute and resource controllers as FPGA logic in the data path (\Cref{fig:impl}), placing authorization checks on the request path at both interconnect boundaries. This implementation realizes the two-stage enforcement invariant (\autoref{inv:two-stage}) in hardware. Each controller validates every inter-node request at its boundary, while controller-local enforcement state remains isolated from host software.  

\noindent\textbf{Controller Architecture.} Each compute node hosts a compute controller ($CC$), while each resource node
hosts a resource controller ($RC$). Both controllers are implemented on Xilinx Alveo U55C FPGAs and synthesized as timing-closed designs at a 250\,MHz operating frequency. The implementation comprises approximately 11\,K lines of VHDL for the resource controller and 9\,K lines of VHDL for the compute controller.  The compute controller integrates with the Xilinx QDMA software stack, whose PF/VF kernel drivers configure the SR-IOV virtual functions (VFs) shown in \Cref{fig:impl}.  These drivers take no part in capability validation or security-critical state updates.  

\noindent\textbf{Capability Representation.} Each capability is a 256-bit token. Its identity comprises a controller-local slot identifier, a generation number, and a MAC tag. The remaining fields encode the permission mask, resource extent, capability type, and indicator flag. The controller verifies the tag before accepting the capability, so a token fabricated or modified by host software cannot authorize a request.

\noindent\textbf{Datapath Integration.} 
The implementation uses QDMA VFs to provide tenants with a kernel-bypass path to the compute controller's enforcement pipeline before the networking stack. The compute controller validates the process capability in a request by checking its MAC tag, permissions, resource range, active local fences, and binding to the hardware-observed QDMA queue identifier. These checks execute in bounded pipeline stages without host-dependent execution delays.

\noindent\textbf{Capability State Management.} Each controller stores its capability trees and revocation state in on-device SmartNIC memory and consults this state on every protected request.  The active fence check is a fixed-latency operation within the authorization pipeline. Subtree cleanup proceeds separately and does not delay resource protection.  The controllers perform all security-critical capability, delegation, and revocation updates, preserving the TCB and isolation assumptions of \S\ref{sec:threat-model}.  

\noindent\textbf{Networking Stack Integration.}  A customized EasyNet network stack~\cite{he2021EasyNet} provides TCP/IP processing on the FPGAs, while Mellanox ConnectX-5 NICs provide the 100\,Gbps physical links through a non-blocking optical switch.

\section{Evaluation}
\label{sec:evaluation}

We evaluate \SADRA{} on FPGA hardware, focusing on datapath performance,  capability cleanup,  and implementation-level security validation. We measure the cost of enforcing the two-stage invariant (\autoref{inv:two-stage}) at both boundaries and verify through adversarial scenarios that invalid and revoked authority is rejected at the expected enforcement point.  

We organize the quantitative evaluation around two questions. \textbf{RQ1 (Datapath performance):} What throughput can \SADRA{} sustain, and what latency and throughput overhead does enforcement introduce? \textbf{RQ2 (Comparison):} How does \SADRA{} compare with prior capability systems in enforcement architecture, revocation, failure recovery, and performance?

\subfile{figure_impl}

In addition, we evaluate representative scenarios involving forged, over-privileged, stale, and revoked authority. For each scenario, we record the boundary that rejects the request and whether any unauthorized request reaches the protected resource (\S\ref{sec:security-validation}).

\noindent\textbf{Experimental Setup.}
We evaluate \SADRA{} on three nodes: two compute nodes and one memory node. Each compute node runs up to 32 tenants in Podman containers, with one SR-IOV virtual function assigned to each tenant. Experiments using both compute nodes therefore include up to 64 concurrent tenants.  The baseline retains the QDMA, networking, transport, and resource-access paths but without \SADRA{} capability processing. We compare each \SADRA{} configuration with its corresponding measured baseline to determine the overall runtime overhead of enforcement. We repeat each configuration three times on each compute node and all measurements cover the complete end-to-end datapath.

\subfile{figure_throughput}
\subfile{figure_latency}

\subsection{Datapath Performance (RQ1)} 
\label{subsec:rq1}  

We measure absolute throughput, throughput relative to the baseline defined above, and per-request latency as the window size, tenant count, and payload size vary.  

The application-agnostic microbenchmark generates controlled remote-memory traffic over the parameter ranges defined in the setup, isolating the datapath cost of \SADRA{}'s two-stage validation and linked capability chain from application-specific behavior. The multi-tenant embedding-lookup workload is representative of deep-learning recommendation systems~\cite{naumov2019dlrm}. It exercises the enforcement datapath under application-driven access patterns, including concurrent access to a shared embedding table in disaggregated memory, private per-tenant write buffers, and skewed memory accesses.  

In both workloads, \SADRA{} validates requests at the compute and resource controllers.

\subsubsection{Application-Agnostic Microbenchmarks}
\label{subsec:appagnostic}

We use application-agnostic microbenchmarks to measure the end-to-end datapath cost of \SADRA{} independently of application computation. The baseline uses the same QDMA, network, transport, and remote-memory datapaths without capability enforcement. In the \SADRA{} configuration, every protected request additionally passes through the compute and resource controllers. The comparison therefore measures the performance cost of \SADRA{} relative to the corresponding non-enforcing remote-memory datapath.        

\noindent\textbf{Workload and methodology.} The microbenchmark uses a closed-loop, windowed request pattern. In each iteration, every tenant submits up to $W$ writes and waits for all responses, then submits up to $W$ reads and again waits for all responses. Each tenant transfers 1\,GiB of write data and 1\,GiB of read data over the complete run. We evaluate $W\in\{1,2,4,8\}$ and payload sizes of 512\,B, 1\,KiB, 2\,KiB, and 4\,KiB. The two compute nodes generate traffic concurrently, with 2, 4, 8, 16, 32, or 64 total tenants distributed equally between them. Each tenant accesses a distinct region of the 16\,GiB remote-memory space. $W=1$ corresponds to synchronous request--response execution, while larger windows increase the number of outstanding requests and drive the datapath toward saturation.

We evaluated each configuration in three synchronized rounds on both compute nodes. For throughput, we sum the measurements from the two nodes within each round and report the mean aggregate throughput across rounds.  Throughput overhead is computed relative to the paired baseline; a negative value therefore means only that \SADRA{} measured slightly faster in that configuration. For latency measurements, request construction, metadata preparation, response-mailbox clearing, and required cache preparation occur before the timer starts. RTT is measured from request submission until the corresponding completion response is received. We report the difference in mean RTT between \SADRA{} and the paired baseline, with paired 95\% confidence intervals.       

\noindent\textbf{Throughput.} \Cref{fig:aae_throughput} shows that \SADRA{} closely tracks the baseline throughout the evaluated operating range. Of the 96 configurations, 94 differ from the baseline by at most 1\% in either direction. Throughput overhead ranges from $-0.78\%$ to $1.56\%$, and 32 configurations have a negative value. The presence of both small positive and negative differences indicates that the enforcement cost is comparable to run-to-run variation; the negative values should not be interpreted as an inherent performance benefit from capability checking.         

Throughput increases with payload size, tenant count, and window size until the underlying remote-memory datapath reaches saturation. At the highest-load configuration, with 64 tenants, 4\,KiB requests, and $W=8$, \SADRA{} reaches 89.542\,Gbit/s, with a throughput overhead $<0.01\%$,  compared with 89.545\,Gbit/s for the baseline. Both are effectively at the measured 90--91\,Gbit/s throughput ceiling of the remote-memory datapath. \SADRA{} therefore reaches the platform's saturation region without capability enforcement becoming a measurable throughput bottleneck.

\noindent\textbf{Request latency.} \Cref{fig:aae_latency} shows the corresponding request RTTs. Absolute RTT increases with larger windows and tenant counts as queueing increases, while the difference between \SADRA{} and the paired baseline remains small across the evaluated operating range. The increase in mean RTT remains below 1\,\textmu s in every configuration and is typically only a few tenths of a microsecond. The paired 95\% confidence intervals further show that many of the measured differences are comparable to run-to-run variation. These results indicate that the additional validation performed by \SADRA{} adds little to end-to-end request latency even as the underlying datapath approaches saturation.

The application-agnostic measurements show that \SADRA{} drives the remote-memory datapath to its throughput ceiling while adding only a small latency cost. These microbenchmarks isolate enforcement behavior under controlled request patterns; the application-driven evaluation examines the same enforcement datapath under an application workload.

\subsubsection{Multi-Tenant ML Workload}
\label{subsec:mlworkload}

We complement the application-agnostic microbenchmarks with a multi-tenant recommendation-inference workload based on Meta's DLRM architecture~\cite{naumov2019dlrm}. The purpose of this experiment is to evaluate \SADRA{} under an application-driven access pattern rather than to compare remote and local memory. The baseline therefore executes the same remote-memory workload without capability enforcement.

\noindent\textbf{Workload and methodology.} The model contains 13 dense features, a $13$--$64$--$32$ bottom MLP, 26 categorical features with 26 embedding tables, dot-product feature interaction, and a $383$--$256$--$64$--$1$ top MLP. Dense computation and input data remain local, while embedding-table accesses and result writes use remote memory. We use batch size 32 and 100,000 test records per tenant.

Each inference generates 26 distinct remote embedding reads and one result write. Each embedding lookup is issued as a separate 512\,B read request, the minimum payload granularity supported by the remote-memory datapath; the required 128-byte embedding vector is extracted locally from the returned block. A tenant writes the inference result with one separate 4\,KiB write request. Thus, every inference generates 27 protected remote requests and transfers 17,408\,B. The window size $W\in\{1,2,4,8\}$ controls only the number of outstanding remote operations and does not coalesce requests.

\subfile{figure_ml}

We evaluate 2, 4, 8, and 16 total tenants across the two compute nodes; the application-agnostic microbenchmarks separately evaluate datapath scaling up to 64 tenants. We measure each ML configuration in three synchronized paired baseline--\SADRA{} rounds. A two-node round serves as the statistical unit, with inference throughput summed across the two nodes and remote-operation latency weighted by the number of completed remote read and write operations. Because the model uses deterministic parameters rather than a trained checkpoint, we evaluate systems behavior and output equivalence rather than model accuracy.

\noindent\textbf{Application correctness.} Across all 48 paired baseline--\SADRA{} rounds, the prediction hashes match. Thus, every evaluated \SADRA{} execution produces bit-identical prediction outputs to its corresponding non-enforcing baseline. This provides application-level evidence that routing the embedding accesses and result writes through the capability-enforcement path preserves the computation's outputs.

\noindent\textbf{Inference throughput.} \Cref{fig:ml_performance} shows the signed inference-throughput overhead relative to the baseline. Across the 16 evaluated configurations, 12 are within 1\% of the baseline and all are within 2\%. The largest positive overhead is 1.76\%. Ten configurations have negative signed overhead; as in the application-agnostic measurements, these small differences should be interpreted as measurement, scheduling, and queueing variation rather than as an inherent performance benefit from enforcement.

The highest inference throughput occurs with 16 tenants and $W=2$. The baseline reaches 52,824 inferences/s, while \SADRA{} reaches 52,678 inferences/s, a difference of 0.276\%. Thus, at the workload's highest-throughput operating point, capability enforcement has little effect on the throughput.

\noindent\textbf{Authorization rate and remote-operation latency.} Although this workload does not saturate network bandwidth, it generates a substantial rate of protected operations because every inference produces 27 separate remote requests. At the peak \SADRA{} operating point, the system processes approximately 1.42 million protected remote requests/s, comprising about 1.37 million embedding read operations per second and 52.7 thousand result write operations per second, while transferring 7.34\,Gbit/s. The corresponding baseline traffic is 7.36\,Gbit/s. The middle panel of \Cref{fig:ml_performance} reports the protected-request rate across the evaluated configurations.

The bottom panel reports \SADRA{}'s mean remote-operation RTT. Increasing the window primarily increases the number of outstanding requests and therefore queueing latency. Across all configurations, the largest positive increase in mean remote-operation RTT relative to the paired baseline is 0.638\,\textmu s. The small application-level throughput differences therefore coincide with a similarly small increase on the remote-access path itself, rather than being hidden solely by local ML computation.

\noindent\textbf{Effect of request concurrency.} Increasing $W$ beyond 2 does not improve inference throughput. At 16 tenants, $W=2$ reaches approximately 52.8 thousand inferences/s, whereas $W=4$ and $W=8$ fall to approximately 48 thousand inferences/s while remote-operation RTT increases. The largest observed throughput overhead, 1.76\%, occurs at 16 tenants and $W=4$, after the workload has already passed its highest-throughput operating point. We therefore report this value as the maximum observed overhead without attributing it to a specific enforcement or queueing mechanism.

The ML experiment exercises a different operating regime from the application-agnostic microbenchmarks. The latter drives the remote-memory datapath to its approximately 90\,Gbit/s throughput ceiling, whereas the ML workload generates smaller requests at up to approximately 1.42 million protected operations per second. Across this application-driven workload, \SADRA{} preserves prediction outputs, remains within 1.8\% of baseline inference throughput in every evaluated configuration, and adds at most 0.638\,\textmu s to mean remote-operation RTT.

\subsection{Comparison with Prior Capability Systems (RQ2)}
\label{subsec:comparison}

We compare \SADRA{} with CEP~\cite{azriel2019memory}, SemperOS~\cite{hille2019semperos}, White et al.~\cite{white2025enabling}, and FractOS~\cite{vilanova2022slashing}. These systems address different points in the capability design space: CEP targets rack-scale disaggregated memory, SemperOS provides distributed capability management across microkernels, White et al. explore version-based capability invalidation, and FractOS provides distributed delegation and owner-side revocation for heterogeneous disaggregated systems. We compare them along the properties most relevant to resource-disaggregated enforcement including enforcement architecture, revocation, failure recovery, and performance.

\noindent\textbf{Enforcement architecture.} The primary distinction among these systems is where they establish authorization and where they make the authoritative decision. CEP validates capabilities at the memory side, but relies on compute-side software to bind requests to processes and preserve channel integrity. SemperOS distributes capability management across cooperating $\mu$-kernels and enforces authorization through a software-managed capability hierarchy. White et al. use per-node controllers and versioned references, but their controllers execute as trusted software components. FractOS moves capability management into distributed Controllers for heterogeneous disaggregation, but its Controllers remain software components rather than hardware enforcement.

\subfile{table_revocation}

\SADRA{} instead places hardware controllers isolated from host software at both ends of the remote access path. The compute controller validates process-granularity authority before fabric admission, while the resource controller performs the authoritative access decision before resource access. These checks use linked local capability state rather than coordination on the access path.

\noindent\textbf{Revocation.} Prior systems provide different mechanisms for invalidating delegated authority. CEP and SemperOS rely on traversal-based revocation, where completion depends on processing the affected delegation structure and, for distributed cases, coordination across enforcement domains. White et al. avoid distributed traversal through versioned references, but stale authority after failure is invalidated rather than restored. FractOS provides immediate owner-side revocation through separately invalidatable delegation structures, with cleanup of derived state proceeding separately.

\SADRA{} separates the security effect of revocation from state reclamation. For revocation at the resource boundary, the resource controller installs a fence on the authoritative resource capability, and every dependent request is denied at the resource boundary regardless of whether remote compute controllers have observed the revocation. Similarly, compute controllers use fences when revoking authority within their local capability trees, ensuring that dependent accesses fail before cleanup completes. Cleanup then proceeds separately by traversing the affected capability state. As shown in \Cref{tab:revoc_result}, for a fixed capability count, cleanup latency is unchanged as the number of affected subtrees varies from one to eight. A 128-capability cleanup completes in 528\,ns. The fence, rather than cleanup traversal, determines authorization.

\noindent\textbf{Failure recovery.} Persistence across failures provides another distinction among capability systems. White et al.\ invalidate stale authority after controller failure and require applications to reacquire capabilities rather than resuming from restored authority state. FractOS similarly treats controller failure as capability invalidation rather than restoration of live delegated authority. \SADRA{} instead defines recovery semantics as part of its capability-chain architecture, preserving persistent capability relationships and revocation state without reactivating stale authority. The complete persistence comparison is provided in the systematization (\S\ref{sec:systematization}).

\noindent\textbf{Performance.} CEP shows that capability validation can be integrated into the memory-side datapath, while SemperOS and FractOS demonstrate distributed capability management in software. Because these systems target different platforms and report different metrics, we compare design properties rather than directly comparing measured performance.

\SADRA{} combines hardware enforcement with distributed capability management. The application-agnostic evaluation shows that the two validation stages preserve throughput up to the measured remote-memory datapath ceiling of 89.5\,Gbit/s. The ML evaluation shows at most 1.8\% inference-throughput degradation relative to the baseline across all evaluated configurations, where twelve of the sixteen configurations remain within 1\%, and all paired runs produce bit-identical outputs.

\subsection{Prototype Security Validation}
\label{sec:security-validation}

\Cref{sec:sec_properties_soundness} establishes capability safety, authority safety, revocation soundness, and strong isolation for the \SADRA{} architectural model and verifies the formalization using SPIN. The deployed resource and compute controllers contain approximately 11\,K and 9\,K lines of VHDL, respectively, which are outside the scope of the formal analysis. Thus,  we complement the formal results with implementation-level tests that exercise a corresponding adversarial scenario for each property and record both the enforcement point that rejects the request and whether the request reaches the protected resource (\Cref{tab:security-validation}).

The scenarios map directly to the four properties. Forged capabilities test capability safety, rights amplification tests authority safety, revoked authority tests revocation soundness, and cross-tenant requests test strong isolation. Forged process capabilities and rights amplification are rejected by the compute controller before entering the fabric. For revocation, authority confined to one compute node is revoked and rejected by its compute controller, while authority delegated across compute nodes is fenced by the authoritative resource controller and rejected before resource access. Cross-tenant requests without valid authority are likewise rejected before accessing the protected resource.

In every trial reported in \Cref{tab:security-validation}, \SADRA{} rejected the unauthorized request at the expected enforcement point, and no unauthorized request reached the protected resource. These tests do not replace the formal analysis, but provide implementation-level evidence that the FPGA prototype behaves consistently with the architectural model for the evaluated attack scenarios.

\begin{table}[t]
\centering
\caption{Implementation-level validation of the four security properties.
All attack attempts were rejected before resource access.}
\label{tab:security-validation}
\small
\begin{tabular}{@{}lll@{}}
\toprule
Attack & Property & Boundary \\
\midrule
Forgery     & Capability safety    & Compute \\
Amplification & Authority safety   & Compute \\
Revocation  & Revocation soundness & Resource \& Compute \\
Cross-tenant & Strong isolation    & Resource \& Compute  \\
\bottomrule
\end{tabular}
\end{table}

\section{Conclusion} 
\label{sec:conclusion}     

Resource disaggregation separates computation from memory and accelerator resources, but it also removes the local protection boundary traditionally provided by the host operating system. This separation requires authorization mechanisms that can enforce access control across disaggregated trust boundaries without placing coordination or trusted host software on the resource-access path.

This paper presented \SADRA{}, a distributed capability-based access-control architecture for resource-disaggregated systems. \SADRA{} introduces a two-stage enforcement invariant in which isolated SmartNIC hardware validates authority at both the compute and resource boundaries. A linked chain of process, compute, and resource capabilities enables fine-grained delegation while allowing each enforcement point to validate requests from local state. For revocation, \SADRA{} separates the security effect from state reclamation: a resource-rooted fence immediately prevents dependent authority from reaching the resource, while capability-state cleanup proceeds independently. The same capability-chain representation also defines recovery semantics that preserve authority relationships without allowing stale authority to be reactivated.

We formally establish capability safety, authority safety, revocation soundness, and strong isolation for the \SADRA{} model and evaluate the implementation through FPGA-based adversarial security experiments. The prototype evaluation demonstrates that these guarantees can be achieved without making authorization the performance bottleneck. The FPGA SmartNIC implementation sustains 89.5\,Gbit/s, with a throughput overhead < $0.01\%$,  effectively reaching the 90--91\,Gbit/s throughput ceiling of the evaluated remote-memory datapath. Under a Meta-DLRM workload, \SADRA{} maintains inference throughput within 1.8\% of the baseline while processing up to 1.42 million protected remote operations per second and producing bit-identical outputs across all paired runs. Revocation cleanup of a 128-capability subtree completes in 528\,ns, while denial is determined by the resource-side fence rather than by completion of cleanup traversal.

By combining distributed capability management with hardware-isolated authorization at the compute and resource boundaries, \SADRA{} demonstrates that resource-disaggregated systems can retain strong access-control guarantees without sacrificing the performance characteristics that motivate disaggregation.

\bibliographystyle{plain}
\bibliography{ref}

\appendix
\subfile{appendix}

\end{document}

%% file: packages.tex
\usepackage{filecontents}
\usepackage{transparent}
\usepackage{xcolor}
\usepackage{nicefrac}
\usepackage{siunitx}
\usepackage{array}
\usepackage{etoolbox}
\usepackage{amsthm}
\usepackage{amsfonts} 
\usepackage{amsmath}  
\usepackage{amssymb}  
\usepackage{thmtools}
\usepackage{tabularx}
\usepackage{booktabs}
\usepackage{pgfplots}
\usepgfplotslibrary{groupplots}
\pgfplotsset{compat=1.18}
\usepackage{tikz}
\usetikzlibrary{positioning, fit, backgrounds, arrows.meta, shapes.geometric,calc,decorations.pathreplacing}

\usepackage{subcaption} 
\usepackage{multirow}
 \usepackage{paralist} 
\usepackage{color, colortbl}
\usepackage{filecontents}
\usepackage{transparent}
\usepackage{xcolor}
\usepackage{nicefrac}
\usepackage{siunitx}
\usepackage{thmtools}
\usepackage{mathtools}
\usepackage{bm}
\usepackage{xargs} 
\usepackage{subfiles}
\usetikzlibrary{plotmarks}
\usetikzlibrary{shapes.misc, tikzmark}
\usepackage{semantic}

\usepackage[scr=boondox,  
            cal=esstix]   
           {mathalpha}
 
\usepackage{booktabs} 
\usepackage{wasysym}
\usepackage{float}
\usepackage{diagbox}
\usepackage{hyphenat}

\newcommand{\FULLY}{$\CIRCLE$}
\newcommand{\PART}{$\LEFTcircle$}
\newcommand{\NO}{$\Circle$}

\newcommand{\splitatcommas}[1]{%
  \begingroup
  \begingroup\lccode`~=`, \lowercase{\endgroup
    \edef~{\mathchar\the\mathcode`, \penalty0 \noexpand\hspace{0pt plus 1em}}%
  }\mathcode`,="8000 #1%
  \endgroup
}

\providecommand\longrightarrowRHD{\relbar\joinrel\relbar\joinrel\mathrel\RHD}
\providecommand\longrightarrowrhd{\relbar\joinrel\relbar\joinrel\mathrel\rhd}

\makeatletter
\providecommand*\xrightarrowRHD[2][]{\ext@arrow 0055{\arrowfill@\relbar\relbar\longrightarrowRHD}{#1}{#2}}
\providecommand*\xrightarrowrhd[2][]{\ext@arrow 0055{\arrowfill@\relbar\relbar\longrightarrowrhd}{#1}{#2}}
\makeatother

\mathchardef\mhyphen="2D
\newtheorem{definition}{Definition}

\newtheorem{lemma}{Lemma}

\newtheorem{remark}{Remark}
\newtheorem{invariant}{Invariant}

\newcommandx{\unsure}[2][1=]{\todo[linecolor=red,backgroundcolor=red!25,bordercolor=red,#1]{#2}}
\newcommandx{\change}[2][1=]{\todo[linecolor=blue,backgroundcolor=blue!25,bordercolor=blue,#1]{#2}}
\newcommandx{\info}[2][1=]{\todo[linecolor=OliveGreen,backgroundcolor=OliveGreen!25,bordercolor=OliveGreen,#1]{#2}}
\newcommandx{\improvement}[2][1=]{\todo[fancyline,linecolor=Plum,backgroundcolor=Plum!25,bordercolor=Plum,#1]{#2}}
\newcommandx{\thiswillnotshow}[2][1=]{\todo[disable,#1]{#2}}

\newcommandx{\rescont}{resource controller}
\newcommandx{\intcont}{intermediate controller}

\newcommandx{\pcap}{p-cap}
\newcommandx{\icap}{i-cap}
\newcommandx{\rcap}{r-cap}
\newcommandx{\indcap}{indicator capability}

\newcommandx{\pnode}{p-node}
\newcommandx{\rnode}{r-node}

\newcommandx{\rtree}{r-tree}
\newcommandx{\itree}{i-tree}

\newcommand{\RNum}[1]{\uppercase\expandafter{\romannumeral #1\relax}}

\definecolor{bgcolor}{RGB}{248,248,248}
\newcommandx{\ce}[1]{\begin{tcolorbox}[ ams align*, top=0mm, width=(\columnwidth),  center, enlarge top initially by=0mm, enlarge bottom finally by=0mm, left=1mm, colback=bgcolor, colframe=black!30 ]
#1
\end{tcolorbox}}

\newcommandx{\cefrac}[1]{\begin{tcolorbox}[ ams align*, width=(\columnwidth),  center, enlarge top initially by=0mm, enlarge bottom finally by=0mm, left=1mm, colback=bgcolor, colframe=black!30 ]
#1
\end{tcolorbox}}

\definecolor{pcapcolor}{RGB}{85, 107, 47}
\DeclareRobustCommand\pcapcircle{
\begin{tikzpicture}[inner sep=0pt,baseline=(base)]
\filldraw[fill=pcapcolor!40, draw=black] (0,0) circle (0.14cm);
\node (base) at (0,-.5ex) {};
\end{tikzpicture}     }

\definecolor{icapcolor}{RGB}{70, 130, 180}
\DeclareRobustCommand\icapcircle{
\begin{tikzpicture}[inner sep=0pt,baseline=(base)]
\filldraw[fill=icapcolor!65, draw=black] (0,0) circle (0.14cm);
\node (base) at (0,-.5ex) {};
\end{tikzpicture}     }

\definecolor{rcapcolor}{RGB}{233, 116, 81}
\DeclareRobustCommand\rcapcircle{
\begin{tikzpicture}[inner sep=0pt,baseline=(base)]
\filldraw[fill=rcapcolor!80, draw=black] (0,0) circle (0.14cm);
\node (base) at (0,-.5ex) {};
\end{tikzpicture}     }

\DeclareRobustCommand\solidcircle{
\begin{tikzpicture}[inner sep=0pt,baseline=(base)]
\filldraw[line width=.5pt, fill=white,draw=black] (0,0) circle (0.14cm);
\node (base) at (0,-.5ex) {};
\end{tikzpicture}     }

\DeclareRobustCommand\dashedcircle{
\begin{tikzpicture}[inner sep=0pt,baseline=(base)]
\filldraw[line width=.5pt, fill=white,draw=black,densely dashed] (0,0) circle (0.14cm);
\node (base) at (0,-.5ex) {};
\end{tikzpicture}     }

\definecolor{trcolor}{RGB}{207,0,0}

\DeclareRobustCommand\circledashedarrow{
\begin{tikzpicture}
\filldraw[fill=black, draw=black] (0ex,-.09) rectangle (.05,.09);
\draw[->,line width=0.4mm,black, -latex] (0.4ex,0) -- (5.2ex,0);
\end{tikzpicture}}

\DeclareRobustCommand\captreedarrow{
\begin{tikzpicture}
\filldraw[fill=white, draw=white] (0ex,0) circle (.45ex);
\draw[->,line width=0.4mm,black, -latex] (0.0ex,0) -- (5.0ex,0);
\end{tikzpicture}}

\DeclareRobustCommand\directmemory{
\begin{tikzpicture}
\filldraw[fill=white, draw=white] (0ex,0) circle (.45ex);
\draw[->,line width=0.4mm,black, -latex] (0.0ex,0) -- (5.0ex,0);
\end{tikzpicture}}

\DeclareRobustCommand\indirectmemory{
\begin{tikzpicture}
\filldraw[fill=white, draw=white] (0ex,0) circle (.45ex);
\draw[->,line width=0.4mm,black,dashed, -latex] (0.0ex,0) -- (5.0ex,0);
\end{tikzpicture}}
\definecolor{backc}{RGB}{242,242,242}

\usetikzlibrary{decorations.pathreplacing}

\def\niterate{128}
\def\rolldice{
    \pgfmathsetmacro\a{(1+rnd)/1}
    \pgfmathsetmacro\b{5+5*rnd}
    \pgfmathsetmacro\c{1+rnd}
    \pgfmathsetmacro\d{rnd*3}
    \pgfmathsetmacro\dark{rnd*50+50}
}
\tikzset{
    put dots/.style={
        /utils/exec=\rolldice,
        line width=2pt,
        dash pattern=on \b off \c,
        dash phase=\c*rnd,
        shift={(rnd*360:\d pt)},
        line cap=round,
        red!\dark,
        opacity=.8
    },
    chalk/.style={
        decorate,
        decoration={
            show path construction,
            lineto code={
                \foreach\i in{1,...,\niterate}{
                    \draw[put dots]
                        (\tikzinputsegmentfirst)--(\tikzinputsegmentlast);
                }
            },
            curveto code={
                \foreach\i in{1,...,\niterate}{
                    \draw[put dots]
                        (\tikzinputsegmentfirst)..controls
                        (\tikzinputsegmentsupporta)and(\tikzinputsegmentsupportb)
                        ..(\tikzinputsegmentlast);
                }
            },
            closepath code={
                \foreach\i in{1,...,\niterate}{
                \draw[put dots]
                    (\tikzinputsegmentfirst)--(\tikzinputsegmentlast);
                }
            }
        }
    }
}

\definecolor{xcrosscolor}{RGB}{197, 24, 27}

\usepackage{array,multirow,graphicx}
\usetikzlibrary{calc,arrows.meta,positioning}
\usetikzlibrary{shapes}
\usepackage[most]{tcolorbox}

\definecolor{silvercolr}{RGB}{217, 218, 219}

\usepackage{courier}

\usepackage[T1]{fontenc}
\usepackage{listings}

\usepackage{hyperref}
\hypersetup{
    colorlinks=false,
    pdfborder={0 0 0},
}

\usepackage{subcaption}
\usepackage[noabbrev, capitalise]{cleveref}
\crefname{lemma}{Lemma}{Lemmas}
\crefname{rul}{Rule}{Rules}

\definecolor{dkgreen}{rgb}{0,0.6,0}
\definecolor{gray}{rgb}{0.5,0.5,0.5}
\definecolor{mauve}{rgb}{0.58,0,0.82}

\crefformat{section}{\S#2#1#3} 
\crefformat{subsection}{\S#2#1#3}
\crefformat{subsubsection}{\S#2#1#3}

%% file: figure_genomeic_example.tex
\newcommand{\fsnodeName}{\fontsize{14}{14}\selectfont}
\newcommand{\fsProcess}{\fontsize{14}{14}\selectfont}

\newcommand{\pxStep}{2.3}       
\newcommand{\pyStep}{1.6}       
\newcommand{\pMemX}{0}          
\newcommand{\pCompOneX}{\pxStep}
\newcommand{\pCompTwoX}{2*\pxStep}
\newcommand{\pCompThreeX}{3*\pxStep}
\newcommand{\pMemFourX}{4*\pxStep}

\newcommand{\pNodeSize}{1}    
\newcommand{\pMemSize}{1.4}     
\newcommand{\pLineWeight}{0.7pt}
\newcommand{\pDottedWeight}{.9pt}
\newcommand{\pIconSizeW}{.95}   
\newcommand{\pIconSizeH}{1}   
\newcommand{\pIconSizeR}{1.15}   

\colorlet{p1Col}{blue!40}
\colorlet{p2Col}{green!60!black!50}
\colorlet{p3Col}{orange!60}
\colorlet{p4Col}{red!50}
\colorlet{p5Col}{teal!50}
\colorlet{p6Col}{orange!40}
\colorlet{memDraw}{teal}
\colorlet{compDraw}{blue!50}

\newcommand{\loadMemIcon}[1]{\includegraphics[width=\pIconSizeW cm, height=\pIconSizeH cm]{#1}}
\newcommand{\loadMemIconResult}[1]{\includegraphics[width=\pIconSizeR cm, height=\pIconSizeR cm]{#1}}

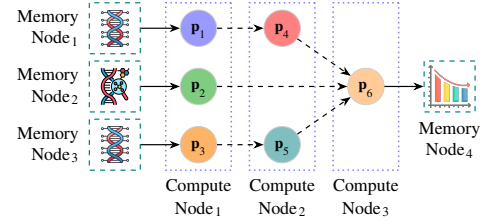
\begin{figure}[!t]
\centering

\resizebox{0.75\columnwidth}{!}{
\begin{tikzpicture}[
    scale=0.75, transform shape,
    mem_node/.style={draw=memDraw, dashed, line width=\pLineWeight, fill=white, minimum size=\pMemSize cm},
    comp_node/.style={draw=compDraw, dotted, line width=\pDottedWeight, fill=white, minimum width=1.8cm},
    proc/.style={
    	circle, 
    draw=black!20, 
    line width=0.5pt, 
    minimum size=\pNodeSize cm, 
    font=\large\bfseries,
    anchor=center,      
    inner sep=0pt,      
    text height=1.5ex,  
    text depth=0.5ex,   
    align=center        
},
    arrow/.style={-Stealth, line width=\pLineWeight},
    delegation/.style={-Stealth, dashed, line width=\pLineWeight}
]

    \foreach \i/\img in {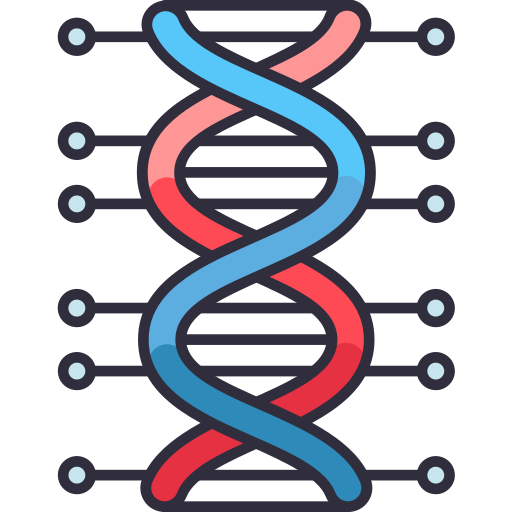, 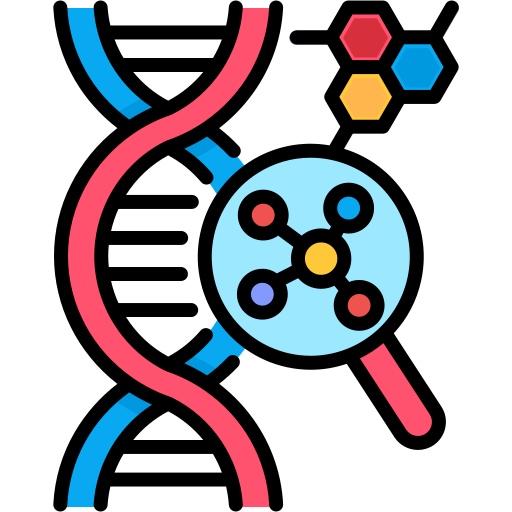, 3/gen_1.png} {
        \node[mem_node] (M\i) at (\pMemX, {-\i*\pyStep + \pyStep}) {\loadMemIcon{\img}};
        \node[left=0.15cm of M\i, align=right, font=\fsnodeName] {Memory\\[.4ex]Node$_{\,\i}$};
    }

    \node[comp_node, minimum height=4.6cm] (C1_box) at (\pCompOneX, -\pyStep) {};
    \node[proc, fill=p1Col] (P1) at (\pCompOneX, 0) {{p\raisebox{-0.7ex}{$_1$}}};
    \node[proc, fill=p2Col] (P2) at (\pCompOneX, -\pyStep) {{p\raisebox{-0.7ex}{$_2$}}};
    \node[proc, fill=p3Col] (P3) at (\pCompOneX, -2*\pyStep) {{p\raisebox{-0.7ex}{$_3$}}};
    \node[below=0.1cm of C1_box, font=\fsnodeName, align=center] {Compute \\[.6ex]Node$_{\,1}$};

    \node[comp_node, minimum height=4.6cm] (C2_box) at (\pCompTwoX, -\pyStep) {};
    \node[proc, fill=p4Col] (P4) at (\pCompTwoX, 0) {{p\raisebox{-0.7ex}{$_4$}}};
    \node[proc, fill=p5Col] (P5) at (\pCompTwoX, -2*\pyStep) {{p\raisebox{-0.7ex}{$_5$}}};
    \node[below=0.1cm of C2_box, font=\fsnodeName, align=center] {Compute \\[.6ex]Node$_{\,2}$};

    \node[comp_node, minimum height=4.6cm] (C3_box) at (\pCompThreeX, -\pyStep) {};
    \node[proc, fill=p6Col] (P6) at (\pCompThreeX, -\pyStep) {{p\raisebox{-0.7ex}{$_6$}}};
    \node[below=0.1cm of C3_box, font=\fsnodeName, align=center] {Compute \\[.6ex]Node$_{\,3}$};

    \node[mem_node] (M4) at (\pMemFourX, -\pyStep) {\loadMemIconResult{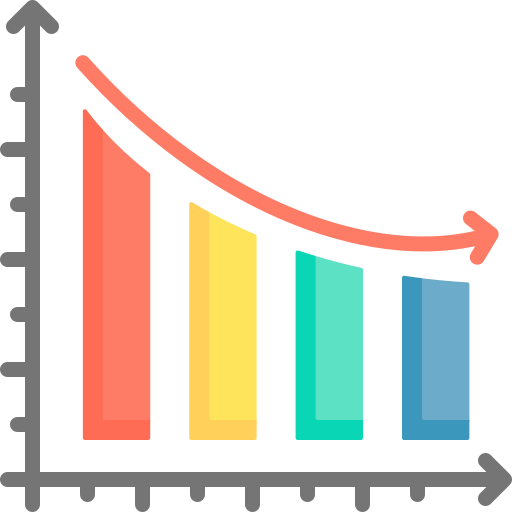}};
    \node[below=0.15cm of M4, align=center, font=\fsnodeName] {Memory\\[.4ex]Node$_{\,4}$};

    \draw[arrow] (M1) -- (P1);
    \draw[arrow] (M2) -- (P2);
    \draw[arrow] (M3) -- (P3);
    \draw[arrow] (P6) -- (M4);

    \draw[delegation] (P1) -- (P4);
    \draw[delegation] (P3) -- (P5);
    \draw[delegation] (P4) -- (P6);
    \draw[delegation] (P2) -- (P6);
    \draw[delegation] (P5) -- (P6);

\end{tikzpicture}
}
\caption{Genome data processing scenario. Legend: \directmemory =direct read/write to memory; \indirectmemory =write to memory followed by a read from the same location.}
\label{fig:gen_ex}

\end{figure}

%% file: figure_obj_cap.tex

\newcommand{\figScaleOcap}{0.45}

\newcommand{\ocapNodeSize}{1.45cm}

\colorlet{ocapObj}{cyan!20}
\colorlet{ocapProxy}{yellow!35}

\newcommand{\xB}{0}
\newcommand{\xF}{2.5}
\newcommand{\xA}{5}
\newcommand{\xR}{7.5}
\newcommand{\xC}{10}

\newcommand{\yMain}{0}

\newcommand{\edgeWidth}{0.75pt}

\newcommand{\ocapACStart}{north east}
\newcommand{\ocapACEnd}{north west}

\newcommand{\ocapFRStart}{south east}
\newcommand{\ocapFREnd}{south west}

\newcommand{\ocapACBend}{25}
\newcommand{\ocapFRBend}{25}

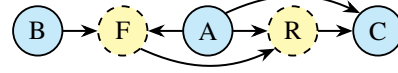
\begin{figure}[t]
\centering
\begin{tikzpicture}[
    scale=\figScaleOcap,
    every node/.style={transform shape},
    >=Stealth,
    ocapedge/.style={->, line width=\edgeWidth},
    ocapobj/.style={
        circle,
        draw,
        line width=.7pt, 
        fill=ocapObj,
        minimum size=\ocapNodeSize,
        inner sep=0pt,
        font=\fontsize{22}{22}\selectfont
    },
    ocapproxy/.style={
        circle,
        draw,
        dashed,
        line width=.7pt, 
        fill=ocapProxy,
        minimum size=\ocapNodeSize,
        inner sep=0pt,
        font=\fontsize{22}{22}\selectfont
    }
]

\node[ocapobj]   (B) at (\xB,\yMain) {B};
\node[ocapproxy] (F) at (\xF,\yMain) {F};
\node[ocapobj]   (A) at (\xA,\yMain) {A};
\node[ocapproxy] (R) at (\xR,\yMain) {R};
\node[ocapobj]   (C) at (\xC,\yMain) {C};

\draw[ocapedge] (B) -- (F);
\draw[ocapedge] (A) -- (F);   
\draw[ocapedge] (A) -- (R);
\draw[ocapedge] (R) -- (C);

\draw[ocapedge] (A.\ocapACStart) to[bend left=\ocapACBend] (C.\ocapACEnd);
\draw[ocapedge] (F.\ocapFRStart) to[bend right=\ocapFRBend] (R.\ocapFREnd);

\end{tikzpicture}

\caption{Capability-Object Model. In this model, F is the forwarding proxy and R is the revoking proxy.}
\label{fig:cap-obj-mdl}
\end{figure}

%% file: table_RW.tex
\newcolumntype{T}{>{\centering\arraybackslash}m{8.5mm}}
\newcolumntype{S}{>{\columncolor[HTML]{EFEFEF}\centering\arraybackslash}m{8.5mm}}

\newcommand{\syshead}[1]{%
  \rotatebox{80}{\fontsize{9.6pt}{9.6pt}\selectfont #1}%
}

\begin{table*}[!t]
\centering
\caption{Capability architectures taxonomized by disaggregation requirements.
Legend: \FULLY\  = Full support; \PART\  = Partial support; \NO\  = No support.
Enforcement: C = Compiler; Co-P = Co-Processor; FPGA = FPGA logic;
HW = Hardware logic; ISA = ISA; K = Kernel; OS = Operating System; SW = Software (host-trusted); SNIC = SmartNIC controller.}
\label{tab:coveredCriteria}

\setlength{\tabcolsep}{1.2pt}
\renewcommand{\arraystretch}{1.1}

\resizebox{\linewidth}{!}{%
\fontsize{10pt}{10pt}\selectfont
\begin{tabular}{@{}>{\raggedright\arraybackslash}m{29mm} *{22}{T} S@{}}
\toprule

& \multicolumn{6}{c}{\textbf{Architectural / Hardware-Assisted}}
& \multicolumn{4}{c}{\textbf{OS Capability Systems}}
& \multicolumn{11}{c}{\textbf{Multiprocessor, Distributed, and Disaggregated Systems}}
& \cellcolor[HTML]{FFFFFF} \\

\cmidrule(lr){2-7}
\cmidrule(lr){8-11}
\cmidrule(lr){12-23}

& \syshead{CAP~\cite{needham1977cambridge}}
& \syshead{IBM/38~\cite{houdek1981ibm}}
& \syshead{Hardbound~\cite{devietti2008hardbound}}
& \syshead{CHERI~\cite{woodruff2014cheri}}
& \syshead{M-Machine~\cite{carter1994hardware}}
& \syshead{CODOMs~\cite{vilanova2014codoms}}

& \syshead{Hydra~\cite{wulf1974hydra}}
& \syshead{KeyKOS~\cite{hardy1985keykos}}
& \syshead{EROS~\cite{shapiro1999eros}}
& \syshead{seL4~\cite{klein2009sel4}}

& \syshead{StarOS~\cite{jones1979staros}}
& \syshead{iAPX432~\cite{cox1981unified}}
& \syshead{Chorus~\cite{rozier1992overview}}
& \syshead{Amoeba~\cite{mullender1990amoeba}}
& \syshead{Accent~\cite{rashid1981accent}}
& \syshead{Mach~\cite{accetta1986mach}}
& \syshead{Barrelfish~\cite{baumann2009multikernel}}
& \syshead{Zeno~\cite{ehret2022zeno}}
& \syshead{CEP~\cite{azriel2019memory}}
& \syshead{SemperOS~\cite{hille2019semperos}}
& \syshead{FractOS~\cite{vilanova2022slashing}}
& \syshead{White et al.~\cite{white2025enabling}}

& \cellcolor[HTML]{EFEFEF}\syshead{\textbf{SADRA}} \\

\midrule

Inter-Node Med.
& \NO & \NO & \NO & \NO & \NO & \NO
& \NO & \NO & \NO & \NO
& \NO & \NO & \PART & \PART & \PART & \PART & \PART & \PART
& \PART & \PART & \PART & \PART
& \FULLY \\

\rowcolor{gray!5}
Subject Gran.
& \FULLY & \FULLY & \NO & \FULLY & \PART & \PART
& \FULLY & \FULLY & \FULLY & \FULLY
& \FULLY & \FULLY & \FULLY & \PART & \FULLY & \FULLY & \FULLY
& \PART & \FULLY & \FULLY & \FULLY & \FULLY
& \FULLY \\

Object Gran.
& \PART & \PART & \FULLY & \FULLY & \FULLY & \FULLY
& \PART & \PART & \PART & \PART
& \PART & \PART & \PART & \PART & \PART & \PART & \PART
& \FULLY & \FULLY & \FULLY & \FULLY & \FULLY
& \FULLY \\

\rowcolor{gray!5}
Delegation
& \PART & \PART & \NO & \PART & \PART & \FULLY
& \PART & \FULLY & \FULLY & \FULLY
& \PART & \PART & \PART & \PART & \PART & \PART & \FULLY
& \FULLY & \FULLY & \FULLY & \FULLY & \FULLY
& \FULLY \\

Revocation
& \PART & \PART & \NO & \PART & \PART & \FULLY
& \PART & \PART & \FULLY & \FULLY
& \NO & \PART & \PART & \PART & \PART & \PART & \PART
& \PART & \PART & \PART & \FULLY & \FULLY 
& \FULLY \\

\rowcolor{gray!5}
Persistence
& \NO & \PART & \NO & \NO & \PART & \NO
& \NO & \FULLY & \FULLY & \NO
& \NO & \PART & \NO & \NO & \NO & \NO & \NO
& \NO& \PART & \NO & \NO & \NO 
& \FULLY \\

\midrule

Enforcement
& \fontsize{9.1pt}{9.1pt}\selectfont ISA/K
& \fontsize{9.1pt}{9.1pt}\selectfont ISA
& \fontsize{9.1pt}{9.1pt}\selectfont ISA/C
& \fontsize{9.1pt}{9.1pt}\selectfont ISA/C
& \fontsize{9.1pt}{9.1pt}\selectfont HW
& \fontsize{9.1pt}{9.1pt}\selectfont HW/K
& \fontsize{9.1pt}{9.1pt}\selectfont K
& \fontsize{9.1pt}{9.1pt}\selectfont K
& \fontsize{9.1pt}{9.1pt}\selectfont K
& \fontsize{9.1pt}{9.1pt}\selectfont K
& \fontsize{9.1pt}{9.1pt}\selectfont ISA/K
& \fontsize{9.1pt}{9.1pt}\selectfont ISA
& \fontsize{9.1pt}{9.1pt}\selectfont K
& \fontsize{9.1pt}{9.1pt}\selectfont K
& \fontsize{9.1pt}{9.1pt}\selectfont K
& \fontsize{9.1pt}{9.1pt}\selectfont K
& \fontsize{9.1pt}{9.1pt}\selectfont K
& \fontsize{9.1pt}{9.1pt}\selectfont FPGA
& \fontsize{9.1pt}{9.1pt}\selectfont Co-P
& \fontsize{9.1pt}{9.1pt}\selectfont HW/K
& \fontsize{9.1pt}{9.1pt}\selectfont SW
& \fontsize{9.1pt}{9.1pt}\selectfont SW
& \fontsize{9.1pt}{9.1pt}\selectfont \textbf{SNIC} \\

\bottomrule
\end{tabular}%
}
\end{table*}

%% file: figure_systems.tex

\newcommand{\fsNodeLabelsCombined}{\fontsize{18}{18}\selectfont}
\newcommand{\fsProcessTxtCombined}{\fontsize{19}{19}\selectfont}
\newcommand{\fsControllerTxtCombined}{\fontsize{19}{19}\selectfont}
\newcommand{\fsNicTxtCombined}{\fontsize{13}{13}\selectfont}
\newcommand{\fsFabricTxtCombined}{\fontsize{18}{18}\selectfont}
\newcommand{\fsSnLabelTxtCombined}{\fontsize{21}{21}\selectfont}

\newcommand{\xCompOneC}{0.0}
\newcommand{\xCompTwoC}{4.5}
\newcommand{\xCompEllipsisC}{7.4}
\newcommand{\xCompMC}{10.2}
\newcommand{\xResOneC}{14.8}
\newcommand{\xResTwoC}{19.3}
\newcommand{\xResEllipsisC}{22.2}
\newcommand{\xResNC}{25.0}

\newcommand{\yEllipsisC}{7.5}
\newcommand{\szEllipsisC}{30}
\newcommand{\wtEllipsisC}{\bfseries}

\newcommand{\arrowWeightC}{0.5pt}
\newcommand{\boxWeightC}{0.5pt}
\newcommand{\dottedWeightC}{0.7pt}
\newcommand{\fabricWeightC}{0.4pt}

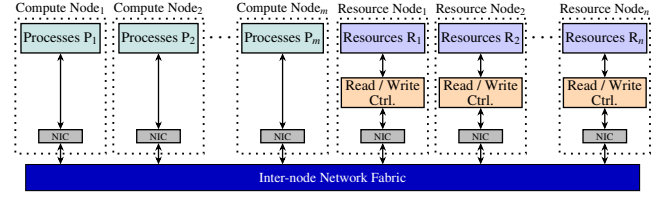
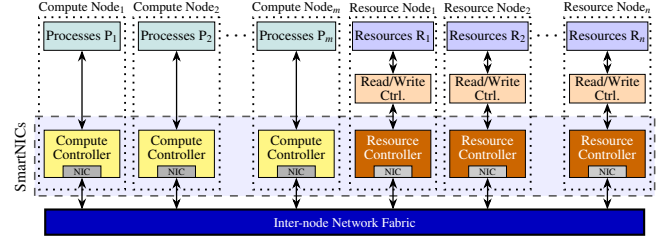
\begin{figure}[t] 
    \centering 
    
    \begin{subfigure}{\columnwidth} 
        \centering
        \begin{tikzpicture}[
            scale=0.288, 
            every node/.style={transform shape},
            proc_style/.style={rectangle, draw, line width=\boxWeightC, fill=teal!20, minimum width=3.6cm, minimum height=1.35cm, align=center, font=\fsProcessTxtCombined},
            res_style/.style={rectangle, draw, line width=\boxWeightC, fill=blue!20, minimum width=3.6cm, minimum height=1.35cm, align=center, font=\fsProcessTxtCombined},
            rw_style/.style={rectangle, draw, line width=\boxWeightC, fill=orange!30, minimum width=3.8cm, minimum height=1.35cm, align=center, font=\fsControllerTxtCombined},
            nic_style/.style={rectangle, draw, line width=\boxWeightC, fill=gray!50, minimum width=2.0cm, minimum height=0.6cm, font=\fsNicTxtCombined},
            fabric_style/.style={rectangle, draw, line width=\fabricWeightC, fill=blue!75!black, text=white, minimum width=28.2cm, minimum height=1.2cm, font=\fsFabricTxtCombined},
            bidir/.style={line width=\arrowWeightC, {Stealth[length=3pt,width=2.5pt]}-{Stealth[length=3pt,width=2.5pt]}},
            dotted_style/.style={draw=black, dotted, line width=\dottedWeightC, minimum width=4.15cm, minimum height=6.2cm}
        ]
            \foreach \x/\sub/\idx in {\xCompOneC/1/1,\xCompTwoC/2/2,\xCompMC/m/3} {
                \node[dotted_style] (contC\idx) at (\x,5.3) {};
                \node[above,font=\fsNodeLabelsCombined] at (\x,8.45) {Compute Node$_\sub$};
                \node[proc_style] (cp\idx) at (\x,7.5) {Processes P$_\sub$};
                \node[nic_style] (cnic\idx) at (\x,3) {NIC};
                \draw[bidir] (cp\idx.south) -- (cnic\idx.north);
                \draw[bidir] (cnic\idx.south) -- (\x,1.1+1.2/2);
            }
            \foreach \x/\sub/\idx in {\xResOneC/1/1,\xResTwoC/2/2,\xResNC/n/3} {
                \node[dotted_style] (contR\idx) at (\x,5.3) {};
                \node[above,font=\fsNodeLabelsCombined] at (\x,8.45) {Resource Node$_\sub$};
                \node[res_style] (rp\idx) at (\x,7.5) {Resources R$_\sub$};
                \node[rw_style] (rw\idx) at (\x,5.0) {Read / Write\\Ctrl.};
                \node[nic_style] (rnic\idx) at (\x,3) {NIC};
                \draw[bidir] (rp\idx.south) -- (rw\idx.north);
                \draw[bidir] (rw\idx.south) -- (rnic\idx.north);
                \draw[bidir] (rnic\idx.south) -- (\x,1.1+1.2/2);
            }
            \node[font=\fontsize{\szEllipsisC}{\szEllipsisC}\selectfont\wtEllipsisC] at (\xCompEllipsisC,\yEllipsisC) {$\cdots$};
            \node[font=\fontsize{\szEllipsisC}{\szEllipsisC}\selectfont\wtEllipsisC] at (\xResEllipsisC,\yEllipsisC) {$\cdots$};
            \node[fabric_style] (fabric) at ({(\xCompOneC+\xResNC)/2},1.1) {Inter-node Network Fabric};
        \end{tikzpicture}
        \caption{Abstract architecture of a resource-disaggregated system.}
        \label{fig:rda_system}
    \end{subfigure}

    \vspace{.4cm} 

 \begin{subfigure}{\columnwidth}
        \centering
        \begin{tikzpicture}[
            scale=0.278, 
            every node/.style={transform shape},
            proc_style/.style={rectangle, draw, line width=0.35pt, fill=teal!20, minimum width=3.6cm, minimum height=1.35cm, align=center, font=\fsProcessTxtCombined},
            res_style/.style={rectangle, draw, line width=0.35pt, fill=blue!20, minimum width=3.6cm, minimum height=1.35cm, align=center, font=\fsProcessTxtCombined},
            rw_style/.style={rectangle, draw, line width=0.35pt, fill=orange!30, minimum width=3.6cm, minimum height=1.3cm, align=center, font=\fsControllerTxtCombined},
            c_ctrl_style/.style={rectangle, draw, line width=0.35pt, fill=yellow!60, minimum width=3.6cm, minimum height=2.25cm},
            r_ctrl_style/.style={rectangle, draw, line width=0.35pt, fill=orange!80!black, minimum width=3.6cm, minimum height=2.25cm},
            nic_style/.style={rectangle, draw, line width=0.35pt, fill=gray!60, minimum width=1.8cm, minimum height=0.5cm, font=\fsNicTxtCombined},
            fabric_style/.style={rectangle, draw, line width=\fabricWeightC, fill=blue!75!black, text=white, thick, minimum width=28.5cm, minimum height=1.2cm, font=\fsFabricTxtCombined},
            bidir/.style={line width=\arrowWeightC, {Stealth[scale=0.8]}-{Stealth[scale=0.8]}},
            dotted_style/.style={draw=black, dotted, line width=\dottedWeightC, minimum width=4.1cm, minimum height=8.2cm},
            sn_shade/.style={draw=black!65, dashed, line width=0.6pt, fill=blue!7}
        ]
            \begin{scope}[on background layer]
                \draw[sn_shade] ({(\xCompOneC+\xResNC)/2 - 29.5/2}, -0.1) rectangle ({(\xCompOneC+\xResNC)/2 + 29.5/2}, 3.68);
            \end{scope}
            \node[anchor=south, rotate=90, font=\fsSnLabelTxtCombined] at (-2.55, 1.75) {SmartNICs};

            \foreach \x/\sub/\idx in {\xCompOneC/1/1, \xCompTwoC/2/2, \xCompMC/m/3} {
                \node[dotted_style] (contC\idx) at (\x, 4.3) {};
                \node[above, font=\fsNodeLabelsCombined] at (\x, 8.45) {Compute Node$_\sub$};
                \node[proc_style] (cp\idx) at (\x, 7.5) {Processes P$_\sub$};
                \node[c_ctrl_style] (cc\idx) at (\x, 1.9) {};
                \node[font=\fsControllerTxtCombined, align=center] at (\x, 2.2) {Compute\\Controller};
                \node[nic_style] (cnic\idx) at (\x, 1.05) {NIC};
                \draw[bidir] (cp\idx.south) -- (cc\idx.north);
                \draw[bidir] (cnic\idx.south) -- (\x, -1.35 + 1.2/2);
            }
            \foreach \x/\sub/\idx in {\xResOneC/1/1, \xResTwoC/2/2, \xResNC/n/3} {
                \node[dotted_style] (contR\idx) at (\x, 4.3) {};
                \node[above, font=\fsNodeLabelsCombined] at (\x, 8.45) {Resource Node$_\sub$};
                \node[res_style] (rp\idx) at (\x, 7.5) {Resources R$_\sub$};
                \node[rw_style] (rw\idx) at (\x, 5.0) {Read/Write\\Ctrl.};
                \node[r_ctrl_style] (rc\idx) at (\x, 1.9) {};
                \node[font=\fsControllerTxtCombined, align=center, text=white] at (\x, 2.2) {Resource\\Controller};
                \node[nic_style, fill=gray!40] (rnic\idx) at (\x, 1.05) {NIC};
                \draw[bidir] (rp\idx.south) -- (rw\idx.north);
                \draw[bidir] (rw\idx.south) -- (rc\idx.north);
                \draw[bidir] (rnic\idx.south) -- (\x, -1.35 + 1.2/2);
            }
            \node[font=\fontsize{\szEllipsisC}{\szEllipsisC}\selectfont\wtEllipsisC] at (\xCompEllipsisC, \yEllipsisC) {$\cdots$};
            \node[font=\fontsize{\szEllipsisC}{\szEllipsisC}\selectfont\wtEllipsisC] at (22.2, \yEllipsisC) {$\cdots$};
            \node[fabric_style] (fabric) at ({(\xCompOneC+\xResNC)/2}, -1.35) {Inter-node Network Fabric};
        \end{tikzpicture}
        \caption{SADRA capability-based access control design.}
        \label{fig:sadra_system}
    \end{subfigure}

    \caption{Overview of resource disaggregation and SADRA's architectural intervention.}
    \label{fig:combined_system}
\end{figure}

%% file: figure_sadra_caps.tex
\providecommand{\boxWeight}{0.8pt}
\providecommand{\arrowWeight}{1pt}
\providecommand{\arrowHeadSize}{0.8}

\newcommand{\fsCapTxt}{\fontsize{24}{24}\selectfont}

\definecolor{cProcRGB}{RGB}{85,107,47}
\definecolor{cCompRGB}{RGB}{70,130,180}
\definecolor{cResRGB}{RGB}{233,116,81}

\colorlet{cProc}{cProcRGB!40}
\colorlet{cComp}{cCompRGB!65}
\colorlet{cRes}{cResRGB!80}
\newcommand{\xCapProc}{0}
\newcommand{\xCapComp}{8}
\newcommand{\xCapRes}{16}

\newcommand{\dCapObj}{5.2cm}

\newcommand{\barSize}{.94}
\newcommand{\barWeight}{1.2pt}
\newcommand{\xBarOne}{2.15}
\newcommand{\xBarTwo}{10.15}


\begin{figure}[t]
\centering
\begin{tikzpicture}[
    scale=0.3,
    every node/.style={transform shape},
    font=\fontsize{26}{26}\selectfont,
    proc_node/.style={
        circle,
        draw,
        line width=\boxWeight,
        fill=cProc,
        text=black,
        minimum size=\dCapObj,
        align=center
    },
    comp_node/.style={
        circle,
        draw,
        line width=\boxWeight,
        fill=cComp,
        text=black!90,
        minimum size=\dCapObj,
        align=center
    },
    res_node/.style={
        circle,
        draw,
        line width=\boxWeight,
        fill=cRes,
        text=black!90,
        minimum size=\dCapObj,
        align=center
    },
    cap_arrow/.style={
        line width=\arrowWeight,
        -{Stealth[scale=\arrowHeadSize]}
    },
    boundary_bar/.style={
        line width=\barWeight
    }
]

\node[proc_node] (P) at (\xCapProc,0) {Process\\Capability};
\node[comp_node] (M) at (\xCapComp,0) {Compute\\Capability};
\node[res_node]  (R) at (\xCapRes,0)  {Resource\\Capability};

\draw[cap_arrow] (\xBarOne,0) -- (M.west);
\draw[boundary_bar] (\xBarOne,0.5*\barSize) -- (\xBarOne,-0.5*\barSize);
\draw[line width=\arrowWeight] (P.east) -- (\xBarOne,0);

\draw[cap_arrow] (\xBarTwo,0) -- (R.west);
\draw[boundary_bar] (\xBarTwo,0.5*\barSize) -- (\xBarTwo,-0.5*\barSize);
\draw[line width=\arrowWeight] (M.east) -- (\xBarTwo,0);

\end{tikzpicture}

\caption{Capability-object model with process ($\mathbb{C}_p$), compute
($\mathbb{C}_c$), and resource ($\mathbb{C}_r$) capabilities. Edges
represent the $\mapsto$ relation linking corresponding capabilities
$c_p \mapsto c_c \mapsto c_r$. Each capability represents authority
over $r \in R_j$, with permissions and resource extent monotonically
attenuated along the chain for enforcement at the compute and resource
controllers.}
\label{fig:caps}
\end{figure}
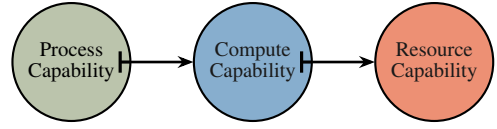

%% file: figure_example_trees.tex

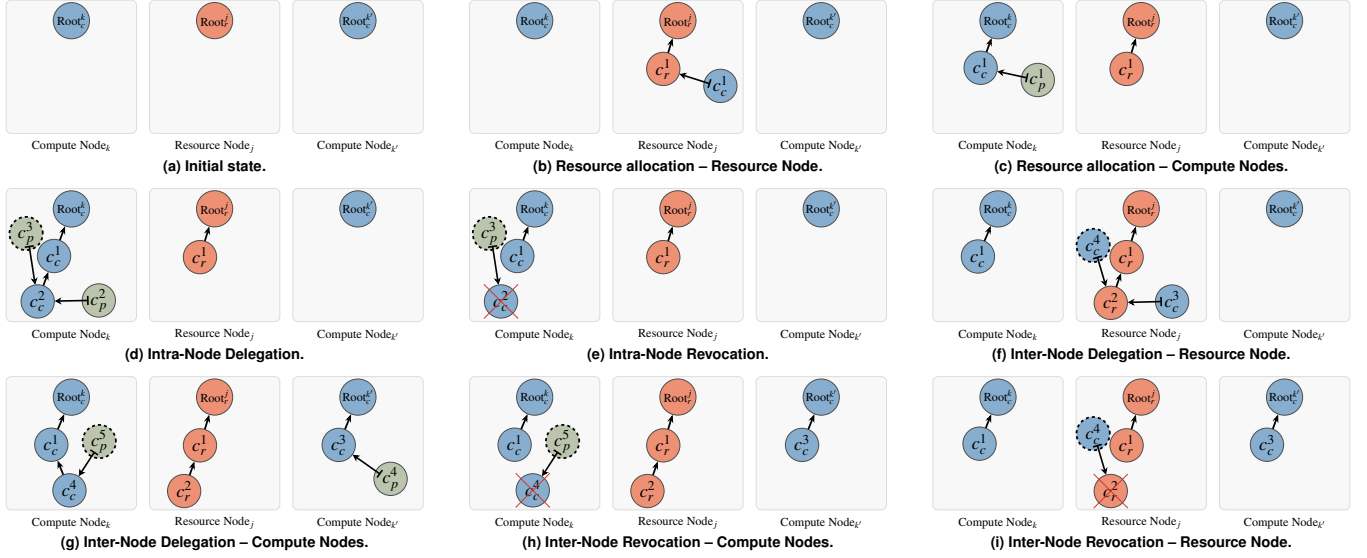
\begin{figure*}[!t]
\centering

\definecolor{ResourceCapColor}{RGB}{233,116,81}
\definecolor{ComputeCapColor}{RGB}{70,130,180}
\definecolor{ProcessCapColor}{RGB}{85,107,47}
\definecolor{IndicatorColor}{RGB}{230,230,230}
\definecolor{SubfigureBoxDrawColor}{RGB}{220,220,220}
\definecolor{SubfigureBoxFillColor}{RGB}{248,248,248}
\definecolor{RevokedColor}{RGB}{204,0,0}

\providecommand{\rcapcircle}{\tikz[baseline=-0.5ex]\draw[draw=black!70,fill=ResourceCapColor!80] (0,0) circle (0.5ex);}
\providecommand{\icapcircle}{\tikz[baseline=-0.5ex]\draw[draw=black!70,fill=ComputeCapColor!65] (0,0) circle (0.5ex);}
\providecommand{\pcapcircle}{\tikz[baseline=-0.5ex]\draw[draw=black!70,fill=ProcessCapColor!40] (0,0) circle (0.5ex);}
\providecommand{\solidcircle}{\tikz[baseline=-0.5ex]\draw[thick] (0,0) circle (0.5ex);}
\providecommand{\dashedcircle}{\tikz[baseline=-0.5ex]\draw[thick,dashed] (0,0) circle (0.5ex);}
\providecommand{\circledashedarrow}{\tikz[baseline=-0.5ex]\draw[-{Stealth[scale=0.5]},dashed] (0,0) -- (1em,0);}
\providecommand{\captreedarrow}{\tikz[baseline=-0.5ex]\draw[-{Stealth[scale=0.5]}] (0,0) -- (1em,0);}

\tikzset{
  sadrafig/.cd,
  figure width/.initial=\textwidth,
  subfig scale/.initial=14.5cm,
  grid x/.initial=16cm,
  grid y/.initial=6.5cm,
  box width/.initial=4.5cm,
  box height/.initial=4.6cm,
  box rounded corners/.initial=5pt,
  box draw/.initial=SubfigureBoxDrawColor,
  box fill/.initial=SubfigureBoxFillColor,
  root minimum size/.initial=1.25cm,
  cap minimum size/.initial=1.15cm,
  revoked minimum size/.initial=1.15cm,
  root inner sep/.initial=0pt,
  cap inner sep/.initial=0pt,
  revoked inner sep/.initial=0pt,
  domain font size/.initial=12pt,
  domain baseline/.initial=14.4pt,
  root font size/.initial=12pt,
  root baseline/.initial=14.4pt,
  cap font size/.initial=18pt,
  cap baseline/.initial=20.4pt,
  title font size/.initial=16pt,
  title baseline/.initial=18.4pt,
  node draw/.initial=black!70,
  indicator draw/.initial=black,
  indicator line width/.initial=1.75pt,
  link line width/.initial=1.75pt,
  arrow scale/.initial=0.75,
  elbow dx/.initial=-0.025,
  revoked line width/.initial=0.95pt,
  compute fill/.initial=ComputeCapColor!65,
  resource fill/.initial=ResourceCapColor!80,
  process fill/.initial=ProcessCapColor!40,
  indicator fill/.initial=IndicatorColor,
  revoked draw/.initial=RevokedColor!80,
  coord/CNk/.initial={(0.15856225,0.15000000)},
  coord/RNj/.initial={(0.49991045,0.15000000)},
  coord/CNkp/.initial={(0.84133720,0.15000000)},
  coord/label CNk/.initial={(0.15856225,-0.04000000)},
  coord/label RNj/.initial={(0.49991045,-0.04000000)},
  coord/label CNkp/.initial={(0.84133720,-0.04000000)},
  coord/subcaption/.initial={(0.49991045,-0.08800000)},
  coord/root Rck/.initial={(0.15866293,0.25976539)},
  coord/root Rrj/.initial={(0.49976471,0.26004780)},
  coord/root Rckp/.initial={(0.84110190,0.25985064)},
  coord/b/cr1/.initial={(0.46490988,0.14496360)},
  coord/b/cc1r/.initial={(0.60094508,0.10454627)},
  coord/c/cr1/.initial={(0.46490988,0.14496360)},
  coord/c/cc1/.initial={(0.11730794,0.14765438)},
  coord/c/cp1/.initial={(0.25343849,0.11857253)},
  coord/d/cr1/.initial={(0.46490988,0.14496360)},
  coord/d/cc1/.initial={(0.11730794,0.14765448)},
  coord/d/cc2/.initial={(0.07841979,0.03948421)},
  coord/d/cp3/.initial={(0.05081004,0.20124818)},
  coord/d/cp2/.initial={(0.22402007,0.04177109)},
  coord/e/cr1/.initial={(0.46490988,0.14496360)},
  coord/e/cc1/.initial={(0.11730794,0.14765438)},
  coord/e/cp3/.initial={(0.05081004,0.20124809)},
  coord/e/cc2/.initial={(0.07841978,0.03948412)},
  coord/f/cc1/.initial={(0.11085391,0.14648996)},
  coord/f/cr1/.initial={(0.46490958,0.14496360)},
  coord/f/cr2/.initial={(0.42721021,0.03613987)},
  coord/f/cc4r/.initial={(0.38524297,0.17280628)},
  coord/f/cc3r/.initial={(0.57328281,0.03923910)},
  coord/g/cc1/.initial={(0.11085424,0.14648974)},
  coord/g/cp5/.initial={(0.22498138,0.15400000)},
  coord/g/cc4/.initial={(0.15404872,0.03777645)},
  coord/g/cr1/.initial={(0.46490996,0.14496362)},
  coord/g/cr2/.initial={(0.42721046,0.03613940)},
  coord/g/cc3/.initial={(0.79540034,0.14601672)},
  coord/g/cp4/.initial={(0.91883190,0.06519552)},
  coord/h/cc1/.initial={(0.11085424,0.14648974)},
  coord/h/cp5/.initial={(0.22498138,0.15400000)},
  coord/h/cc4/.initial={(0.15404872,0.03777645)},
  coord/h/cr1/.initial={(0.46490996,0.14496362)},
  coord/h/cr2/.initial={(0.42721046,0.03613940)},
  coord/h/cc3/.initial={(0.79540034,0.14601672)},
  coord/i/cc1/.initial={(0.11470205,0.14896360)},
  coord/i/cr1/.initial={(0.46490961,0.14496337)},
  coord/i/cr2/.initial={(0.42721023,0.03613965)},
  coord/i/cc3/.initial={(0.79927076,0.14645602)},
  coord/i/cc4r/.initial={(0.38524300,0.17280606)}
}

\newcommand{\SFkey}[1]{\pgfkeysvalueof{/tikz/sadrafig/#1}}
\newcommand{\SFcoord}[1]{\pgfkeysvalueof{/tikz/sadrafig/coord/#1}}
\newcommand{\SFelbow}{\SFkey{elbow dx}}

\newcommand{\SadraDomainFont}{\fontsize{\SFkey{domain font size}}{\SFkey{domain baseline}}\selectfont}
\newcommand{\SadraTitleFont}{\sffamily\bfseries\fontsize{\SFkey{title font size}}{\SFkey{title baseline}}\selectfont}
\newcommand{\SadraRootFont}{\bfseries\fontsize{\SFkey{root font size}}{\SFkey{root baseline}}\selectfont}
\newcommand{\SadraCapFont}{\fontsize{\SFkey{cap font size}}{\SFkey{cap baseline}}\selectfont}

\newcommand{\SadraCapLabel}[1]{{\SadraCapFont\ensuremath{#1}}}
\newcommand{\SadraRootLabel}[1]{{\SadraRootFont\ensuremath{#1}}}

\tikzset{
  sadra subfigure box/.style={
    draw=\SFkey{box draw},
    fill=\SFkey{box fill},
    rounded corners=\SFkey{box rounded corners},
    minimum width=\SFkey{box width},
    minimum height=\SFkey{box height}
  },
  sadra domain label/.style={
    font=\SadraDomainFont,
    execute at begin node=\SadraDomainFont
  },
  sadra subcaption/.style={
    font=\SadraTitleFont,
    execute at begin node=\SadraTitleFont
  },
  color_compute/.style={draw=\SFkey{node draw}, fill=\SFkey{compute fill}},
  color_resource/.style={draw=\SFkey{node draw}, fill=\SFkey{resource fill}},
  color_process/.style={draw=\SFkey{node draw}, fill=\SFkey{process fill}},
  color_indicator_c/.style={draw=\SFkey{indicator draw}, dashed, line width=\SFkey{indicator line width}, fill=\SFkey{compute fill}},
  color_indicator_p/.style={draw=\SFkey{indicator draw}, dashed, line width=\SFkey{indicator line width}, fill=\SFkey{process fill}},
  root_node/.style={
    circle,
    minimum size=\SFkey{root minimum size},
    inner sep=\SFkey{root inner sep},
    font=\SadraRootFont,
    execute at begin node=\SadraRootFont
  },
  cap_node/.style={
    circle,
    minimum size=\SFkey{cap minimum size},
    inner sep=\SFkey{cap inner sep},
    font=\SadraCapFont,
    execute at begin node=\SadraCapFont
  },
  link_to_arrow/.style={
    {|[width=9pt]}-{Stealth[length=7pt,width=7pt]},
    line width=1.6pt
  },
  link_arrow/.style={-{Stealth[scale=\SFkey{arrow scale}]}, line width=\SFkey{link line width}},
  revoked_x/.style={
    draw=\SFkey{revoked draw},
    line width=\SFkey{revoked line width},
    cross out,
    minimum size=\SFkey{revoked minimum size},
    inner sep=\SFkey{revoked inner sep}
  }
}

\newcounter{sadraSubfigure}

\renewcommand{\thesadraSubfigure}{\alph{sadraSubfigure}}

\newcommand{\SADRASubfigure}[4]{%
  \refstepcounter{sadraSubfigure}%
  \begin{scope}[x=\SFkey{subfig scale},y=\SFkey{subfig scale}]
    \node[sadra subfigure box] (CNk)  at \SFcoord{CNk} {};
    \node[sadra subfigure box] (RNj)  at \SFcoord{RNj} {};
    \node[sadra subfigure box] (CNkp) at \SFcoord{CNkp} {};

    \node[sadra domain label] at \SFcoord{label CNk}  {Compute Node$_k$};
    \node[sadra domain label] at \SFcoord{label RNj}  {Resource Node$_j$};
    \node[sadra domain label] at \SFcoord{label CNkp} {Compute Node$_{k'}$};

    \node[sadra subcaption] at \SFcoord{subcaption} {(#1) #2};
    \label{#3}%

    \node[root_node, color_compute]  (Rck)  at \SFcoord{root Rck}  {\SadraRootLabel{\mathrm{Root}_{c}^{k}}};
    \node[root_node, color_resource] (Rrj)  at \SFcoord{root Rrj}  {\SadraRootLabel{\mathrm{Root}_{r}^{j}}};
    \node[root_node, color_compute]  (Rckp) at \SFcoord{root Rckp} {\SadraRootLabel{\mathrm{Root}_{c}^{k'}}};

    #4
  \end{scope}%
}

\resizebox{\SFkey{figure width}}{!}{%
\begin{tikzpicture}

\begin{scope}[shift={(0,0)}]
  \SADRASubfigure{a}{Initial state.}{fig:tree-1}{}
\end{scope}

\begin{scope}[shift={(\SFkey{grid x},0)}]
  \SADRASubfigure{b}{Resource allocation -- Resource Node.}{fig:tree-2}{
    \node[cap_node, color_resource] (cr1)  at \SFcoord{b/cr1}  {\SadraCapLabel{c_r^1}};
    \node[cap_node, color_compute]  (cc1r) at \SFcoord{b/cc1r} {\SadraCapLabel{c_c^1}};
    \draw[link_arrow] (cr1) -- (Rrj);
    \draw[link_to_arrow] ([xshift=-10pt,yshift=2pt]cc1r.center) -- (cr1);
  }
\end{scope}

\begin{scope}[shift={(2*\SFkey{grid x},0)}]
  \SADRASubfigure{c}{Resource allocation -- Compute Nodes.}{fig:tree-3}{
    \node[cap_node, color_resource] (cr1) at \SFcoord{c/cr1} {\SadraCapLabel{c_r^1}};
    \node[cap_node, color_compute]  (cc1) at \SFcoord{c/cc1} {\SadraCapLabel{c_c^1}};
    \node[cap_node, color_process]  (cp1) at \SFcoord{c/cp1} {\SadraCapLabel{c_p^1}};
    \draw[link_arrow] (cr1) -- (Rrj);
    \draw[link_arrow] (cc1) -- (Rck);
    \draw[link_to_arrow] ([xshift=-10pt,yshift=2pt]cp1.center) -- (cc1);
  }
\end{scope}

\begin{scope}[shift={(0,-\SFkey{grid y})}]
  \SADRASubfigure{d}{Intra-Node Delegation.}{fig:tree-4}{
    \node[cap_node, color_resource]  (cr1) at \SFcoord{d/cr1} {\SadraCapLabel{c_r^1}};
    \node[cap_node, color_compute]   (cc1) at \SFcoord{d/cc1} {\SadraCapLabel{c_c^1}};
    \node[cap_node, color_compute]   (cc2) at \SFcoord{d/cc2} {\SadraCapLabel{c_c^2}};
    \node[cap_node, color_indicator_p] (cp3) at \SFcoord{d/cp3} {\SadraCapLabel{c_p^3}};
    \node[cap_node, color_process]   (cp2) at \SFcoord{d/cp2} {\SadraCapLabel{c_p^2}};
    \draw[link_arrow] (cr1) -- (Rrj);
    \draw[link_arrow] (cc1) -- (Rck);
    \draw[link_arrow] (cc2) -- (cc1);
    \draw[link_to_arrow] ([xshift=3pt,yshift=-12pt]cp3.center) -- (cc2);
    \draw[link_to_arrow] ([xshift=-10pt]cp2.center) -- (cc2);
  }
\end{scope}

\begin{scope}[shift={(\SFkey{grid x},-\SFkey{grid y})}]
  \SADRASubfigure{e}{Intra-Node Revocation.}{fig:tree-5}{
    \node[cap_node, color_resource]  (cr1) at \SFcoord{e/cr1} {\SadraCapLabel{c_r^1}};
    \node[cap_node, color_compute]   (cc1) at \SFcoord{e/cc1} {\SadraCapLabel{c_c^1}};
    \node[cap_node, color_indicator_p] (cp3) at \SFcoord{e/cp3} {\SadraCapLabel{c_p^3}};
    \node[cap_node, color_compute]   (cc2) at \SFcoord{e/cc2} {\SadraCapLabel{c_c^2}};
    \draw[link_arrow] (cr1) -- (Rrj);
    \draw[link_arrow] (cc1) -- (Rck);
    \draw[link_to_arrow] ([xshift=3pt,yshift=-12pt]cp3.center) -- (cc2);
    \node[revoked_x] at (cc2) {};
  }
\end{scope}

\begin{scope}[shift={(2*\SFkey{grid x},-\SFkey{grid y})}]
  \SADRASubfigure{f}{Inter-Node Delegation -- Resource Node.}{fig:tree-6}{
    \node[cap_node, color_compute]   (cc1)  at \SFcoord{f/cc1}  {\SadraCapLabel{c_c^1}};
    \node[cap_node, color_resource]  (cr1)  at \SFcoord{f/cr1}  {\SadraCapLabel{c_r^1}};
    \node[cap_node, color_resource]  (cr2)  at \SFcoord{f/cr2}  {\SadraCapLabel{c_r^2}};
    \node[cap_node, color_indicator_c] (cc4r) at \SFcoord{f/cc4r} {\SadraCapLabel{c_c^4}};
    \node[cap_node, color_compute]   (cc3r) at \SFcoord{f/cc3r} {\SadraCapLabel{c_c^3}};
    \draw[link_arrow] (cc1) -- (Rck);
    \draw[link_arrow] (cr1) -- (Rrj);
    \draw[link_arrow] (cr2) -- (cr1);
    \draw[link_to_arrow] ([xshift=4pt,yshift=-11pt]cc4r.center) -- (cr2);
    \draw[link_to_arrow] ([xshift=-10pt]cc3r.center) -- (cr2);
  }
\end{scope}

\begin{scope}[shift={(0,-2*\SFkey{grid y})}]
  \SADRASubfigure{g}{Inter-Node Delegation -- Compute Nodes.}{fig:tree-7}{
    \node[cap_node, color_compute]   (cc1) at \SFcoord{g/cc1} {\SadraCapLabel{c_c^1}};
    \node[cap_node, color_indicator_p] (cp5) at \SFcoord{g/cp5} {\SadraCapLabel{c_p^5}};
    \node[cap_node, color_compute]   (cc4) at \SFcoord{g/cc4} {\SadraCapLabel{c_c^4}};
    \node[cap_node, color_resource]  (cr1) at \SFcoord{g/cr1} {\SadraCapLabel{c_r^1}};
    \node[cap_node, color_resource]  (cr2) at \SFcoord{g/cr2} {\SadraCapLabel{c_r^2}};
    \node[cap_node, color_compute]   (cc3) at \SFcoord{g/cc3} {\SadraCapLabel{c_c^3}};
    \node[cap_node, color_process]   (cp4) at \SFcoord{g/cp4} {\SadraCapLabel{c_p^4}};
    \draw[link_arrow] (cc1) -- (Rck);
    \draw[link_arrow] (cc4) -- (cc1);
    \draw[link_to_arrow] ([xshift=-4.5pt,yshift=-10pt]cp5.center) -- (cc4);
    \draw[link_arrow] (cr1) -- (Rrj);
    \draw[link_arrow] (cr2) -- (cr1);
    \draw[link_arrow] (cc3) -- (Rckp);
    \draw[link_to_arrow] ([xshift=-10pt,yshift=4.5pt]cp4.center) -- (cc3);
  }
\end{scope}

\begin{scope}[shift={(\SFkey{grid x},-2*\SFkey{grid y})}]
  \SADRASubfigure{h}{Inter-Node Revocation -- Compute Nodes.}{fig:tree-8}{
    \node[cap_node, color_compute]   (cc1) at \SFcoord{h/cc1} {\SadraCapLabel{c_c^1}};
    \node[cap_node, color_indicator_p] (cp5) at \SFcoord{h/cp5} {\SadraCapLabel{c_p^5}};
    \node[cap_node, color_compute]   (cc4) at \SFcoord{h/cc4} {\SadraCapLabel{c_c^4}};
    \node[cap_node, color_resource]  (cr1) at \SFcoord{h/cr1} {\SadraCapLabel{c_r^1}};
    \node[cap_node, color_resource]  (cr2) at \SFcoord{h/cr2} {\SadraCapLabel{c_r^2}};
    \node[cap_node, color_compute]   (cc3) at \SFcoord{h/cc3} {\SadraCapLabel{c_c^3}};
    \draw[link_arrow] (cc1) -- (Rck);
    \draw[link_to_arrow] ([xshift=-4.5pt,yshift=-10pt]cp5.center) -- (cc4);
    \draw[link_arrow] (cr1) -- (Rrj);
    \draw[link_arrow] (cr2) -- (cr1);
    \draw[link_arrow] (cc3) -- (Rckp);
    \node[revoked_x] at (cc4) {};
  }
\end{scope}

\begin{scope}[shift={(2*\SFkey{grid x},-2*\SFkey{grid y})}]
  \SADRASubfigure{i}{Inter-Node Revocation -- Resource Node.}{fig:tree-9}{
    \node[cap_node, color_compute]   (cc1)  at \SFcoord{i/cc1}  {\SadraCapLabel{c_c^1}};
    \node[cap_node, color_resource]  (cr1)  at \SFcoord{i/cr1}  {\SadraCapLabel{c_r^1}};
    \node[cap_node, color_resource]  (cr2)  at \SFcoord{i/cr2}  {\SadraCapLabel{c_r^2}};
    \node[cap_node, color_compute]   (cc3)  at \SFcoord{i/cc3}  {\SadraCapLabel{c_c^3}};
    \node[cap_node, color_indicator_c] (cc4r) at \SFcoord{i/cc4r} {\SadraCapLabel{c_c^4}};
    \draw[link_arrow] (cc1) -- (Rck);
    \draw[link_arrow] (cr1) -- (Rrj);
    \draw[link_arrow] (cc3) -- (Rckp);
    \draw[link_to_arrow] ([xshift=4pt,yshift=-11pt]cc4r.center) -- (cr2);
    \node[revoked_x] at (cr2) {};
  }
\end{scope}

\end{tikzpicture}%
}

\caption{Capability Trees. Legend: \protect\rcapcircle{} = resource capability; \protect\icapcircle{} = compute capability; \protect\pcapcircle{} = process capability; \\
  \protect\solidcircle{} = indicator flag is unset (usable capability); \protect\dashedcircle{} = indicator flag is set (revocation handle);
  \protect\circledashedarrow{}: link-to connection in our capability-object model; 
  \protect\captreedarrow{}: child-parent relation in a capability tree.}
\label{fig:cap-trees1}
\end{figure*}

%% file: figure_impl.tex

\newcommand{\xMemNode}{0}         
\newcommand{\xCompOne}{3.7}       
\newcommand{\xCompTwo}{7.4}         

\newcommand{\yApp}{9.95}          
\newcommand{\yVfs}{8.7}           
\newcommand{\yQdma}{7.45}          
\newcommand{\ySadra}{6.85}         
\newcommand{\yStack}{6.25}         
\newcommand{\ySwitch}{4.75}        

\newcommand{\wBox}{2.7}           
\newcommand{\hBox}{0.6}           
\newcommand{\wSwitch}{10.5}       
\newcommand{\hSwitch}{0.6}        

\newcommand{\wFpga}{3}          
\newcommand{\hFpga}{2.6}          
\newcommand{\yFpgaMid}{6.8}      
\newcommand{\offFpgaX}{-0.03}      
\newcommand{\offFpgaY}{0.03}      

\newcommand{\wCont}{3.3}          
\newcommand{\hCont}{4.95}          
\newcommand{\yContMid}{7.88}      
\newcommand{\offNodeLabelY}{-0.07}  

\newcommand{\szNodeHeader}{9}     
\newcommand{\szBoxLabel}{9}       
\newcommand{\szFpgaLabel}{7}      
\newcommand{\szSwitchLabel}{10}   

\newcommand{\fsNodeHeader}{\fontsize{\szNodeHeader}{\szNodeHeader}\selectfont}
\newcommand{\fsBoxLabel}{\fontsize{\szBoxLabel}{\szBoxLabel}\selectfont}
\newcommand{\fsFpgaLabel}{\fontsize{\szFpgaLabel}{\szFpgaLabel}\selectfont}
\newcommand{\fsSwitchLabel}{\fontsize{\szSwitchLabel}{\szSwitchLabel}\selectfont}

\providecommand{\boxWeight}{0.7pt}    
\providecommand{\arrowWeight}{0.6pt}  
\providecommand{\dottedWeight}{0.8pt} 


\begin{figure}[t]
\centering
\begin{tikzpicture}[
    scale=0.74, every node/.style={transform shape},
    font=\sffamily\bfseries,
    box/.style={rectangle, draw, line width=\boxWeight, minimum width=\wBox cm, minimum height=\hBox cm, align=center, font=\fsBoxLabel},
    node_container/.style={rectangle, draw, dashed, line width=\dottedWeight, minimum width=\wCont cm, minimum height=\hCont cm},
    fpga_card/.style={rectangle, draw, fill=blue!7.5, line width=\boxWeight, minimum width=\wFpga cm, minimum height=\hFpga cm},
    arrow/.style={-{Stealth[scale=0.8]}, line width=\arrowWeight},
    bi_arrow/.style={{Stealth[scale=0.8]}-{Stealth[scale=0.8]}, line width=\arrowWeight},
    fabric_arrow/.style={{Stealth[scale=0.9]}-{Stealth[scale=0.9]}, line width=\arrowWeight, cyan!80!black}
]

\node[box, fill=yellow!15] (mem) at (\xMemNode, \yApp) {Memory Module};
\node[box, fill=purple!20] (rw)  at (\xMemNode, \yQdma) {Read / Write Ctrl.};
\node[box, fill=orange!30] (rc)  at (\xMemNode, \ySadra) {SADRA (RC)};
\node[box, fill=teal!20]   (tm)  at (\xMemNode, \yStack) {TCP/IP Stack};

\begin{scope}[on background layer]
    \node[fpga_card] (fpgam) at (\xMemNode, \yFpgaMid) {};
    \node[anchor=south west, font=\bfseries\fsFpgaLabel] at ([xshift=\offFpgaX cm, yshift=\offFpgaY cm]fpgam.south west) {FPGA Card};
    
    \node[node_container] (mn) at (\xMemNode, \yContMid) {};
    \node[above, font=\fsNodeHeader] at ([yshift=\offNodeLabelY cm]mn.north) {Memory Node};
\end{scope}

\foreach \x/\idx in {\xCompOne/1, \xCompTwo/2} {
    \node[box, fill=green!15]  (p\idx) at (\x, \yApp) {Podman};
    \node[box, fill=orange!15] (v\idx) at (\x, \yVfs) {SR-IOV VFs};
    \node[box, fill=blue!30]   (q\idx) at (\x, \yQdma) {QDMA};
    \node[box, fill=orange!30] (s\idx) at (\x, \ySadra) {SADRA (CC)};
    \node[box, fill=teal!20]   (t\idx) at (\x, \yStack) {TCP/IP Stack};

    \begin{scope}[on background layer]
        \node[fpga_card] (fpga\idx) at (\x, \yFpgaMid) {};
        \node[anchor=south west, font=\bfseries\fsFpgaLabel] at ([xshift=\offFpgaX cm, yshift=\offFpgaY cm]fpga\idx.south west) {FPGA Card};
        
        \node[node_container] (cn\idx) at (\x, \yContMid) {};
        \node[above, font=\fsNodeHeader] at ([yshift=\offNodeLabelY cm]cn\idx.north) {Compute Node \idx};
    \end{scope}

    \draw[bi_arrow] (p\idx.south) -- (v\idx.north);
    \draw[bi_arrow] (v\idx.south) -- (q\idx.north);
}

\draw[bi_arrow] (mem.south) -- (rw.north);
\node[below=0.1cm] at (rc.north) {}; 

\node[box, fill=gray!70, text=white, minimum width=\wSwitch cm, minimum height=\hSwitch cm, font=\fsSwitchLabel] 
    (switch) at ({(\xMemNode+\xCompTwo)/2}, \ySwitch) {100 GbE Optical Switch};

\foreach \t in {tm, t1, t2} {
    \draw[fabric_arrow] (\t.south) -- (\t.south |- switch.north);
}

\end{tikzpicture}
\caption{Prototype deployment of \SADRA{} in a memory-disaggregated architecture with two compute controllers, one resource controller, and a 100~GbE switch.}
\label{fig:impl}
\end{figure}
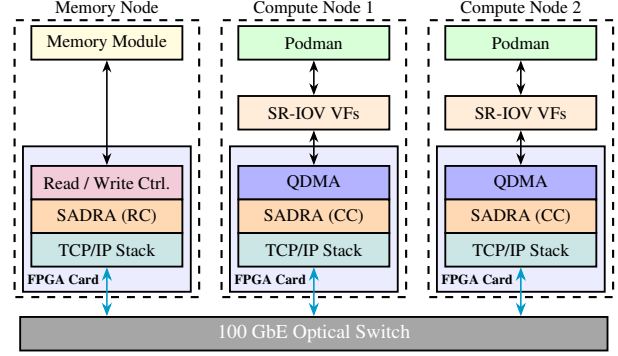

%% file: figure_throughput.tex
\definecolor{aaeWOne}{HTML}{0072B2}
\definecolor{aaeWTwo}{HTML}{E69F00}
\definecolor{aaeWFour}{HTML}{009E73}
\definecolor{aaeWEight}{HTML}{D55E00}
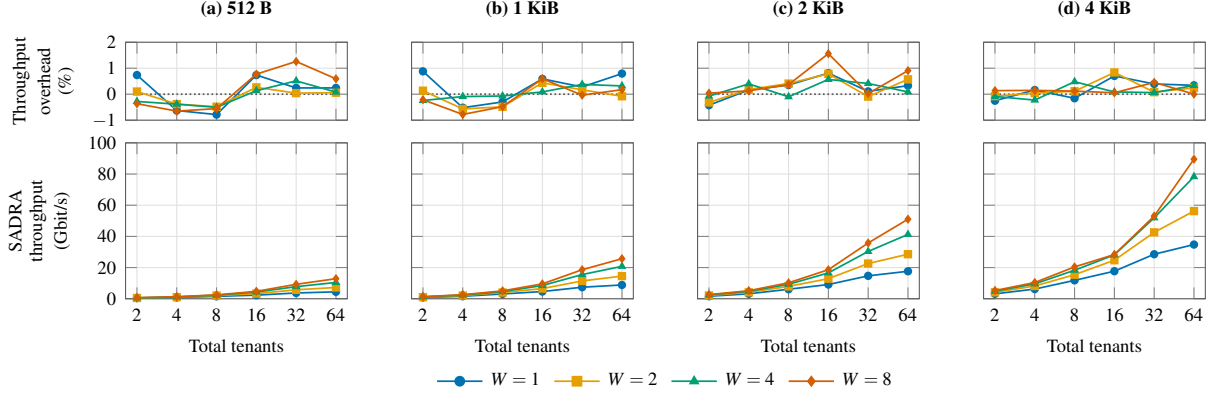
\begin{figure*}[t]
\centering
\begin{tikzpicture}
\begin{groupplot}[
group style={group size=4 by 2,horizontal sep=0.85cm,vertical sep=0.3cm},
width=0.165\textwidth,
height=0.116\textwidth,
scale only axis,
xmode=log,
log basis x=2,
xtick={1,2,4,8,16,32},
xticklabels={2,4,8,16,32,64},
xmin=0.82,
xmax=39,
grid=major,
grid style={draw=black!12},
axis line style={black!65},
tick align=inside,
tick label style={font=\scriptsize},
label style={font=\scriptsize},
title style={font=\scriptsize\bfseries},
every axis plot/.append style={line width=0.65pt,mark size=1.25pt},
clip mode=individual,
]

\nextgroupplot[title={(a) 512 B},height=0.058\textwidth,ymin=-1,ymax=2,ytick={-1,0,1,2},xticklabels=\empty,ylabel={\shortstack{Throughput \\overhead\\(\%)}},legend columns=-1,legend cell align={left},legend style={at={(2.44,-3.10)},anchor=north,draw=none,font=\scriptsize,/tikz/every even column/.append style={column sep=3pt},/tikz/every odd column/.append style={column sep=1pt}}]
\addlegendimage{aaeWOne,solid,mark=*,mark size=1.7pt}
\addlegendentry{$W=1$}
\addlegendimage{aaeWTwo,solid,mark=square*,mark size=1.7pt}
\addlegendentry{$W=2$}
\addlegendimage{aaeWFour,solid,mark=triangle*,mark size=1.7pt}
\addlegendentry{$W=4$}
\addlegendimage{aaeWEight,solid,mark=diamond*,mark size=1.7pt}
\addlegendentry{$W=8$}
\addplot+[black!65,densely dotted,no marks,forget plot] coordinates {(1,0) (32,0)};
\addplot+[aaeWOne,solid,mark=*,mark options={fill=aaeWOne},forget plot] coordinates {(1,0.736964) (2,-0.628525) (4,-0.784916) (8,0.735841) (16,0.241845) (32,0.240368)};
\addplot+[aaeWTwo,solid,mark=square*,mark options={fill=aaeWTwo},forget plot] coordinates {(1,0.102288) (2,-0.381736) (4,-0.480434) (8,0.263064) (16,0.041425) (32,0.056268)};
\addplot+[aaeWFour,solid,mark=triangle*,mark options={fill=aaeWFour},forget plot] coordinates {(1,-0.278937) (2,-0.384029) (4,-0.506653) (8,0.134344) (16,0.513185) (32,0.080014)};
\addplot+[aaeWEight,solid,mark=diamond*,mark options={fill=aaeWEight},forget plot] coordinates {(1,-0.366646) (2,-0.650771) (4,-0.555103) (8,0.770807) (16,1.259390) (32,0.592451)};

\nextgroupplot[title={(b) 1 KiB},height=0.058\textwidth,ymin=-1,ymax=2,ytick={-1,0,1,2},xticklabels=\empty,yticklabels=\empty]
\addplot+[black!65,densely dotted,no marks,forget plot] coordinates {(1,0) (32,0)};
\addplot+[aaeWOne,solid,mark=*,mark options={fill=aaeWOne},forget plot] coordinates {(1,0.879161) (2,-0.519978) (4,-0.297783) (8,0.589082) (16,0.277641) (32,0.793182)};
\addplot+[aaeWTwo,solid,mark=square*,mark options={fill=aaeWTwo},forget plot] coordinates {(1,0.135054) (2,-0.566469) (4,-0.485363) (8,0.438987) (16,0.169029) (32,-0.079229)};
\addplot+[aaeWFour,solid,mark=triangle*,mark options={fill=aaeWFour},forget plot] coordinates {(1,-0.242647) (2,-0.090978) (4,-0.069701) (8,0.088135) (16,0.376878) (32,0.318335)};
\addplot+[aaeWEight,solid,mark=diamond*,mark options={fill=aaeWEight},forget plot] coordinates {(1,-0.213045) (2,-0.779442) (4,-0.480018) (8,0.594983) (16,-0.037909) (32,0.173787)};

\nextgroupplot[title={(c) 2 KiB},height=0.058\textwidth,ymin=-1,ymax=2,ytick={-1,0,1,2},xticklabels=\empty,yticklabels=\empty]
\addplot+[black!65,densely dotted,no marks,forget plot] coordinates {(1,0) (32,0)};
\addplot+[aaeWOne,solid,mark=*,mark options={fill=aaeWOne},forget plot] coordinates {(1,-0.424295) (2,0.164638) (4,0.352258) (8,0.811308) (16,0.108812) (32,0.328401)};
\addplot+[aaeWTwo,solid,mark=square*,mark options={fill=aaeWTwo},forget plot] coordinates {(1,-0.317020) (2,0.182787) (4,0.408708) (8,0.785093) (16,-0.092835) (32,0.572363)};
\addplot+[aaeWFour,solid,mark=triangle*,mark options={fill=aaeWFour},forget plot] coordinates {(1,-0.074835) (2,0.397076) (4,-0.100323) (8,0.567514) (16,0.411552) (32,0.088489)};
\addplot+[aaeWEight,solid,mark=diamond*,mark options={fill=aaeWEight},forget plot] coordinates {(1,0.039371) (2,0.128316) (4,0.367611) (8,1.560232) (16,0.024044) (32,0.905971)};

\nextgroupplot[title={(d) 4 KiB},height=0.058\textwidth,ymin=-1,ymax=2,ytick={-1,0,1,2},xticklabels=\empty,yticklabels=\empty]
\addplot+[black!65,densely dotted,no marks,forget plot] coordinates {(1,0) (32,0)};
\addplot+[aaeWOne,solid,mark=*,mark options={fill=aaeWOne},forget plot] coordinates {(1,-0.251435) (2,0.160939) (4,-0.157020) (8,0.702150) (16,0.393767) (32,0.341536)};
\addplot+[aaeWTwo,solid,mark=square*,mark options={fill=aaeWTwo},forget plot] coordinates {(1,-0.074770) (2,0.040884) (4,0.112585) (8,0.837160) (16,0.078064) (32,0.243224)};
\addplot+[aaeWFour,solid,mark=triangle*,mark options={fill=aaeWFour},forget plot] coordinates {(1,-0.087645) (2,-0.226379) (4,0.480243) (8,0.086086) (16,0.061370) (32,0.334699)};
\addplot+[aaeWEight,solid,mark=diamond*,mark options={fill=aaeWEight},forget plot] coordinates {(1,0.136586) (2,0.147226) (4,0.118921) (8,0.054061) (16,0.454377) (32,0.002956)};

\nextgroupplot[ymin=0,ymax=100,ytick={0,20,40,60,80,100},xlabel={Total tenants},ylabel={\shortstack{SADRA\\throughput\\(Gbit/s)}}]
\addplot+[aaeWOne,solid,mark=*,mark options={fill=aaeWOne},forget plot] coordinates {(1,0.404075) (2,0.799286) (4,1.550541) (8,2.324578) (16,3.705523) (32,4.450488)};
\addplot+[aaeWTwo,solid,mark=square*,mark options={fill=aaeWTwo},forget plot] coordinates {(1,0.522171) (2,1.034664) (4,2.008350) (8,3.304300) (16,5.743809) (32,7.247472)};
\addplot+[aaeWFour,solid,mark=triangle*,mark options={fill=aaeWFour},forget plot] coordinates {(1,0.609479) (2,1.194934) (4,2.295249) (8,4.291151) (16,7.760533) (32,10.449349)};
\addplot+[aaeWEight,solid,mark=diamond*,mark options={fill=aaeWEight},forget plot] coordinates {(1,0.661637) (2,1.288607) (4,2.499222) (8,4.742566) (16,9.268809) (32,12.870198)};

\nextgroupplot[ymin=0,ymax=100,ytick={0,20,40,60,80,100},xlabel={Total tenants},yticklabels=\empty]
\addplot+[aaeWOne,solid,mark=*,mark options={fill=aaeWOne},forget plot] coordinates {(1,0.806163) (2,1.599623) (4,3.085335) (8,4.626438) (16,7.401212) (32,8.845011)};
\addplot+[aaeWTwo,solid,mark=square*,mark options={fill=aaeWTwo},forget plot] coordinates {(1,1.043113) (2,2.068902) (4,4.004401) (8,6.539766) (16,11.430372) (32,14.526041)};
\addplot+[aaeWFour,solid,mark=triangle*,mark options={fill=aaeWFour},forget plot] coordinates {(1,1.217470) (2,2.385159) (4,4.581328) (8,8.484415) (16,15.520339) (32,20.784515)};
\addplot+[aaeWEight,solid,mark=diamond*,mark options={fill=aaeWEight},forget plot] coordinates {(1,1.322566) (2,2.582748) (4,5.049625) (8,9.541981) (16,18.625211) (32,25.688474)};

\nextgroupplot[ymin=0,ymax=100,ytick={0,20,40,60,80,100},xlabel={Total tenants},yticklabels=\empty]
\addplot+[aaeWOne,solid,mark=*,mark options={fill=aaeWOne},forget plot] coordinates {(1,1.619240) (2,3.167194) (4,6.103192) (8,9.132163) (16,14.698097) (32,17.629543)};
\addplot+[aaeWTwo,solid,mark=square*,mark options={fill=aaeWTwo},forget plot] coordinates {(1,2.084165) (2,4.126042) (4,7.944972) (8,12.876266) (16,22.584703) (32,28.580584)};
\addplot+[aaeWFour,solid,mark=triangle*,mark options={fill=aaeWFour},forget plot] coordinates {(1,2.433403) (2,4.763723) (4,9.263396) (8,16.589191) (16,30.348477) (32,41.263701)};
\addplot+[aaeWEight,solid,mark=diamond*,mark options={fill=aaeWEight},forget plot] coordinates {(1,2.641330) (2,5.161870) (4,10.242288) (8,18.623208) (16,35.772089) (32,51.023818)};

\nextgroupplot[ymin=0,ymax=100,ytick={0,20,40,60,80,100},xlabel={Total tenants},yticklabels=\empty]
\addplot+[aaeWOne,solid,mark=*,mark options={fill=aaeWOne},forget plot] coordinates {(1,3.158100) (2,6.212840) (4,11.802967) (8,17.681812) (16,28.554928) (32,34.770833)};
\addplot+[aaeWTwo,solid,mark=square*,mark options={fill=aaeWTwo},forget plot] coordinates {(1,4.145584) (2,8.183178) (4,15.422845) (8,24.732946) (16,42.587259) (32,56.233338)};
\addplot+[aaeWFour,solid,mark=triangle*,mark options={fill=aaeWFour},forget plot] coordinates {(1,4.843863) (2,9.544695) (4,18.299052) (8,28.146569) (16,51.929286) (32,78.296909)};
\addplot+[aaeWEight,solid,mark=diamond*,mark options={fill=aaeWEight},forget plot] coordinates {(1,5.273206) (2,10.409466) (4,20.433483) (8,28.350087) (16,53.005410) (32,89.542011)};
\end{groupplot}
\end{tikzpicture}
\caption{Throughput and overhead across payload sizes. Top: mean signed throughput overhead relative to baseline; negative values indicate that SADRA measured faster. Bottom: achieved SADRA throughput.}
\label{fig:aae_throughput}
\end{figure*}

%% file: figure_latency.tex
\definecolor{aaeWOne}{HTML}{0072B2}
\definecolor{aaeWTwo}{HTML}{E69F00}
\definecolor{aaeWFour}{HTML}{009E73}
\definecolor{aaeWEight}{HTML}{D55E00}
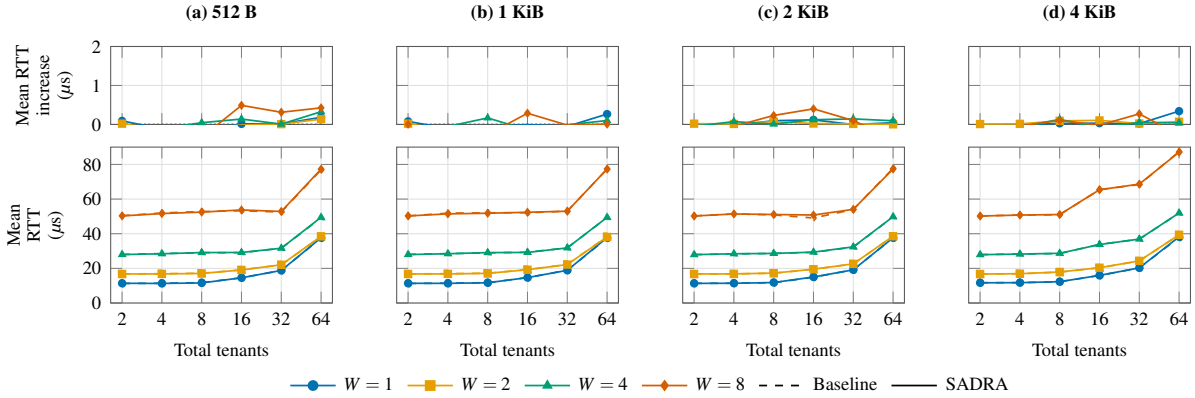
\begin{figure*}[t]
\centering
\begin{tikzpicture}
\begin{groupplot}[
group style={group size=4 by 2,horizontal sep=0.85cm,vertical sep=0.3cm},
width=0.165\textwidth,
height=0.116\textwidth,
scale only axis,
xmode=log,
log basis x=2,
xtick={1,2,4,8,16,32},
xticklabels={2,4,8,16,32,64},
xmin=0.82,
xmax=39,
grid=major,
grid style={draw=black!12},
axis line style={black!65},
tick align=inside,
tick label style={font=\scriptsize},
label style={font=\scriptsize},
title style={font=\scriptsize\bfseries},
every axis plot/.append style={line width=0.65pt,mark size=1.25pt},
clip mode=individual,
]
\nextgroupplot[title={(a) 512 B},height=0.058\textwidth,ymin=0,ymax=2,ytick={0,1,2},xticklabels=\empty,ylabel={\shortstack{Mean RTT\\increase\\($\mu$s)}},legend columns=-1,legend cell align={left},legend style={at={(2.44,-3.10)},anchor=north,draw=none,font=\scriptsize,/tikz/every even column/.append style={column sep=3pt},/tikz/every odd column/.append style={column sep=1pt}}]
\addlegendimage{aaeWOne,solid,mark=*,mark size=1.7pt}
\addlegendentry{$W=1$}
\addlegendimage{aaeWTwo,solid,mark=square*,mark size=1.7pt}
\addlegendentry{$W=2$}
\addlegendimage{aaeWFour,solid,mark=triangle*,mark size=1.7pt}
\addlegendentry{$W=4$}
\addlegendimage{aaeWEight,solid,mark=diamond*,mark size=1.7pt}
\addlegendentry{$W=8$}
\addlegendimage{black,dashed,no marks}
\addlegendentry{Baseline}
\addlegendimage{black,solid,no marks}
\addlegendentry{SADRA}

\addplot+[black!65,densely dotted,no marks,forget plot] coordinates {(1,0) (32,0)};
\addplot+[aaeWOne,solid,mark=*,mark options={fill=aaeWOne},forget plot] coordinates {(1,0.091314) (2,-0.115837) (4,-0.168525) (8,0.013655) (16,0.014425) (32,0.182612)};
\addplot+[aaeWTwo,solid,mark=square*,mark options={fill=aaeWTwo},forget plot] coordinates {(1,0.020104) (2,-0.127174) (4,-0.117956) (8,-0.016375) (16,0.009022) (32,0.136595)};
\addplot+[aaeWFour,solid,mark=triangle*,mark options={fill=aaeWFour},forget plot] coordinates {(1,-0.054037) (2,-0.105839) (4,0.044791) (8,0.138293) (16,0.008273) (32,0.326474)};
\addplot+[aaeWEight,solid,mark=diamond*,mark options={fill=aaeWEight},forget plot] coordinates {(1,-0.070360) (2,-0.337821) (4,-0.331730) (8,0.489520) (16,0.312905) (32,0.425022)};
\nextgroupplot[title={(b) 1 KiB},height=0.058\textwidth,ymin=0,ymax=2,ytick={0,1,2},xticklabels=\empty,yticklabels=\empty]
\addplot+[black!65,densely dotted,no marks,forget plot] coordinates {(1,0) (32,0)};
\addplot+[aaeWOne,solid,mark=*,mark options={fill=aaeWOne},forget plot] coordinates {(1,0.079263) (2,-0.106728) (4,-0.062114) (8,-0.004848) (16,-0.049608) (32,0.265856)};
\addplot+[aaeWTwo,solid,mark=square*,mark options={fill=aaeWTwo},forget plot] coordinates {(1,0.019219) (2,-0.140021) (4,-0.074862) (8,-0.067522) (16,-0.027142) (32,-0.203038)};
\addplot+[aaeWFour,solid,mark=triangle*,mark options={fill=aaeWFour},forget plot] coordinates {(1,-0.096939) (2,-0.051686) (4,0.166912) (8,-0.140696) (16,-0.009861) (32,0.100124)};
\addplot+[aaeWEight,solid,mark=diamond*,mark options={fill=aaeWEight},forget plot] coordinates {(1,0.004336) (2,-0.397447) (4,-0.261278) (8,0.287068) (16,-0.018221) (32,0.008201)};
\nextgroupplot[title={(c) 2 KiB},height=0.058\textwidth,ymin=0,ymax=2,ytick={0,1,2},xticklabels=\empty,yticklabels=\empty]
\addplot+[black!65,densely dotted,no marks,forget plot] coordinates {(1,0) (32,0)};
\addplot+[aaeWOne,solid,mark=*,mark options={fill=aaeWOne},forget plot] coordinates {(1,-0.016195) (2,-0.023150) (4,0.095931) (8,0.116873) (16,-0.002335) (32,0.044376)};
\addplot+[aaeWTwo,solid,mark=square*,mark options={fill=aaeWTwo},forget plot] coordinates {(1,0.019823) (2,0.013007) (4,0.074117) (8,0.023943) (16,0.013643) (32,0.003057)};
\addplot+[aaeWFour,solid,mark=triangle*,mark options={fill=aaeWFour},forget plot] coordinates {(1,-0.048252) (2,0.072642) (4,0.008353) (8,0.120760) (16,0.142354) (32,0.094235)};
\addplot+[aaeWEight,solid,mark=diamond*,mark options={fill=aaeWEight},forget plot] coordinates {(1,-0.042819) (2,-0.048845) (4,0.230405) (8,0.4) (16,0.099660) (32,-0.373024)};
\nextgroupplot[title={(d) 4 KiB},height=0.058\textwidth,ymin=0,ymax=2,ytick={0,1,2},xticklabels=\empty,yticklabels=\empty]
\addplot+[black!65,densely dotted,no marks,forget plot] coordinates {(1,0) (32,0)};
\addplot+[aaeWOne,solid,mark=*,mark options={fill=aaeWOne},forget plot] coordinates {(1,-0.043523) (2,-0.013836) (4,0.027209) (8,0.032077) (16,0.020865) (32,0.340951)};
\addplot+[aaeWTwo,solid,mark=square*,mark options={fill=aaeWTwo},forget plot] coordinates {(1,0.009948) (2,0.014785) (4,0.091961) (8,0.103845) (16,0.023823) (32,0.068870)};
\addplot+[aaeWFour,solid,mark=triangle*,mark options={fill=aaeWFour},forget plot] coordinates {(1,-0.037283) (2,-0.065120) (4,0.128702) (8,-0.079485) (16,0.050898) (32,0.041910)};
\addplot+[aaeWEight,solid,mark=diamond*,mark options={fill=aaeWEight},forget plot] coordinates {(1,-0.050270) (2,-0.041539) (4,0.103254) (8,-0.019929) (16,0.271187) (32,-0.270142)};

\nextgroupplot[ymin=0,ymax=90.0,xlabel={Total tenants},ylabel={\shortstack{Mean\\RTT\\($\mu$s)}}]
\addplot+[aaeWOne,dashed,mark=*,mark options={fill=white},forget plot] coordinates {(1,11.290784) (2,11.421808) (4,11.749284) (8,14.544440) (16,18.782297) (32,37.500188)};
\addplot+[aaeWTwo,dashed,mark=square*,mark options={fill=white},forget plot] coordinates {(1,16.764330) (2,16.932984) (4,17.180482) (8,19.096035) (16,22.056403) (32,38.363600)};
\addplot+[aaeWFour,dashed,mark=triangle*,mark options={fill=white},forget plot] coordinates {(1,28.025569) (2,28.534086) (4,29.071544) (8,29.039531) (16,31.607534) (32,49.157335)};
\addplot+[aaeWEight,dashed,mark=diamond*,mark options={fill=white},forget plot] coordinates {(1,50.361853) (2,51.978886) (4,52.809540) (8,53.233644) (16,52.534113) (32,76.851557)};
\addplot+[aaeWOne,solid,mark=*,mark options={fill=aaeWOne},forget plot] coordinates {(1,11.382098) (2,11.305971) (4,11.580759) (8,14.558095) (16,18.796723) (32,37.682800)};
\addplot+[aaeWTwo,solid,mark=square*,mark options={fill=aaeWTwo},forget plot] coordinates {(1,16.784434) (2,16.805810) (4,17.062526) (8,19.079660) (16,22.065426) (32,38.500195)};
\addplot+[aaeWFour,solid,mark=triangle*,mark options={fill=aaeWFour},forget plot] coordinates {(1,27.971533) (2,28.428246) (4,29.116335) (8,29.177824) (16,31.615807) (32,49.483808)};
\addplot+[aaeWEight,solid,mark=diamond*,mark options={fill=aaeWEight},forget plot] coordinates {(1,50.291493) (2,51.641065) (4,52.477810) (8,53.723164) (16,52.847018) (32,77.276579)};

\nextgroupplot[ymin=0,ymax=90.0,xlabel={Total tenants},yticklabels=\empty]
\addplot+[aaeWOne,dashed,mark=*,mark options={fill=white},forget plot] coordinates {(1,11.298455) (2,11.452841) (4,11.740923) (8,14.671090) (16,18.925817) (32,37.473912)};
\addplot+[aaeWTwo,dashed,mark=square*,mark options={fill=white},forget plot] coordinates {(1,16.761895) (2,16.943805) (4,17.198037) (8,19.297708) (16,22.228488) (32,38.335511)};
\addplot+[aaeWFour,dashed,mark=triangle*,mark options={fill=white},forget plot] coordinates {(1,28.068933) (2,28.500088) (4,28.931914) (8,29.338547) (16,31.822936) (32,49.340667)};
\addplot+[aaeWEight,dashed,mark=diamond*,mark options={fill=white},forget plot] coordinates {(1,50.270159) (2,51.861494) (4,52.079993) (8,52.101608) (16,53.014827) (32,77.323500)};
\addplot+[aaeWOne,solid,mark=*,mark options={fill=aaeWOne},forget plot] coordinates {(1,11.377718) (2,11.346112) (4,11.678809) (8,14.666243) (16,18.876209) (32,37.739768)};
\addplot+[aaeWTwo,solid,mark=square*,mark options={fill=aaeWTwo},forget plot] coordinates {(1,16.781114) (2,16.803784) (4,17.123175) (8,19.230186) (16,22.201346) (32,38.132473)};
\addplot+[aaeWFour,solid,mark=triangle*,mark options={fill=aaeWFour},forget plot] coordinates {(1,27.971994) (2,28.448401) (4,29.098826) (8,29.197851) (16,31.813074) (32,49.440791)};
\addplot+[aaeWEight,solid,mark=diamond*,mark options={fill=aaeWEight},forget plot] coordinates {(1,50.274495) (2,51.464047) (4,51.818715) (8,52.388677) (16,52.996606) (32,77.331701)};

\nextgroupplot[ymin=0,ymax=90.0,xlabel={Total tenants},yticklabels=\empty]
\addplot+[aaeWOne,dashed,mark=*,mark options={fill=white},forget plot] coordinates {(1,11.355033) (2,11.427035) (4,11.727854) (8,14.897025) (16,19.160907) (32,37.607292)};
\addplot+[aaeWTwo,dashed,mark=square*,mark options={fill=white},forget plot] coordinates {(1,16.753411) (2,16.804755) (4,17.188997) (8,19.488385) (16,22.578069) (32,38.549866)};
\addplot+[aaeWFour,dashed,mark=triangle*,mark options={fill=white},forget plot] coordinates {(1,27.960521) (2,28.357036) (4,28.617981) (8,29.253061) (16,32.266975) (32,49.733285)};
\addplot+[aaeWEight,dashed,mark=diamond*,mark options={fill=white},forget plot] coordinates {(1,50.267388) (2,51.454597) (4,50.864875) (8,49.091129) (16,53.969667) (32,77.756798)};
\addplot+[aaeWOne,solid,mark=*,mark options={fill=aaeWOne},forget plot] coordinates {(1,11.338839) (2,11.403885) (4,11.823786) (8,15.013897) (16,19.158571) (32,37.651668)};
\addplot+[aaeWTwo,solid,mark=square*,mark options={fill=aaeWTwo},forget plot] coordinates {(1,16.773234) (2,16.817763) (4,17.263114) (8,19.512328) (16,22.591712) (32,38.552923)};
\addplot+[aaeWFour,solid,mark=triangle*,mark options={fill=aaeWFour},forget plot] coordinates {(1,27.912269) (2,28.429678) (4,28.626334) (8,29.373821) (16,32.409329) (32,49.827520)};
\addplot+[aaeWEight,solid,mark=diamond*,mark options={fill=aaeWEight},forget plot] coordinates {(1,50.224570) (2,51.405752) (4,51.095280) (8,50.798725) (16,54.069327) (32,77.383775)};

\nextgroupplot[ymin=0,ymax=90.0,xlabel={Total tenants},yticklabels=\empty]
\addplot+[aaeWOne,dashed,mark=*,mark options={fill=white},forget plot] coordinates {(1,11.712073) (2,11.758081) (4,12.280680) (8,15.909329) (16,20.257359) (32,37.873915)};
\addplot+[aaeWTwo,dashed,mark=square*,mark options={fill=white},forget plot] coordinates {(1,16.743206) (2,16.895620) (4,17.826282) (8,20.278452) (16,24.291881) (32,39.328845)};
\addplot+[aaeWFour,dashed,mark=triangle*,mark options={fill=white},forget plot] coordinates {(1,27.944784) (2,28.303292) (4,28.561229) (8,33.839333) (16,36.856045) (32,51.911864)};
\addplot+[aaeWEight,dashed,mark=diamond*,mark options={fill=white},forget plot] coordinates {(1,50.259971) (2,50.808168) (4,50.964941) (8,65.414342) (16,68.365592) (32,87.363693)};
\addplot+[aaeWOne,solid,mark=*,mark options={fill=aaeWOne},forget plot] coordinates {(1,11.668550) (2,11.744245) (4,12.307888) (8,15.941406) (16,20.278224) (32,38.214866)};
\addplot+[aaeWTwo,solid,mark=square*,mark options={fill=aaeWTwo},forget plot] coordinates {(1,16.753154) (2,16.910405) (4,17.918243) (8,20.382297) (16,24.315704) (32,39.397715)};
\addplot+[aaeWFour,solid,mark=triangle*,mark options={fill=aaeWFour},forget plot] coordinates {(1,27.907501) (2,28.238171) (4,28.689931) (8,33.759848) (16,36.906943) (32,51.953774)};
\addplot+[aaeWEight,solid,mark=diamond*,mark options={fill=aaeWEight},forget plot] coordinates {(1,50.209701) (2,50.766629) (4,51.068195) (8,65.394413) (16,68.636778) (32,87.093551)};
\end{groupplot}
\end{tikzpicture}
\caption{Request RTT across payload sizes. Top: difference in mean RTT
between \SADRA{} and the paired baseline. Bottom: absolute mean RTT. Each point reports the mean across
three synchronized rounds; run-to-run variation is discussed in the text.}
\label{fig:aae_latency}
\end{figure*}

%% file: figure_ml.tex
\definecolor{mlWOne}{HTML}{0072B2}
\definecolor{mlWTwo}{HTML}{E69F00}
\definecolor{mlWFour}{HTML}{009E73}
\definecolor{mlWEight}{HTML}{D55E00}
\begin{figure}[t]
\centering
\begin{tikzpicture}
\begin{groupplot}[
group style={group size=1 by 3,vertical sep=.5cm},
width=0.76\columnwidth,
height=0.15\columnwidth,
scale only axis,
xmode=log,
log basis x=2,
xtick={2,4,8,16},
xticklabels={2,4,8,16},
xmin=1.739130,
xmax=18.400000,
grid=major,
grid style={draw=black!12},
axis line style={black!65},
tick align=inside,
tick label style={font=\scriptsize},
label style={font=\scriptsize},
every axis plot/.append style={line width=0.65pt,mark size=1.35pt},
clip mode=individual,
]

\nextgroupplot[height=0.1125\columnwidth,ymin=-2.000000,ymax=2.000000,ylabel={\shortstack{Throughput\\overhead\\(\%)}},xticklabels=\empty]
\addplot+[black!65,densely dotted,no marks,forget plot] coordinates {(2,0) (16,0)};
\addplot+[mlWOne,solid,mark=*,mark options={fill=mlWOne},forget plot] coordinates {(2,-0.320850) (4,-0.862953) (8,0.041281) (16,0.165864)};
\addplot+[mlWTwo,solid,mark=square*,mark options={fill=mlWTwo},forget plot] coordinates {(2,-0.416203) (4,-0.874102) (8,-1.131193) (16,0.276299)};
\addplot+[mlWFour,solid,mark=triangle*,mark options={fill=mlWFour},forget plot] coordinates {(2,-0.487318) (4,-1.526919) (8,-1.037871) (16,1.764024)};
\addplot+[mlWEight,solid,mark=diamond*,mark options={fill=mlWEight},forget plot] coordinates {(2,-0.638167) (4,-0.283245) (8,0.894257) (16,0.445688)};

\nextgroupplot[ymin=0,ymax=1.600000,ylabel={\shortstack{Authorized req.\\per second\\(millions)}},xticklabels=\empty]
\addplot+[mlWOne,solid,mark=*,mark options={fill=mlWOne},forget plot] coordinates {(2,0.174545) (4,0.344935) (8,0.680408) (16,1.412023)};
\addplot+[mlWTwo,solid,mark=square*,mark options={fill=mlWTwo},forget plot] coordinates {(2,0.174918) (4,0.345860) (8,0.678514) (16,1.422301)};
\addplot+[mlWFour,solid,mark=triangle*,mark options={fill=mlWFour},forget plot] coordinates {(2,0.174876) (4,0.340613) (8,0.654647) (16,1.275173)};
\addplot+[mlWEight,solid,mark=diamond*,mark options={fill=mlWEight},forget plot] coordinates {(2,0.172617) (4,0.319167) (8,0.626696) (16,1.282560)};

\nextgroupplot[ymin=0,ymax=70.000000,ylabel={\shortstack{SADRA\\mean RTT\\($\mu$s)}},xlabel={Total tenants},legend columns=-1,legend cell align={left},legend style={at={(0.5,-0.66)},anchor=north,draw=none,font=\scriptsize,/tikz/every even column/.append style={column sep=3pt},/tikz/every odd column/.append style={column sep=1pt}}]
\addplot+[mlWOne,solid,mark=*,mark options={fill=mlWOne},forget plot] coordinates {(2,10.458163) (4,10.589175) (8,10.692023) (16,10.008086)};
\addplot+[mlWTwo,solid,mark=square*,mark options={fill=mlWTwo},forget plot] coordinates {(2,16.176624) (4,16.326180) (8,16.606262) (16,15.117620)};
\addplot+[mlWFour,solid,mark=triangle*,mark options={fill=mlWFour},forget plot] coordinates {(2,27.722175) (4,28.398501) (8,29.412641) (16,29.291844)};
\addplot+[mlWEight,solid,mark=diamond*,mark options={fill=mlWEight},forget plot] coordinates {(2,50.988440) (4,55.101055) (8,55.859303) (16,53.159078)};
\addlegendimage{mlWOne,solid,mark=*,mark size=1.7pt}
\addlegendentry{$W=1$}
\addlegendimage{mlWTwo,solid,mark=square*,mark size=1.7pt}
\addlegendentry{$W=2$}
\addlegendimage{mlWFour,solid,mark=triangle*,mark size=1.7pt}
\addlegendentry{$W=4$}
\addlegendimage{mlWEight,solid,mark=diamond*,mark size=1.7pt}
\addlegendentry{$W=8$}
\end{groupplot}
\end{tikzpicture}
\caption{Meta-DLRM performance for batch size 32 on the test split.
Top: mean signed inference-throughput overhead relative to baseline; negative
values indicate that \SADRA{} measured faster. Middle: protected remote
request rate under \SADRA{}, including 26 embedding reads and one result write
per inference. Bottom: \SADRA{} mean remote-operation RTT.}
\label{fig:ml_performance}
\end{figure}
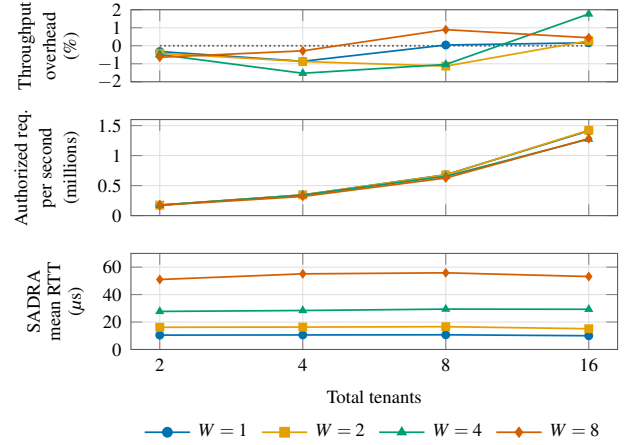

%% file: table_revocation.tex
\begin{table}[!t]
  \centering
  \caption{Capability Cleanup Time. Legend: C = Cycles.}
  \label{tab:revoc_result}
  {\fontsize{8pt}{9.5pt}\selectfont
  \setlength{\tabcolsep}{2.2pt}%
  \renewcommand{\arraystretch}{0.95}%

  \begin{tabular*}{\columnwidth}{@{\extracolsep{\fill}} l rrrrrr }
    \toprule
    \multicolumn{1}{l}{} &
    \multicolumn{6}{c}{\textbf{Capability Count}} \\
    \cmidrule(lr){2-7}
    \multicolumn{1}{c}{\textbf{No. of Sub-trees}} &
    \multicolumn{1}{r}{128} &
    \multicolumn{1}{r}{256} &
    \multicolumn{1}{r}{512} &
    \multicolumn{1}{r}{1024} &
    \multicolumn{1}{r}{2048} &
    \multicolumn{1}{r}{4096} \\
    \midrule
    One   & 132C & 260C & 516C & 1028C & 2052C & 4100C \\
    Two   & 132C & 260C & 516C & 1028C & 2052C & 4100C \\
    Four  & 132C & 260C & 516C & 1028C & 2052C & 4100C \\
    Eight & 132C & 260C & 516C & 1028C & 2052C & 4100C \\
    \bottomrule
  \end{tabular*}
  }

\end{table}

%% file: appendix.tex
\section{Formal Model}
\label{app:formal-model}
This appendix defines the formal model of \SADRA{} and the reasoning used to establish its security properties. The model captures capability relationships, delegation, validation, revocation, and the interaction between the compute and resource controllers. It underlies the four security properties established in \S\ref{sec:sec_properties_soundness}: capability safety, authority safety, revocation soundness, and strong isolation.            

The model represents authority through the linked capability chain $c_p \mapsto c_c \mapsto c_r$. Resource capabilities anchor authoritative resource rights at the resource controller, compute capabilities represent authority delegated to a compute controller, and process capabilities associate subsets of that authority with individual processes. Each controller maintains its local delegation state as a rooted capability tree. These linked trees determine how authority is allocated, attenuated, validated, delegated, and revoked.            

We model the compute and resource controllers as independent state machines operating at their respective enforcement boundaries. The compute controller validates process-granularity authority before fabric admission, while the resource controller makes the authoritative decision before resource access. An inter-node request is authorized exactly when both validation stages succeed. Revocation is represented through controller-local fences that cause dependent authority to fail validation before the corresponding capability state is reclaimed.             

\noindent\textbf{Basic notation.} Let $\mathbb{C} = \mathbb{C}_r \cup \mathbb{C}_c \cup \mathbb{C}_p$ denote the set of all capabilities. Let $x \in \{r,c\}$ identify the resource or compute controller, so that $\mathbb{C}_x$ and $\preceq_x$ denote the capability set and delegation relation associated with controller $x$. Finally, let $\mathbb{B} = \{\mathit{true},\mathit{false}\}$ denote the Boolean domain.

\subsection{Formal Definitions}
\label{app:formal-definitions}

We first define the principals, resources, permissions, capability representation, trusted request context, and capability authenticity used throughout the formal model.

\begin{definition}[Principals]
\label{def:principals}

Let
$
\mathbb{P}
$
denote the set of processes,
$
\mathbb{CC}
$
the set of compute controllers, and
$
\mathbb{RC}
$
the set of resource controllers. These sets are pairwise disjoint.

Each compute and resource controller has a globally unique node identifier. Let
$
\operatorname{NID}_c :
\mathbb{CC}\rightarrow\mathbb{N}_c
$
and
$
\operatorname{NID}_r :
\mathbb{RC}\rightarrow\mathbb{N}_r
$
be injective mappings assigning identifiers to compute and resource controllers, respectively.

The identifier value $0$ is reserved and is not in the range of either mapping.

\end{definition}

\begin{definition}[Resources]
\label{def:resources}

Let
$
\mathbb{R}
$
be the set of disaggregated resources managed by SADRA. Resource ownership is partitioned across resource controllers:
\[
\mathbb{R}
=
\biguplus_{rc\in\mathbb{RC}}
\mathbb{R}_{rc},
\]
where
$
\mathbb{R}_{rc}
$
is the resource extent managed by resource controller $rc$.

A capability may designate any subset of the resource extent managed by the resource controller from which its authority derives.

\end{definition}

\begin{definition}[Permissions]
\label{def:permissions}

The permission set is
\[
\mathbb{V}
=
\{
\mathtt{read},
\mathtt{write},
\mathtt{delegate},
\mathtt{exclusive}
\}.
\]

The permissions that directly authorize resource-access operations are
\[
\mathbb{V}_{\mathsf{acc}}
=
\{
\mathtt{read},
\mathtt{write}
\}.
\]

The permissions $\mathtt{delegate}$ and $\mathtt{exclusive}$ have different semantics.

The permission
$
\mathtt{delegate}
$
allows the holder of a non-exclusive capability to derive authority for another principal. A delegator may omit this permission from a delegated capability to prevent the recipient from subsequently re-delegating that authority.

The permission
$
\mathtt{exclusive}
$
records that strong isolation was requested when the resource was initially allocated. Exclusive authority cannot be delegated, including by the process that requested the allocation. Hence, $\mathtt{exclusive}$ overrides $\mathtt{delegate}$ if both are present.

The $\mathtt{exclusive}$ permission is introduced only by initial resource allocation. SADRA preserves it through the corresponding resource, root-anchored compute, and process capabilities. It cannot be introduced by a delegation operation.

Creating the process capability corresponding to an initially allocated root-anchored compute capability is not delegation and is therefore permitted even when that authority is exclusive.

\end{definition}

\begin{definition}[Capabilities]
\label{def:capability}

A capability is represented abstractly as
\[
c
=
\langle
s,r,v,t,
\mathit{rnid},
\mathit{cnid},
\mathit{rid},
\mathit{cid}
\rangle,
\]
where:

\begin{itemize}
    \item
    $s$ is the capability subject;

    \item
    $r\subseteq\mathbb{R}$ is the resource extent;

    \item
    $v\subseteq\mathbb{V}$ is the permission set;

    \item
    $
    t\in
    \{
    \mathtt{usable},
    \mathtt{handle}
    \}
    $
    is the capability type;

    \item
    $\mathit{rnid}$ identifies the resource node;

    \item
    $\mathit{cnid}$ identifies the compute node;

    \item
    $\mathit{rid}$ is an authority identifier allocated locally by the resource controller identified by $\mathit{rnid}$;

    \item
    $\mathit{cid}$ is an authority identifier allocated locally by the compute controller identified by $\mathit{cnid}$.
\end{itemize}

We use the projections
\[
\operatorname{subj}(c),
\quad
\operatorname{res}(c),
\quad
\operatorname{perm}(c),
\quad
\operatorname{type}(c),
\]
\[
\operatorname{rnid}(c),
\quad
\operatorname{cnid}(c),
\quad
\operatorname{rid}(c),
\quad
\operatorname{cid}(c)
\]
for the corresponding fields.

Let
$
\mathbb{C}_p,\,
\mathbb{C}_c,
\,
\mathbb{C}_r
$
denote the process, compute, and resource capability domains, respectively, and let
$
\mathbb{C}
=
\mathbb{C}_p
\cup
\mathbb{C}_c
\cup
\mathbb{C}_r.
$

The identifier pairs
$
\langle\mathit{rnid},\mathit{rid}\rangle
$
and
$
\langle\mathit{cnid},\mathit{cid}\rangle
$
are controller-scoped authority references. Neither $\mathit{rid}$ nor $\mathit{cid}$ is a globally unique capability identifier.

The value $0$ denotes an absent or inapplicable identifier and is never allocated as an $\mathit{rid}$ or $\mathit{cid}$.

\end{definition}

\begin{definition}[Capability-Layer Field Conventions]
\label{def:capability-fields}

A usable resource capability has the form
\[
c_r
=
\langle
s_r,r,v,\mathtt{usable},
\mathit{rnid},
\mathit{cnid},
\mathit{rid},
0
\rangle,
\]
where
\[
s_r\in\mathbb{RC},
\]
\[
\mathit{rnid}
=
\operatorname{NID}_r(s_r),
\]
and
\[
\mathit{rid}\neq 0.
\]

Every allocated usable resource capability is issued for a particular compute node and therefore satisfies
\[
\operatorname{cnid}(c_r)\neq 0.
\]

The controller-owned resource-tree root $c_0$ is the exception and satisfies
\[
\operatorname{cnid}(c_0)=0.
\]

A compute capability has the form
\[
c_c
=
\langle
s_c,r,v,t,
\mathit{rnid},
\mathit{cnid},
\mathit{rid},
\mathit{cid}
\rangle,
\]
where
\[
s_c\in\mathbb{CC},
\]
\[
\mathit{cnid}
=
\operatorname{NID}_c(s_c),
\]
and
\[
\mathit{rnid},
\mathit{rid},
\mathit{cid}
\neq 0.
\]

A process capability has the form
\[
c_p
=
\langle
p,r,v,t,
0,
\mathit{cnid},
0,
\mathit{cid}
\rangle,
\]
where
\[
p\in\mathbb{P},
\qquad
\mathit{cnid}\neq 0,
\qquad
\mathit{cid}\neq 0.
\]

Thus, process capabilities expose only compute-side linkage. They do not carry the resource-controller-local $\mathit{rid}$ or the resource-node identifier $\mathit{rnid}$.

This separation is a layering property rather than a secrecy property: SADRA does not rely on hiding linkage identifiers for security.

\end{definition}

\begin{definition}[Authority Attenuation]
\label{def:attenuation}

For capabilities $c,c'\in\mathbb{C}$, we write
\[
c'\sqsubseteq c
\]
if and only if
\[
\operatorname{res}(c')
\subseteq
\operatorname{res}(c)
\]
and
\[
\operatorname{perm}(c')
\subseteq
\operatorname{perm}(c).
\]

Thus, derived authority cannot increase either resource extent or permissions.

The $\mathtt{exclusive}$ permission has an additional preservation rule for the initial allocation chain: if strong isolation is requested, $\mathtt{exclusive}$ must remain present in the corresponding resource, root-anchored compute, and process capabilities. Ordinary attenuation cannot clear it along that initial chain.

\end{definition}

\begin{definition}[Exclusive and Delegable Authority]
\label{def:delegable}

For any usable capability $c$, define
\[
\mathrm{Exclusive}(c)
\iff
\mathtt{exclusive}
\in
\operatorname{perm}(c).
\]

We define
\[
\begin{aligned}
\mathrm{Delegable}(c)
\iff
&\mathtt{delegate}
\in
\operatorname{perm}(c)\\
&\land\,
\mathtt{exclusive}
\notin
\operatorname{perm}(c).
\end{aligned}
\]

Thus, the presence of $\mathtt{delegate}$ is necessary but not sufficient for delegation. Exclusive authority is never delegable.

For a non-exclusive delegated capability, the delegator determines whether the recipient may re-delegate by including or omitting $\mathtt{delegate}$ from the delegated permission set.

\end{definition}

\begin{definition}[Trusted Request Context]
\label{def:request-context}

Let
\[
\mathcal{Q}
\]
be the set of requests processed by SADRA, and let
\[
\mathcal{Q}_{\mathsf{acc}}
\subseteq
\mathcal{Q}
\]
denote resource-access requests.

For
\[
req\in\mathcal{Q}_{\mathsf{acc}},
\]
define
\[
\operatorname{Op}(req)
\in
\mathbb{V}_{\mathsf{acc}}
\]
as the requested resource-access operation and
\[
\operatorname{Target}(req)
\subseteq
\mathbb{R}
\]
as the requested resource extent.

For a request presented by a process to a compute controller,
\[
\operatorname{Requester}(req)\in\mathbb{P}
\]
denotes the requesting process identity observed by the trusted compute-side request path.

For a request received by a resource controller,
\[
\operatorname{Origin}(req)\in\mathbb{N}_c
\]
denotes the originating compute-node identity established by the trusted resource-side request path.

$\operatorname{Requester}(req)$ and $\operatorname{Origin}(req)$ are trusted enforcement observations. They are not identities supplied by untrusted host software.

The requester observation applies to process-initiated operations authorized by a process capability or process revocation handle, including resource access, delegation, and handle-based revocation.

The origin observation applies to controller-to-controller operations whose authorization depends on the originating compute node, including inter-node delegation and remote revocation.

\end{definition}

\begin{definition}[Capability Authenticity]
\label{def:authenticity}

For a controller
\[
x\in
\mathbb{CC}
\cup
\mathbb{RC},
\]
we write
\[
\operatorname{Authentic}_x(c)
\]
when the capability representation presented to $x$ passes the controller's capability-authentication check.

Operational capability representations contain a MAC tag. The MAC is omitted from the logical capability tuple because it does not contribute authority semantics.

The authentication domain is controller-specific. A compute controller authenticates process capabilities and process revocation handles presented to it. A resource controller authenticates the resource-controller-created root-anchored compute capabilities and compute revocation handles presented to it.

For each controller, the MAC protects every capability field consumed by that controller's authorization decision. A field used only as local metadata by another controller need not belong to that controller's MAC domain.

Consequently, knowledge or fabrication of linkage identifiers alone does not create usable authority.

\end{definition}

\subsection{Capability Trees Model}
\label{app:formal_cap_trees}

\SADRA{} represents controller-resident capability relationships using resource and compute capability trees. Tree edges record immediate structural dependencies. For usable capabilities, these edges represent authority derivation. A handle edge in a compute tree records the revocation dependency created by inter-node delegation.

The relations $\preceq_r$ and $\preceq_c$ define the authority and linkage constraints that must hold along resource-tree and compute-tree derivations, respectively.

\begin{definition}[Tree Reachability]
\label{def:tree-reachability}

Let
\[
T=\langle V,E,\rho\rangle
\]
be a rooted directed tree. We write
\[
\operatorname{Reach}_T(c,c')
\]
if and only if there exists a directed path from $c$ to $c'$ in $T$.

Equivalently, $\operatorname{Reach}_T$ is the reflexive-transitive closure of the edge relation $E$.
\end{definition}

\paragraph{Resource Capability Tree.}

\begin{definition}[Resource Delegation Relation]
\label{def:preceq_r}

For
\[
c,c'\in\mathbb{C}_r,
\]
we write
\[
c'\preceq_r c
\]
if and only if:

\begin{enumerate}
    \item
    $
    \operatorname{res}(c')
    \subseteq
    \operatorname{res}(c);
    $

    \item
    $
    \operatorname{perm}(c')
    \subseteq
    \operatorname{perm}(c);
    $

    \item
    $
    \mathtt{delegate}
    \in
    \operatorname{perm}(c);
    $
    
    \item
    $
    \mathtt{exclusive}
    \notin
    \operatorname{perm}(c);
    $

    \item
    $
    \operatorname{type}(c)
    =
    \mathtt{usable};
    $

    \item
    $
    \operatorname{type}(c')
    =
    \mathtt{usable};
    $
    and

    \item
    $
    \operatorname{subj}(c')
    =
    \operatorname{subj}(c).
    $
\end{enumerate}

Thus, resource-tree delegation can only attenuate resource extent and permissions. The parent must carry delegation authority and must not represent an exclusive allocation. An exclusive resource capability therefore cannot have a resource capability derived from it.

Because resource capabilities held by the same resource controller carry that controller's resource-node identifier,
\[
\operatorname{subj}(c')
=
\operatorname{subj}(c)
\]
also implies
\[
\operatorname{rnid}(c')
=
\operatorname{rnid}(c).
\]

The compute-node identifiers need not be equal. In particular, inter-node delegation may derive a resource capability for a destination compute node different from the compute node for which its parent was issued.

The relation $\preceq_r$ constrains a permitted derivation but does not itself allocate a new resource-controller-local authority identifier or define a parent--child edge.
\end{definition}

\begin{definition}[Resource Capability Tree]
\label{def:res-tree}

For a resource controller
\[
rc\in\mathbb{RC},
\]
let
\[
c_0\in\mathbb{C}_r
\]
be its controller-owned root resource capability satisfying
\[
\operatorname{subj}(c_0)=rc,
\]
\[
\operatorname{type}(c_0)=\mathtt{usable},
\]
and
\[
\operatorname{cnid}(c_0)=0.
\]

The resource capability tree is
\[
T_r(c_0)
=
\langle
V_r,E_r,c_0
\rangle,
\]
where
\[
V_r\subseteq\mathbb{C}_r,
\qquad
c_0\in V_r,
\]
and
\[
E_r\subseteq V_r\times V_r.
\]

An edge
\[
(c,c')\in E_r
\]
represents an immediate resource-authority derivation from $c$ to $c'$.
\end{definition}

\begin{definition}[Well-Formed Resource Capability Tree]
\label{def:wf-res-tree}

A resource capability tree
\[
T_r(c_0)=\langle V_r,E_r,c_0\rangle
\]
maintained by resource controller $rc$ is well-formed if and only if:

\begin{enumerate}
    \item $T_r(c_0)$ is acyclic;

    \item $c_0$ has no incoming edge;

    \item every
          $
          c'\in V_r\setminus\{c_0\}
          $
          has exactly one parent in $V_r$;

    \item every
          $
          c\in V_r
          $
          is reachable from $c_0$;

    \item every
          $
          c\in V_r
          $
          satisfies
          $
          \operatorname{subj}(c)=rc;
          $

    \item every allocated resource capability
          $
          c\in V_r\setminus\{c_0\}
          $
          satisfies
          $
          \operatorname{cnid}(c)\neq 0;
          $

    \item resource-controller-local authority identifiers are unique
          among resource-tree entries:
          \[
          \forall c,c'\in V_r,\quad
          \operatorname{rid}(c)
          =
          \operatorname{rid}(c')
          \Rightarrow
          c=c';
          \]
          and

    \item for every
          $
          (c,c')\in E_r,\;
          $
          $
          c'\preceq_r c.
          $
\end{enumerate}

\end{definition}

\paragraph{Compute Capability Tree.}

\begin{definition}[Compute Delegation Relation]
\label{def:preceq_c}

For
\[
c,c'\in\mathbb{C}_c,
\]
we write
\[
c'\preceq_c c
\]
if and only if:

\begin{enumerate}
    \item
    $
    \operatorname{res}(c')
    \subseteq
    \operatorname{res}(c);
    $

    \item
    $
    \operatorname{perm}(c')
    \subseteq
    \operatorname{perm}(c);
    $

    \item
    $
    \mathtt{delegate}
    \in
    \operatorname{perm}(c);
    $
    
    \item
    $
    \mathtt{exclusive}
    \notin
    \operatorname{perm}(c);
    $

    \item
    $
    \operatorname{type}(c)
    =
    \mathtt{usable};
    $

    \item
    $
    \operatorname{type}(c')
    \in
    \{
        \mathtt{usable},
        \mathtt{handle}
    \};
    $

    \item
    $
    \operatorname{subj}(c')
    =
    \operatorname{subj}(c);
    $

    \item
    $
    \operatorname{rnid}(c')
    =
    \operatorname{rnid}(c);
    $
    and

    \item if
          $
          \operatorname{type}(c')
          =
          \mathtt{usable},
          $
          then
          $
          \operatorname{rid}(c')
          =
          \operatorname{rid}(c).
          $
\end{enumerate}

Thus, a usable compute capability derived within a compute tree retains
the resource-side linkage of its parent:
\[
\left\langle
\operatorname{rnid}(c'),
\operatorname{rid}(c')
\right\rangle
=
\left\langle
\operatorname{rnid}(c),
\operatorname{rid}(c)
\right\rangle.
\]

It nevertheless receives a distinct compute-controller-local authority identifier when created.

A compute handle created during inter-node delegation has different semantics. It remains on the same resource node,
\[
\operatorname{rnid}(c')
=
\operatorname{rnid}(c),
\]
but its $\mathit{rid}$ identifies the newly derived resource capability whose revocation the handle controls. Therefore, for a handle child,
\[
\operatorname{rid}(c')
\]
need not equal
\[
\operatorname{rid}(c).
\]

Because compute-tree nodes held by the same compute controller carry that controller's compute-node identifier,
\[
\operatorname{subj}(c')
=
\operatorname{subj}(c)
\]
also implies
\[
\operatorname{cnid}(c')
=
\operatorname{cnid}(c).
\]

The relation $\preceq_c$ constrains structural derivation within a compute tree. Allocation of a new compute-controller-local authority identifier is defined by the corresponding capability-creation transition.

The parent must carry delegation authority and must not represent an exclusive allocation. Thus, a root-anchored or derived compute capability carrying $\mathtt{exclusive}$ cannot acquire a child through intra-node or inter-node delegation.

This restriction does not prevent the compute controller from creating the process capability corresponding to an exclusive compute capability. Process-capability creation establishes the process-to-compute linkage rather than a child edge in $T_c(cc)$ and therefore is not governed by $\preceq_c$.
\end{definition}

\begin{definition}[Compute Capability Tree]
\label{def:compute-tree}

For a compute controller
\[
cc\in\mathbb{CC},
\]
the compute capability tree is
\[
T_c(cc)
=
\langle
V_c,E_c,\bot
\rangle,
\]
where
\[
V_c\subseteq\mathbb{C}_c
\]
and
\[
E_c
\subseteq
(V_c\cup\{\bot\})\times V_c.
\]

The symbol
\[
\bot\notin\mathbb{C}
\]
denotes a distinguished virtual root. It carries no authority and does not participate in $\preceq_c$.

An edge
\[
(\bot,c)\in E_c
\]
is structural only and identifies a compute capability directly attached to the virtual root.

An edge
\[
(c,c')\in E_c,
\qquad
c\neq\bot,
\]
represents an immediate structural derivation relationship.

For a usable child, the edge represents authority delegation. For a handle child, it records the revocation dependency created when authority is delegated to another compute node.
\end{definition}

\begin{definition}[Root-Anchored Compute Capability]
\label{def:root-anchored}

Let
\[
T_c(cc)
=
\langle
V_c,E_c,\bot
\rangle
\]
be the compute capability tree maintained by compute controller
\[
cc\in\mathbb{CC}.
\]

A compute capability
\[
c_c\in\mathbb{C}_c
\]
is root-anchored at $cc$ if
\[
\operatorname{RootAnchored}_{cc}(c_c)
\iff
c_c\in V_c
\land
(\bot,c_c)\in E_c.
\]

Root anchoring is a structural property. Usability of root-anchored capabilities follows from compute-tree well-formedness.
\end{definition}

\begin{definition}[Well-Formed Compute Capability Tree]
\label{def:wf-comp-tree}

A compute capability tree
\[
T_c(cc)
=
\langle
V_c,E_c,\bot
\rangle
\]
is well-formed if and only if:

\begin{enumerate}
    \item $T_c(cc)$ is acyclic;

    \item $\bot$ has no incoming edge;

    \item every
          \[
          c'\in V_c
          \]
          has exactly one parent in
          \[
          V_c\cup\{\bot\};
          \]

    \item every
          \[
          c\in V_c
          \]
          is reachable from $\bot$;

    \item every
          \[
          c\in V_c
          \]
          satisfies
          \[
          \operatorname{subj}(c)=cc;
          \]

    \item compute-controller-local authority identifiers are unique among compute-tree entries:
          \[
          \forall c,c'\in V_c,\quad
          \operatorname{cid}(c)
          =
          \operatorname{cid}(c')
          \Rightarrow
          c=c';
          \]

    \item for every
          \[
          (\bot,c)\in E_c,
          \]
          \[
          \operatorname{type}(c)
          =
          \mathtt{usable};
          \]
          and

    \item for every
          \[
          (c,c')\in E_c
          \quad\text{with}\quad
          c\neq\bot,
          \]
          \[
          c'\preceq_c c.
          \]
\end{enumerate}

\end{definition}

\begin{lemma}[Compute Tree Root Reachability]
\label{lem:compute-root-reachability}

For every
\[
c_c\in V_c
\]
in a well-formed compute capability tree
\[
T_c(cc),
\]
there exists a unique root-anchored compute capability
\[
c_c^{\mathsf{ra}}\in V_c
\]
such that
\[
\operatorname{RootAnchored}_{cc}
(c_c^{\mathsf{ra}})
\]
and
\[
\operatorname{Reach}_{T_c(cc)}
(c_c^{\mathsf{ra}},c_c).
\]
\end{lemma}

\begin{proof}

By Definition~\ref{def:wf-comp-tree}, every
\[
c_c\in V_c
\]
is reachable from $\bot$. Moreover, every capability in $V_c$ has exactly one parent in
\[
V_c\cup\{\bot\},
\]
and the tree is acyclic. Hence there is a unique path from $\bot$ to $c_c$.

Let
\[
c_c^{\mathsf{ra}}
\]
be the first capability on this path after $\bot$. Then
\[
(\bot,c_c^{\mathsf{ra}})\in E_c,
\]
and therefore
\[
\operatorname{RootAnchored}_{cc}
(c_c^{\mathsf{ra}})
\]
by Definition~\ref{def:root-anchored}.

Since $c_c^{\mathsf{ra}}$ lies on the path from $\bot$ to $c_c$,
\[
\operatorname{Reach}_{T_c(cc)}
(c_c^{\mathsf{ra}},c_c)
\]
holds.

For uniqueness, suppose another root-anchored compute capability
\[
\hat{c}_c^{\mathsf{ra}}
\neq
c_c^{\mathsf{ra}}
\]
also satisfies
\[
\operatorname{Reach}_{T_c(cc)}
(\hat{c}_c^{\mathsf{ra}},c_c).
\]

Then there would be two distinct paths from $\bot$ to $c_c$, one beginning with $c_c^{\mathsf{ra}}$ and one beginning with $\hat{c}_c^{\mathsf{ra}}$, contradicting the unique-parent property of the well-formed compute capability tree.

Therefore $c_c^{\mathsf{ra}}$ is unique.
\end{proof}

\begin{lemma}[Resource-Link Inheritance]
\label{lem:resource-link-inheritance}

Let
\[
c_c,c_c^{\mathsf{ra}}\in V_c
\]
be usable compute capabilities in a well-formed compute capability tree $T_c(cc)$ such that
\[
\operatorname{RootAnchored}_{cc}
(c_c^{\mathsf{ra}})
\]
and
\[
\operatorname{Reach}_{T_c(cc)}
(c_c^{\mathsf{ra}},c_c).
\]

Then
\[
\operatorname{rnid}(c_c)
=
\operatorname{rnid}(c_c^{\mathsf{ra}})
\]
and
\[
\operatorname{rid}(c_c)
=
\operatorname{rid}(c_c^{\mathsf{ra}}).
\]
\end{lemma}

\begin{proof}

If
\[
c_c=c_c^{\mathsf{ra}},
\]
the result is immediate.

Otherwise, let
\[
c_c^{\mathsf{ra}}
=
d_0,d_1,\ldots,d_k
=
c_c
\]
be the unique path from $c_c^{\mathsf{ra}}$ to $c_c$.

Because $c_c$ is usable, every node on this derivation path is usable: a handle cannot serve as the parent of another compute-tree node under Definition~\ref{def:preceq_c}.

For each edge
\[
(d_i,d_{i+1}),
\]
Definition~\ref{def:preceq_c} therefore gives
\[
\operatorname{rnid}(d_{i+1})
=
\operatorname{rnid}(d_i)
\]
and
\[
\operatorname{rid}(d_{i+1})
=
\operatorname{rid}(d_i).
\]

Applying these equalities along the path yields
\[
\operatorname{rnid}(c_c)
=
\operatorname{rnid}(c_c^{\mathsf{ra}})
\]
and
\[
\operatorname{rid}(c_c)
=
\operatorname{rid}(c_c^{\mathsf{ra}}).
\]
\end{proof}

\begin{definition}[Fence-Anchor Domains]
\label{def:fence-anchor-domains}

Fences are installed on usable authority capabilities rather than on handle capabilities.

The resource fence-anchor domain is
\[
\mathbb{A}_r
=
\left\{
c\in\mathbb{C}_r
\;\middle|\;
\operatorname{type}(c)=\mathtt{usable}
\right\},
\]
and the compute fence-anchor domain is
\[
\mathbb{A}_c
=
\left\{
c\in\mathbb{C}_c
\;\middle|\;
\operatorname{type}(c)=\mathtt{usable}
\right\}.
\]

Handle capabilities may exist in compute capability trees but cannot serve as fence anchors. The virtual root $\bot$ is not a fence anchor.
\end{definition}

\subsection{System State and Capability Linking}
\label{app:formal_system_state}

The capability trees define the structural relationships among controller-resident capabilities. The linkage fields defined in \Cref{def:capability} identify corresponding authority across capability layers. The pair $\langle\mathit{cnid},\mathit{cid}\rangle$ links a process-visible capability to compute-side authority, while $\langle\mathit{rnid},\mathit{rid}\rangle$ links compute-side authority
to resource-side authority.

These linkage relations are predicates over capability objects rather than mutable components of the system state. They describe which authority a capability refers to, independently of whether that authority is currently live or active. Capability activity, authenticity, and request-specific authorization are defined separately in
\Cref{app:cap-validity}.

\begin{definition}[System State]
\label{def:system-state}

A system state is a tuple
\[
\Sigma =
(
\Sigma_P,
\Sigma_{CC},
\Sigma_{RC},
E_c,
E_r,
\operatorname{ctr}_{CC},
\operatorname{ctr}_{RC},
F_{CC},
F_{RC}
),
\]
where:

\begin{itemize}
    \item
    $\Sigma_P : \mathbb{P}\rightarrow 2^{\mathbb{C}_p}$ maps each process to the process capabilities currently associated with that process;

    \item
    $\Sigma_{CC} : \mathbb{CC}\rightarrow 2^{\mathbb{C}_c}$ maps each compute controller to the compute capabilities currently stored locally;

    \item
    $\Sigma_{RC} : \mathbb{RC}\rightarrow 2^{\mathbb{C}_r}$ maps each resource controller to the resource capabilities currently stored locally;

    \item
    $E_c$ and $E_r$ denote the current edge relations of the compute and resource capability trees, respectively;

    \item
    $\operatorname{ctr}_{CC} :
    \mathbb{CC}\rightarrow\mathbb{N}_{>0}$ maps each compute controller to its next compute-controller-local authority identifier;

    \item
    $\operatorname{ctr}_{RC} :
    \mathbb{RC}\rightarrow\mathbb{N}_{>0}$ maps each resource controller to its next resource-controller-local authority identifier;

    \item
    $F_{CC}$ and $F_{RC}$ denote the active fence sets maintained by compute and resource controllers.
\end{itemize}

For every process $p$, compute controller $cc$, and resource controller $rc$, stored capability state is consistent with capability subjects:
\[
c_p\in\Sigma_P(p)
\Rightarrow
\operatorname{subj}(c_p)=p,
\]
\[
c_c\in\Sigma_{CC}(cc)
\Rightarrow
\operatorname{subj}(c_c)=cc,
\]
and
\[
c_r\in\Sigma_{RC}(rc)
\Rightarrow
\operatorname{subj}(c_r)=rc.
\]

The vertex sets of the controller capability trees are derived from the corresponding controller-resident state:
\[
V_c(cc)
\equiv
\Sigma_{CC}(cc)
\]
and
\[
V_r(rc)
\equiv
\Sigma_{RC}(rc).
\]
Thus, $V_c$ and $V_r$ are not additional mutable components of
$\Sigma$.

A compute or resource capability is \emph{live} while it remains in the corresponding controller-resident capability tree. Removal from that tree therefore makes the capability no longer live. Liveness alone does not imply that the capability carries currently usable authority: a live capability may be covered by an active fence.

A process may retain a stale capability token in untrusted memory after its controller-side authority has been revoked or reclaimed. Its linkage fields remain well defined, but the token cannot authorize a request unless the corresponding controller-side authority is active and the request satisfies the authenticity and request-binding checks defined in \Cref{app:cap-validity}.

The local authority-identifier counters are part of the trusted persistent controller state. They are monotonically increasing and never wrap. The value $0$ is reserved and is never allocated. Before authority carrying a newly allocated local identifier becomes usable, the issuing controller durably advances the corresponding counter. Consequently, a previously allocated $\mathit{rid}$ or $\mathit{cid}$ is never reassigned to different authority at the same controller.

\end{definition}

\begin{definition}[Process-to-Compute Linking Relation]
\label{def:process-compute-link}

For
\[
c_p\in\mathbb{C}_p
\qquad\text{and}\qquad
c_c\in\mathbb{C}_c,
\]
we write
\[
c_p\mapsto c_c
\]
if and only if
\[
\begin{aligned}
&\operatorname{type}(c_p)=\mathtt{usable}
\\
&\land\;
\operatorname{type}(c_c)=\mathtt{usable}
\\
&\land\;
\operatorname{cnid}(c_p)
=
\operatorname{cnid}(c_c)
\\
&\land\;
\operatorname{cid}(c_p)
=
\operatorname{cid}(c_c)
\\
&\land\;
\operatorname{res}(c_p)
\subseteq
\operatorname{res}(c_c)
\\
&\land\;
\operatorname{perm}(c_p)
\subseteq
\operatorname{perm}(c_c).
\end{aligned}
\]

Thus, the pair
\[
\left\langle
\operatorname{cnid}(c_p),
\operatorname{cid}(c_p)
\right\rangle
\]
identifies the compute-side authority referenced by the process capability. The resource and permission conditions enforce attenuation of the authority exposed to the process.

The relation is independent of membership in $\Sigma_P$ or $\Sigma_{CC}$. Hence, a presented, replayed, stale, or fabricated token always has a well-defined linkage predicate. Whether that token can authorize a request is determined separately by controller-side activity, authenticity, requester binding, and the requested operation.

Among live entries of a well-formed compute capability tree, the target of a process-to-compute link is unique because compute-controller-local authority identifiers are unique within that tree.

Process-held handle capabilities do not participate in this usable authorization relation.

\end{definition}

\begin{definition}[Compute-to-Resource Linking Relation]
\label{def:compute-resource-link}

For
\[
c_c\in\mathbb{C}_c
\qquad\text{and}\qquad
c_r\in\mathbb{C}_r,
\]
we write
\[
c_c\mapsto c_r
\]
if and only if
\[
\begin{aligned}
&\operatorname{type}(c_c)=\mathtt{usable}
\\
&\land\;
\operatorname{type}(c_r)=\mathtt{usable}
\\
&\land\;
\operatorname{rnid}(c_c)
=
\operatorname{rnid}(c_r)
\\
&\land\;
\operatorname{rid}(c_c)
=
\operatorname{rid}(c_r)
\\
&\land\;
\operatorname{cnid}(c_c)
=
\operatorname{cnid}(c_r)
\\
&\land\;
\operatorname{res}(c_c)
\subseteq
\operatorname{res}(c_r)
\\
&\land\;
\operatorname{perm}(c_c)
\subseteq
\operatorname{perm}(c_r).
\end{aligned}
\]

The pair
\[
\left\langle
\operatorname{rnid}(c_c),
\operatorname{rid}(c_c)
\right\rangle
\]
identifies the resource-side authority from which $c_c$ derives. Equality of the $\mathit{cnid}$ fields records that this resource-side authority was issued for the compute node represented by $c_c$.

The relation expresses an authority association rather than activity or liveness. It therefore does not require either capability to belong to the current controller-resident state.

For a non-root usable compute capability, the resource-side linkage is inherited from its root-anchored ancestor. By \Cref{lem:resource-link-inheritance}, if $c_c^{\mathsf{ra}}$ is the root-anchored ancestor of usable capability $c_c$, then
\[
\operatorname{rnid}(c_c)
=
\operatorname{rnid}(c_c^{\mathsf{ra}})
\]
and
\[
\operatorname{rid}(c_c)
=
\operatorname{rid}(c_c^{\mathsf{ra}}).
\]

Thus, a non-root usable compute capability remains associated with the same resource-side authority as its root-anchored ancestor, although the non-root capability may carry attenuated resource extent and permissions.

Compute handle capabilities do not participate in this usable authorization relation.

\end{definition}

\begin{definition}[Root-Anchored Capability Selection]
\label{def:root-anchor-selection}

For every usable compute capability
\[
c_c\in V_c(cc)
\]
in a well-formed compute capability tree $T_c(cc)$, the root-anchored capability associated with $c_c$ is the unique capability
\[
c_c^{\mathsf{ra}}\in V_c(cc)
\]
satisfying
\[
\operatorname{RootAnchored}_{cc}
(c_c^{\mathsf{ra}})
\]
and
\[
\operatorname{Reach}_{T_c(cc)}
(c_c^{\mathsf{ra}},c_c).
\]

Its existence and uniqueness follow from \Cref{lem:compute-root-reachability}.

By \Cref{lem:resource-link-inheritance},
\[
\operatorname{rnid}(c_c)
=
\operatorname{rnid}(c_c^{\mathsf{ra}})
\]
and
\[
\operatorname{rid}(c_c)
=
\operatorname{rid}(c_c^{\mathsf{ra}}).
\]

The tree relation establishes the structural ancestry of the root-anchored capability, while the inherited resource-side linkage identifies the resource authority associated with that root.

\end{definition}

\begin{definition}[Capability Chain]
\label{def:capability-chain}

The capability chain represents the composition of usable authority across the three SADRA capability layers.

For a process capability $c_p$, compute capability $c_c$, and resource capability $c_r$, the authority chain has the form
\[
c_p
\mapsto
c_c
\mapsto
c_r.
\]

The first relation identifies the compute-side authority referenced by $c_p$. The second relation identifies the resource-side authority from which $c_c$ derives.

For a non-root compute capability, the second relation expresses its resource-side authority association through the inherited $\langle\mathit{rnid},\mathit{rid}\rangle$ linkage. The actual root-anchored capability presented to the resource controller is
determined during resource-access authorization and is defined separately in \Cref{app:cap-validity}.

Thus, the capability chain describes authority composition rather than the sequence of capability representations presented at the two enforcement boundaries.

Handle capabilities do not form part of the usable authorization chain.

\end{definition}

\begin{definition}[Process-Handle-to-Compute Linking Relation]
\label{def:process-handle-link}

For
\[
h_p\in\mathbb{C}_p
\qquad\text{and}\qquad
c_a\in\mathbb{C}_c,
\]
we write
\[
h_p\xmapsto{H}c_a
\]
if and only if
\[
\begin{aligned}
&\operatorname{type}(h_p)=\mathtt{handle}
\\
&\land\;
\operatorname{cnid}(h_p)
=
\operatorname{cnid}(c_a)
\\
&\land\;
\operatorname{cid}(h_p)
=
\operatorname{cid}(c_a).
\end{aligned}
\]

The target $c_a$ may be either a usable compute capability or a compute handle.

For intra-node delegation,
\[
\operatorname{type}(c_a)=\mathtt{usable},
\]
and the relation identifies the delegated compute subtree controlled by the process handle.

For inter-node delegation,
\[
\operatorname{type}(c_a)=\mathtt{handle},
\]
and the relation identifies the source-side compute revocation handle.

The relation is derived solely from capability fields and does not require $h_p$ or $c_a$ to be live. Among live compute-tree entries at the controller identified by $\operatorname{cnid}(h_p)$, its target is unique because compute-controller-local authority identifiers are unique within that tree.

\end{definition}

\begin{definition}[Compute-Handle-to-Resource Linking Relation]
\label{def:compute-handle-link}

For
\[
h_c\in\mathbb{C}_c
\qquad\text{and}\qquad
c_r\in\mathbb{C}_r,
\]
we write
\[
h_c\xmapsto{R}c_r
\]
if and only if
\[
\begin{aligned}
&\operatorname{type}(h_c)=\mathtt{handle}
\\
&\land\;
\operatorname{type}(c_r)=\mathtt{usable}
\\
&\land\;
\operatorname{rnid}(h_c)
=
\operatorname{rnid}(c_r)
\\
&\land\;
\operatorname{rid}(h_c)
=
\operatorname{rid}(c_r).
\end{aligned}
\]

The pair
\[
\left\langle
\operatorname{rnid}(h_c),
\operatorname{rid}(h_c)
\right\rangle
\]
identifies the resource capability whose revocation the compute handle controls.

The $\mathit{cnid}$ fields are intentionally not required to match. The compute handle belongs to the delegating compute node, while the resource capability created by inter-node delegation is issued for the destination compute node.

The relation is independent of current resource-controller state. Among live resource capabilities at the resource controller identified by $\operatorname{rnid}(h_c)$, the target is unique because resource-controller-local authority identifiers are unique within that resource tree.

\end{definition}

\begin{definition}[Revocation-Control Path]
\label{def:revocation-control}

Revocation-control relationships are separate from usable authorization relationships.

For intra-node delegation, a process-held handle identifies the usable compute capability created for the delegation:
\[
h_p
\xmapsto{H}
c_c^{\mathsf{del}}.
\]

For inter-node delegation, the delegating process's handle identifies the compute revocation handle retained by the source compute controller, which in turn identifies the resource capability created for the delegation:
\[
h_p^{\mathsf{src}}
\xmapsto{H}
h_c^{\mathsf{src}}
\xmapsto{R}
c_r^{\mathsf{del}}.
\]

The corresponding usable authority at the destination follows
\[
c_p^{\mathsf{dst}}
\mapsto
c_c^{\mathsf{dst}}
\mapsto
c_r^{\mathsf{del}}.
\]

Thus, authorization and revocation control refer to the same delegated authority but follow distinct capability paths. Process and compute handles convey revocation control only and do not participate in the usable authorization chain.

\end{definition}

\begin{definition}[Fence State]
\label{def:fence-state}

The fence state is represented by
\[
F_{CC} :
\mathbb{CC}\rightarrow2^{\mathbb{A}_c}
\]
and
\[
F_{RC} :
\mathbb{RC}\rightarrow2^{\mathbb{A}_r},
\]
where $\mathbb{A}_c$ and $\mathbb{A}_r$ are the compute and resource fence-anchor domains defined in \Cref{def:fence-anchor-domains}.

A fence identifies a usable capability at which authority becomes inactive at that controller. Descendants need not appear individually in a fence set because the activity check determines whether the capability path contains an active fence ancestor.

Handle capabilities cannot serve as fence anchors, and the virtual root $\bot$ is not a fence anchor.

\end{definition}

\subsection{Capability Activity and Authorization}
\label{app:cap-validity}

This subsection distinguishes local capability activity from end-to-end resource-access authorization.

A controller-resident capability is \emph{live} if it remains in the controller's capability tree. It is \emph{active} if it is live, usable, and unaffected by the controller's current fence state. Activity is therefore controller-local: compute-side activity depends only on compute-controller state, while resource-side activity depends only on resource-controller state.

Activity does not by itself imply that a resource-access request is authorized. A request must first satisfy the compute-side acceptance predicate, after which the compute controller resolves and forwards the corresponding root-anchored compute capability. The resource controller then performs an independent resource-side acceptance check. End-to-end authorization requires both stages to accept the request.

This distinction permits the two enforcement points to observe different revocation state. In particular, a process capability may continue to resolve to active compute-side authority after the corresponding resource-side authority has been fenced. Such a request may pass the compute-side check but is rejected at the resource boundary.

\paragraph{Local Capability Activity.}

\begin{definition}[Fence Path Clearance]
\label{def:fence-path-validity}

Let $c_r$ be a resource capability in resource capability tree
\[
T_r(c_0)=\langle V_r,E_r,c_0\rangle
\]
maintained by resource controller $rc$. We define
\[
\mathrm{NoFence}_r(c_r,\Sigma)
\iff
\forall f\in F_{RC}(rc),\;
\neg\operatorname{Reach}_{T_r(c_0)}(f,c_r).
\]

Let $c_c$ be a compute capability in compute capability tree
\[
T_c(cc)=\langle V_c,E_c,\bot\rangle
\]
maintained by compute controller $cc$. We define
\[
\mathrm{NoFence}_c(c_c,\Sigma)
\iff
\forall f\in F_{CC}(cc),\;
\neg\operatorname{Reach}_{T_c(cc)}(f,c_c).
\]

A live capability is clear of the local fence state if no active fence anchor is an ancestor of that capability in the corresponding capability tree.

\end{definition}

\begin{definition}[Permission Predicate]
\label{def:permission-predicate}

For a capability $c\in\mathbb{C}$ and permission
$op\in\mathbb{V}$, define
\[
\mathrm{Permitted}(c,op)
\iff
op\in\operatorname{perm}(c).
\]

\end{definition}

\begin{definition}[Resource Capability Activity]
\label{def:resource-validity}

Let
\[
c_r\in\mathbb{C}_r
\]
and let
\[
rc=\operatorname{subj}(c_r).
\]

The resource capability is active in system state $\Sigma$ if
\[
\begin{aligned}
\mathrm{Active}_r(c_r,\Sigma)
\iff\;&
c_r\in V_r(rc)
\\
&\land\;
\operatorname{type}(c_r)=\mathtt{usable}
\\
&\land\;
\mathrm{NoFence}_r(c_r,\Sigma).
\end{aligned}
\]

Thus, $\mathrm{Active}_r$ depends only on resource-controller state.

\end{definition}

\begin{definition}[Compute Capability Activity]
\label{def:compute-validity}

Let
\[
c_c\in\mathbb{C}_c
\]
and let
\[
cc=\operatorname{subj}(c_c).
\]

The compute capability is active in system state $\Sigma$ if
\[
\begin{aligned}
\mathrm{Active}_c(c_c,\Sigma)
\iff\;&
c_c\in V_c(cc)
\\
&\land\;
\operatorname{type}(c_c)=\mathtt{usable}
\\
&\land\;
\mathrm{NoFence}_c(c_c,\Sigma).
\end{aligned}
\]

The predicate deliberately does not depend on resource-controller state. Consequently, a compute capability may remain active after its corresponding resource-side authority has become inactive at the resource controller.

\end{definition}

\begin{definition}[Active Process-to-Compute Resolution]
\label{def:live-process-compute}

For
\[
c_p\in\mathbb{C}_p
\qquad\text{and}\qquad
c_c\in\mathbb{C}_c,
\]
we write
\[
c_p\xmapsto{\Sigma}c_c
\]
if and only if
\[
c_p\mapsto c_c
\land
\mathrm{Active}_c(c_c,\Sigma).
\]

The relation therefore means that the process capability refers to currently active compute-side authority.

Because compute-controller-local authority identifiers are unique among live entries of a well-formed compute capability tree, a process capability resolves to at most one active compute capability in any system state.

\end{definition}

\begin{definition}[Process Capability Activity]
\label{def:process-validity}

A process capability
\[
c_p\in\mathbb{C}_p
\]
is compute-locally active in system state $\Sigma$ if
\[
\begin{aligned}
\mathrm{Active}_p(c_p,\Sigma)
\iff\;&
\operatorname{type}(c_p)=\mathtt{usable}
\\
&\land\;
\exists c_c\in\mathbb{C}_c:
c_p\xmapsto{\Sigma}c_c.
\end{aligned}
\]

Process-capability activity is a local property. It states that the process capability resolves to active compute-side authority; it does not state that a resource-access request using the capability will pass resource-side authorization.

Thus,
\[
\mathrm{Active}_p(c_p,\Sigma)
\]
may hold even when the resource authority associated with its compute capability is no longer active at the resource controller.

\end{definition}

\begin{lemma}[Fence-Induced Inactivity]
\label{lem:monotonic-revocation}

Let $T_x$ be a capability tree maintained by controller $x_0$, where $x\in\{r,c\}$. Let $f$ be an active fence anchor at $x_0$. If
\[
\operatorname{Reach}_{T_x}(f,c),
\]
then
\[
\neg\mathrm{Active}_x(c,\Sigma).
\]

\end{lemma}

\begin{proof}

Since $f$ is an active fence anchor and
\[
\operatorname{Reach}_{T_x}(f,c),
\]
\Cref{def:fence-path-validity} implies
\[
\neg\mathrm{NoFence}_x(c,\Sigma).
\]

The corresponding activity predicate requires
$\mathrm{NoFence}_x(c,\Sigma)$. Therefore
\[
\neg\mathrm{Active}_x(c,\Sigma).
\]
\end{proof}

\paragraph{Two-Stage Resource-Access Authorization.}

\begin{definition}[Compute-Side Acceptance]
\label{def:compute-accept}

Let
\[
req\in\mathcal{Q}_{\mathsf{acc}}
\]
be a resource-access request presented to compute controller
\[
cc\in\mathbb{CC}.
\]

For process capability
\[
c_p\in\mathbb{C}_p
\]
and compute capability
\[
c_c\in\mathbb{C}_c,
\]
the compute controller accepts the request if
\[
\begin{aligned}
\mathrm{ComputeAccept}_{cc}
(req,c_p,& c_c,\Sigma)
\iff\;\\
&c_p\xmapsto{\Sigma}c_c
\\
&\land\;
\operatorname{cnid}(c_c)
=
\operatorname{NID}_c(cc)
\\
&\land\;
\operatorname{Authentic}_{cc}(c_p)
\\
&\land\;
\operatorname{subj}(c_p)
=
\operatorname{Requester}(req)
\\
&\land\;
\operatorname{Target}(req)
\subseteq
\operatorname{res}(c_p)
\\
&\land\;
\mathrm{Permitted}
\bigl(c_p,\operatorname{Op}(req)\bigr).
\end{aligned}
\]

Because
\[
c_p\mapsto c_c
\]
requires
\[
\operatorname{res}(c_p)
\subseteq
\operatorname{res}(c_c)
\]
and
\[
\operatorname{perm}(c_p)
\subseteq
\operatorname{perm}(c_c),
\]
compute-side acceptance also implies
\[
\operatorname{Target}(req)
\subseteq
\operatorname{res}(c_c)
\]
and
\[
\mathrm{Permitted}
\bigl(c_c,\operatorname{Op}(req)\bigr).
\]

The predicate depends only on the presented process capability, trusted compute-side request context, and current compute-controller state. It does not inspect resource-controller activity or fence state.

\end{definition}

\begin{definition}[Resource-Access Capability Resolution]
\label{def:resource-access-resolution}

Suppose
\[
\mathrm{ComputeAccept}_{cc}
(req,c_p,c_c,\Sigma)
\]
holds.

The compute controller resolves the corresponding root-anchored compute capability
\[
c_c^{\mathsf{ra}}
\]
according to \Cref{def:root-anchor-selection}. We write
\[
c_c
\leadsto
c_c^{\mathsf{ra}}
\]
if
\[
\operatorname{RootAnchored}_{cc}
(c_c^{\mathsf{ra}})
\]
and
\[
\operatorname{Reach}_{T_c(cc)}
(c_c^{\mathsf{ra}},c_c)
\]
hold.

By \Cref{lem:resource-link-inheritance},
\[
\operatorname{rnid}(c_c)
=
\operatorname{rnid}(c_c^{\mathsf{ra}})
\]
and
\[
\operatorname{rid}(c_c)
=
\operatorname{rid}(c_c^{\mathsf{ra}}).
\]

Thus, $c_c$ and $c_c^{\mathsf{ra}}$ identify the same resource-side authority, although the non-root compute capability may contain more attenuated resource extent or permissions.

The compute controller forwards the resource-controller-created root-anchored capability $c_c^{\mathsf{ra}}$, rather than $c_p$ or a non-root compute capability, to the resource controller.

The operational path after compute-side acceptance is therefore
\[
c_p
\mapsto
c_c
\leadsto
c_c^{\mathsf{ra}}
\mapsto
c_r.
\]

The relation $\leadsto$ denotes local root-capability resolution rather than an additional authority-delegation relation.

\end{definition}

\begin{definition}[Resource-Side Acceptance]
\label{def:resource-accept}

Let
\[
req\in\mathcal{Q}_{\mathsf{acc}}
\]
be a resource-access request forwarded to resource controller
\[
rc\in\mathbb{RC}
\]
with root-anchored compute capability
\[
c_c^{\mathsf{ra}}.
\]

Let
\[
c_r\in\mathbb{C}_r
\]
be the resource capability identified by the resource-side linkage of $c_c^{\mathsf{ra}}$.

The resource controller accepts the request if
\[
\begin{aligned}
\mathrm{ResourceAccept}_{rc}
(req,c_c^{\mathsf{ra}},& c_r,\Sigma)
\iff\;\\
&c_c^{\mathsf{ra}}
\mapsto
c_r
\\
&\land\;
\mathrm{Active}_r(c_r,\Sigma)
\\
&\land\;
\operatorname{rnid}(c_r)
=
\operatorname{NID}_r(rc)
\\
&\land\;
\operatorname{Authentic}_{rc}
(c_c^{\mathsf{ra}})
\\
&\land\;
\operatorname{Origin}(req)
=
\operatorname{cnid}(c_r)
\\
&\land\;
\operatorname{Target}(req)
\subseteq
\operatorname{res}(c_c^{\mathsf{ra}})
\\
&\land\;
\mathrm{Permitted}
\bigl(
c_c^{\mathsf{ra}},
\operatorname{Op}(req)
\bigr).
\end{aligned}
\]

Since
\[
c_c^{\mathsf{ra}}\mapsto c_r
\]
requires
\[
\operatorname{res}(c_c^{\mathsf{ra}})
\subseteq
\operatorname{res}(c_r)
\]
and
\[
\operatorname{perm}(c_c^{\mathsf{ra}})
\subseteq
\operatorname{perm}(c_r),
\]
resource-side acceptance also implies
\[
\operatorname{Target}(req)
\subseteq
\operatorname{res}(c_r)
\]
and
\[
\mathrm{Permitted}
\bigl(c_r,\operatorname{Op}(req)\bigr).
\]

The origin check
\[
\operatorname{Origin}(req)
=
\operatorname{cnid}(c_r)
\]
binds ordinary resource access to the compute node for which the resource authority was issued. Here $\operatorname{Origin}(req)$ is the trusted compute-node origin established by the resource-side request path as defined in \Cref{def:request-context}.

Resource-side acceptance depends only on the forwarded resource-controller-authenticated capability, trusted resource-side request context, and current resource-controller state. It does not depend on current process or compute-tree state at the originating compute controller.

\end{definition}

\begin{lemma}[Authoritative Resource-Side Revocation]
\label{lem:cross-layer-revocation}

Let
\[
c_c^{\mathsf{ra}}\mapsto c_r.
\]

If
\[
\neg\mathrm{Active}_r(c_r,\Sigma),
\]
then for every
\[
req\in\mathcal{Q}_{\mathsf{acc}},
\]
\[
\neg
\mathrm{ResourceAccept}_{rc}
(req,c_c^{\mathsf{ra}},c_r,\Sigma),
\]
where
\[
rc=\operatorname{subj}(c_r).
\]

This conclusion does not require either
\[
\neg\mathrm{Active}_c(c_c,\Sigma)
\]
or
\[
\neg\mathrm{Active}_p(c_p,\Sigma)
\]
at the compute controller.

\end{lemma}

\begin{proof}

By \Cref{def:resource-accept}, resource-side acceptance requires
\[
\mathrm{Active}_r(c_r,\Sigma).
\]

Therefore,
\[
\neg\mathrm{Active}_r(c_r,\Sigma)
\]
implies
\[
\neg
\mathrm{ResourceAccept}_{rc}
(req,c_c^{\mathsf{ra}},c_r,\Sigma).
\]

The predicates $\mathrm{Active}_c$ and $\mathrm{Active}_p$ depend only on compute-controller state and therefore need not change when the resource-side authority becomes inactive.
\end{proof}

\begin{definition}[End-to-End Authorization]
\label{def:authorization}

For
\[
req\in\mathcal{Q}_{\mathsf{acc}},
\]
end-to-end resource access is authorized if and only if
\[
\begin{aligned}
\mathrm{Allow}(&req,\Sigma)
\iff\\
&\exists\;
cc\in\mathbb{CC},
rc\in\mathbb{RC},
c_p\in\mathbb{C}_p,
c_c,c_c^{\mathsf{ra}}\in\mathbb{C}_c,
c_r\in\mathbb{C}_r:
\\
&\:
\mathrm{ComputeAccept}_{cc}
(req,c_p,c_c,\Sigma)
\\
&\:\land\;
c_c\leadsto c_c^{\mathsf{ra}}
\\
&\:\land\;
\mathrm{ResourceAccept}_{rc}
(req,c_c^{\mathsf{ra}},c_r,\Sigma).
\end{aligned}
\]

Thus, an access is allowed only when both enforcement stages accept the request. The compute controller validates process-bound local authority before fabric admission, while the resource controller independently validates the forwarded root-anchored authority against current resource-side state immediately before resource access.

A request may therefore satisfy
\[
\mathrm{ComputeAccept}_{cc}
(req,c_p,c_c,\Sigma)
\]
while failing
\[
\mathrm{ResourceAccept}_{rc}
(req,c_c^{\mathsf{ra}},c_r,\Sigma).
\]
In that case,
\[
\neg\mathrm{Allow}(req,\Sigma).
\]

\end{definition}

\TWOSTENF*

By Definition~\ref{def:authorization}, the invariant is realized by
\[
\begin{aligned}
\mathrm{Allow}(req,\Sigma)
\iff
&\mathrm{ComputeAccept}_{cc}(req,c_p,c_c,\Sigma)\\
&\,\land
c_c\leadsto c_c^{\mathsf{ra}}\\
&\,\land
\mathrm{ResourceAccept}_{rc}
(req,c_c^{\mathsf{ra}},c_r,\Sigma).\\
\end{aligned}
\]

\subsection{Capability-Related Operations}
\label{app:cap-ops}

Capability-related operations are modeled as transitions over capability trees, controller-local authority identifiers, fence state, and capability ownership state. Cross-layer authorization and revocation-control relationships follow from the linkage fields defined in \Cref{app:formal_system_state}; they are not maintained as independent mutable relations.

Each transition preserves the capability-tree well-formedness conditions defined in \Cref{app:formal_cap_trees}. New controller-local authority identifiers are obtained from the persistent counters defined in \Cref{def:system-state}. A controller durably advances the corresponding counter before authority carrying the new identifier becomes usable, and an allocated identifier is never reassigned to different authority.

\begin{definition}[Resource Initialization]
\label{def:resource-init}

Let
\[
rc\in\mathbb{RC}
\]
be a resource controller managing
\[
\mathbb{R}_{rc}.
\]

Let
\[
rid_0
=
\operatorname{ctr}_{RC}(rc).
\]

The resource controller first durably advances
\[
\operatorname{ctr}'_{RC}(rc)
=
rid_0+1
\]
and then creates the controller-owned root capability
\[
c_0
=
\left\langle
rc,
\mathbb{R}_{rc},
\mathbb{V}\setminus\{\mathtt{exclusive}\},
\mathtt{usable},
\operatorname{NID}_r(rc),
0,
rid_0,
0
\right\rangle.
\]

The initial resource capability tree is
\[
T_r(c_0)
=
\left\langle
\{c_0\},
\emptyset,
c_0
\right\rangle.
\]

The root capability is controller-owned rather than allocated to a compute node and therefore satisfies
\[
\operatorname{cnid}(c_0)=0.
\]

\end{definition}

\begin{definition}[Allocation Disjointness]
\label{def:allocation-disjointness}

Let
\[
T_r(c_0)
=
\langle
V_r,E_r,c_0
\rangle
\]
be a resource capability tree.

Resource capabilities produced by independent allocation operations are direct children of $c_0$. A candidate allocated resource capability $c_r$ satisfies allocation disjointness if
\[
\begin{aligned}
&\mathrm{AllocationDisjoint}(c_r,T_r)
\iff\\
&\qquad\qquad\forall a\in V_r:
(c_0,a)\in E_r
\Rightarrow
\operatorname{res}(c_r)
\cap
\operatorname{res}(a)
=
\emptyset.
\end{aligned}
\]

Thus, simultaneously live allocations always designate disjoint resource extents, independently of whether strong isolation was requested.

Because every capability derived from an allocation root is attenuated to a subset of that root's extent, disjointness of allocation roots also implies disjointness between authority derived from different allocations.

\end{definition}

\begin{definition}[Capability Allocation]
\label{def:allocation}

Capability allocation establishes an initial authority chain consisting of a resource capability, a root-anchored compute capability, and a process capability.

Let
\[
rc\in\mathbb{RC},
\qquad
cc\in\mathbb{CC},
\qquad
p\in\mathbb{P}.
\]

Let
\[
req_{\mathsf{alloc}}
\]
be an admitted allocation request specifying a requested resource extent $r_{\mathsf{req}}$, an ordinary requested permission set
\[
v_{\mathsf{req}}
\subseteq
\mathbb{V}\setminus\{\mathtt{exclusive}\},
\]
and whether strong isolation is requested.

SADRA constructs the resource permission set as
\[
v_r
=
\begin{cases}
v_{\mathsf{req}}
\cup
\{\mathtt{exclusive}\},
&
\text{if strong isolation is requested},
\\[1mm]
v_{\mathsf{req}},
&
\text{otherwise}.
\end{cases}
\]

Thus, $\mathtt{exclusive}$ is set by SADRA as a consequence of the strong-isolation request rather than introduced by delegation.

Let
\[
rid_r
=
\operatorname{ctr}_{RC}(rc).
\]
Before the allocated authority becomes usable, the resource controller durably advances
\[
\operatorname{ctr}'_{RC}(rc)
=
rid_r+1.
\]

The resource controller creates
\[
c_r
=
\left\langle
rc,
r_r,
v_r,
\mathtt{usable},
\operatorname{NID}_r(rc),
\operatorname{NID}_c(cc),
rid_r,
0
\right\rangle,
\]
where
\[
r_r
\subseteq
r_{\mathsf{req}}
\cap
\mathbb{R}_{rc},
\]
\[
\operatorname{res}(c_r)
\subseteq
\operatorname{res}(c_0),
\]
and
\[
\mathrm{AllocationDisjoint}
(c_r,T_r(c_0)).
\]

Allocation is not resource delegation. In particular, it does not require
\[
c_r\preceq_r c_0.
\]
This distinction permits initial allocation to introduce $\mathtt{exclusive}$ authority even though exclusive authority cannot subsequently be delegated.

The resource tree becomes
\[
T'_r(c_0)
=
\left\langle
V_r(rc)\cup\{c_r\},
E_r(rc)\cup\{(c_0,c_r)\},
c_0
\right\rangle.
\]

The corresponding compute controller allocates
\[
cid_c
=
\operatorname{ctr}_{CC}(cc)
\]
and durably advances
\[
\operatorname{ctr}'_{CC}(cc)
=
cid_c+1
\]
before the resulting compute authority becomes usable.

The corresponding root-anchored compute capability is
\[
c_c
=
\left\langle
cc,
r_c,
v_c,
\mathtt{usable},
\operatorname{NID}_r(rc),
\operatorname{NID}_c(cc),
rid_r,
cid_c
\right\rangle,
\]
where
\[
\operatorname{res}(c_c)
\subseteq
\operatorname{res}(c_r),
\]
\[
\operatorname{perm}(c_c)
\subseteq
\operatorname{perm}(c_r),
\]
and
\[
\mathrm{Exclusive}(c_c)
\iff
\mathrm{Exclusive}(c_r).
\]

It is inserted directly below the structural compute-tree root:
\[
T'_c(cc)
=
\left\langle
V_c(cc)\cup\{c_c\},
E_c(cc)\cup\{(\bot,c_c)\},
\bot
\right\rangle.
\]

The compute controller creates the corresponding process capability
\[
c_p
=
\left\langle
p,
r_p,
v_p,
\mathtt{usable},
0,
\operatorname{NID}_c(cc),
0,
cid_c
\right\rangle,
\]
where
\[
\operatorname{res}(c_p)
\subseteq
\operatorname{res}(c_c),
\]
\[
\operatorname{perm}(c_p)
\subseteq
\operatorname{perm}(c_c),
\]
and
\[
\mathrm{Exclusive}(c_p)
\iff
\mathrm{Exclusive}(c_c).
\]

Hence, when strong isolation is requested,
\[
\mathtt{exclusive}
\in
\operatorname{perm}(c_r)
\cap
\operatorname{perm}(c_c)
\cap
\operatorname{perm}(c_p).
\]

Creating $c_p$ does not create a child in either capability tree and is not a delegation operation. It therefore remains permitted when the allocated authority is exclusive.

The initial linkage is
\[
c_p
\mapsto
c_c
\mapsto
c_r.
\]

The state is updated with
\[
c_r\in\Sigma'_{RC}(rc),
\qquad
c_c\in\Sigma'_{CC}(cc),
\qquad
c_p\in\Sigma'_P(p).
\]

The process capability is authenticated by the compute controller. The resource-facing representation of the root-anchored compute capability is authenticated by the resource controller.

Initial allocation creates no revocation handle.

\end{definition}

\begin{definition}[Intra-Node Delegation]
\label{def:intra-delegation}

Intra-node delegation derives attenuated authority for another process managed by the same compute controller. The delegating process retains its original usable authority.

Let
\[
req_{\mathsf{del}}^{P}
\]
be an intra-node delegation request presented to compute controller
\[
cc\in\mathbb{CC}.
\]

Let
\[
p
=
\operatorname{Requester}(req_{\mathsf{del}}^{P})
\]
denote the delegating process, and let
\[
p'\in\mathbb{P}
\]
be the receiving process.

Let
\[
c_p^{\mathsf{src}}\in\mathbb{C}_p
\]
be the process capability carried by the request, and let
\[
c_c^{\mathsf{src}}\in V_c(cc)
\]
be its associated compute capability.

The operation requires
\[
c_p^{\mathsf{src}}
\xmapsto{\Sigma}
c_c^{\mathsf{src}},
\]
\[
\operatorname{Authentic}_{cc}
(c_p^{\mathsf{src}}),
\]
\[
\operatorname{subj}(c_p^{\mathsf{src}})
=
\operatorname{Requester}(req_{\mathsf{del}}^{P}),
\]
and
\[
\mathrm{Delegable}(c_p^{\mathsf{src}}).
\]

Thus,
\[
\mathtt{delegate}
\in
\operatorname{perm}(c_p^{\mathsf{src}})
\]
and
\[
\mathtt{exclusive}
\notin
\operatorname{perm}(c_p^{\mathsf{src}}).
\]

Let
\[
cid_{\mathsf{del}}
=
\operatorname{ctr}_{CC}(cc).
\]

Before making the delegated authority usable, the compute controller durably advances
\[
\operatorname{ctr}'_{CC}(cc)
=
cid_{\mathsf{del}}+1.
\]

It creates the delegated usable compute capability
\[
c'_c
=
\left\langle
cc,
r'_c,
v'_c,
\mathtt{usable},
\operatorname{rnid}(c_c^{\mathsf{src}}),
\operatorname{cnid}(c_c^{\mathsf{src}}),
\operatorname{rid}(c_c^{\mathsf{src}}),
cid_{\mathsf{del}}
\right\rangle,
\]
where
\[
c'_c
\preceq_c
c_c^{\mathsf{src}},
\]
\[
\operatorname{res}(c'_c)
\subseteq
\operatorname{res}(c_p^{\mathsf{src}}),
\]
and
\[
\operatorname{perm}(c'_c)
\subseteq
\operatorname{perm}(c_p^{\mathsf{src}}).
\]

Because delegation requires non-exclusive source authority and delegated permissions are attenuated,
\[
\mathtt{exclusive}
\notin
\operatorname{perm}(c'_c).
\]
Thus, delegation cannot introduce the $\mathtt{exclusive}$ permission.

The compute tree becomes
\[
T'_c(cc)
=
\left\langle
V_c(cc)\cup\{c'_c\},
E_c(cc)\cup
\{(c_c^{\mathsf{src}},c'_c)\},
\bot
\right\rangle.
\]

The recipient receives
\[
c'_p
=
\left\langle
p',
r'_p,
v'_p,
\mathtt{usable},
0,
\operatorname{NID}_c(cc),
0,
cid_{\mathsf{del}}
\right\rangle,
\]
where
\[
\operatorname{res}(c'_p)
\subseteq
\operatorname{res}(c'_c)
\]
and
\[
\operatorname{perm}(c'_p)
\subseteq
\operatorname{perm}(c'_c).
\]

Thus,
\[
c'_p
\mapsto
c'_c.
\]

The delegator controls whether the recipient may subsequently re-delegate this authority through the delegated permission set. If
\[
\mathtt{delegate}
\in
\operatorname{perm}(c'_p),
\]
the recipient may re-delegate subject to the remaining delegation
conditions. If
\[
\mathtt{delegate}
\notin
\operatorname{perm}(c'_p),
\]
the recipient cannot re-delegate the capability.

The delegating process receives a separate process revocation handle
\[
h_p
=
\left\langle
p,
r'_c,
v'_c,
\mathtt{handle},
0,
\operatorname{NID}_c(cc),
0,
cid_{\mathsf{del}}
\right\rangle.
\]

The handle and delegated compute authority intentionally share the same compute-side linkage:
\[
h_p
\xmapsto{H}
c'_c.
\]

The handle conveys revocation control only and is not inserted into the compute capability tree.

The state is updated with
\[
c'_c\in\Sigma'_{CC}(cc),
\qquad
c'_p\in\Sigma'_P(p'),
\qquad
h_p\in\Sigma'_P(p).
\]

The original capabilities
$c_p^{\mathsf{src}}$ and $c_c^{\mathsf{src}}$ remain unchanged.

\end{definition}

\begin{definition}[Inter-Node Delegation]
\label{def:inter-delegation}

Inter-node delegation derives attenuated resource authority for a process managed by a different compute controller.

Let
\[
req_{\mathsf{del}}^{P}
\]
be an inter-node delegation request presented by a process to source compute controller
\[
cc\in\mathbb{CC}.
\]

Let
\[
p
=
\operatorname{Requester}(req_{\mathsf{del}}^{P})
\]
denote the delegating process.

Let
\[
cc'\in\mathbb{CC},
\qquad
cc'\neq cc,
\]
be the destination compute controller, and let
\[
p'\in\mathbb{P}
\]
be the receiving process.

Let
\[
c_p^{\mathsf{src}}\in\mathbb{C}_p
\]
be the process capability carried by
$req_{\mathsf{del}}^{P}$, and let
\[
c_c^{\mathsf{src}}\in V_c(cc)
\]
be its associated compute capability.

The source compute controller requires
\[
c_p^{\mathsf{src}}
\xmapsto{\Sigma}
c_c^{\mathsf{src}},
\]
\[
\operatorname{Authentic}_{cc}
(c_p^{\mathsf{src}}),
\]
\[
\operatorname{subj}(c_p^{\mathsf{src}})
=
\operatorname{Requester}(req_{\mathsf{del}}^{P}),
\]
and
\[
\mathrm{Delegable}(c_p^{\mathsf{src}}).
\]

Let
\[
c_c^{\mathsf{ra}}
\]
be the root-anchored compute capability associated with $c_c^{\mathsf{src}}$ according to \Cref{def:root-anchor-selection}, and let
\[
c_c^{\mathsf{ra}}
\mapsto
c_r.
\]

After source-side authorization, the compute controller sends a controller-to-controller delegation request
\[
req_{\mathsf{del}}^{R}
\]
to the authoritative resource controller $rc$.

The resource controller verifies
\[
\mathrm{Active}_r(c_r,\Sigma),
\]
\[
\operatorname{Authentic}_{rc}
(c_c^{\mathsf{ra}}),
\]
and
\[
\operatorname{Origin}(req_{\mathsf{del}}^{R})
=
\operatorname{cnid}(c_r)
=
\operatorname{NID}_c(cc).
\]

The distinction between $req_{\mathsf{del}}^{P}$ and $req_{\mathsf{del}}^{R}$ separates the trusted process identity observed at the source compute controller from the trusted compute-node origin observed at the resource controller.

The newly delegated resource authority must satisfy
\[
c'_r
\preceq_r
c_r.
\]

Consequently, resource-side delegation requires
\[
\mathtt{delegate}
\in
\operatorname{perm}(c_r)
\]
and
\[
\mathtt{exclusive}
\notin
\operatorname{perm}(c_r).
\]

Therefore, an initially allocated resource marked $\mathtt{exclusive}$ cannot be delegated to another compute node.

Let
\[
rid_{\mathsf{del}}
=
\operatorname{ctr}_{RC}(rc).
\]

Before making the new resource authority usable, the resource controller durably advances
\[
\operatorname{ctr}'_{RC}(rc)
=
rid_{\mathsf{del}}+1.
\]

It creates
\[
c'_r
=
\left\langle
rc,
r'_r,
v'_r,
\mathtt{usable},
\operatorname{NID}_r(rc),
\operatorname{NID}_c(cc'),
rid_{\mathsf{del}},
0
\right\rangle,
\]
where
\[
c'_r
\preceq_r
c_r,
\]
\[
\operatorname{res}(c'_r)
\subseteq
\operatorname{res}(c_p^{\mathsf{src}}),
\]
and
\[
\operatorname{perm}(c'_r)
\subseteq
\operatorname{perm}(c_p^{\mathsf{src}}).
\]

Since the source authority is non-exclusive and permissions are attenuated,
\[
\mathtt{exclusive}
\notin
\operatorname{perm}(c'_r).
\]

The resource tree becomes
\[
T'_r(c_0)
=
\left\langle
V_r(rc)\cup\{c'_r\},
E_r(rc)\cup\{(c_r,c'_r)\},
c_0
\right\rangle.
\]

The destination compute controller allocates
\[
cid_{\mathsf{dst}}
=
\operatorname{ctr}_{CC}(cc')
\]
and durably advances
\[
\operatorname{ctr}'_{CC}(cc')
=
cid_{\mathsf{dst}}+1.
\]

The corresponding destination root-anchored compute capability is
\[
c_c^{\mathsf{dst}}
=
\left\langle
cc',
r_{\mathsf{dst}},
v_{\mathsf{dst}},
\mathtt{usable},
\operatorname{NID}_r(rc),
\operatorname{NID}_c(cc'),
rid_{\mathsf{del}},
cid_{\mathsf{dst}}
\right\rangle,
\]
where
\[
\operatorname{res}(c_c^{\mathsf{dst}})
\subseteq
\operatorname{res}(c'_r)
\]
and
\[
\operatorname{perm}(c_c^{\mathsf{dst}})
\subseteq
\operatorname{perm}(c'_r).
\]

It is inserted directly below the destination compute-tree root:
\[
T'_c(cc')
=
\left\langle
V_c(cc')\cup\{c_c^{\mathsf{dst}}\},
E_c(cc')\cup
\{(\bot,c_c^{\mathsf{dst}})\},
\bot
\right\rangle.
\]

The receiving process obtains
\[
c_p^{\mathsf{dst}}
=
\left\langle
p',
r_{\mathsf{dst},p},
v_{\mathsf{dst},p},
\mathtt{usable},
0,
\operatorname{NID}_c(cc'),
0,
cid_{\mathsf{dst}}
\right\rangle,
\]
where
\[
\operatorname{res}(c_p^{\mathsf{dst}})
\subseteq
\operatorname{res}(c_c^{\mathsf{dst}})
\]
and
\[
\operatorname{perm}(c_p^{\mathsf{dst}})
\subseteq
\operatorname{perm}(c_c^{\mathsf{dst}}).
\]

The usable destination chain is therefore
\[
c_p^{\mathsf{dst}}
\mapsto
c_c^{\mathsf{dst}}
\mapsto
c'_r.
\]

As with intra-node delegation, the source delegator determines whether the destination may subsequently re-delegate the authority by including or omitting $\mathtt{delegate}$ from the delegated permission set.

The source compute controller allocates a fresh local identifier
\[
cid_h
=
\operatorname{ctr}_{CC}(cc)
\]
for the inter-node compute revocation handle and durably advances
\[
\operatorname{ctr}'_{CC}(cc)
=
cid_h+1.
\]

The source-side compute handle is
\[
h_c^{\mathsf{src}}
=
\left\langle
cc,
r'_r,
v'_r,
\mathtt{handle},
\operatorname{NID}_r(rc),
\operatorname{NID}_c(cc),
rid_{\mathsf{del}},
cid_h
\right\rangle.
\]

It is inserted below the usable source compute capability:
\[
T'_c(cc)
=
\left\langle
V_c(cc)\cup\{h_c^{\mathsf{src}}\},
E_c(cc)\cup
\{(c_c^{\mathsf{src}},h_c^{\mathsf{src}})\},
\bot
\right\rangle.
\]

The handle belongs structurally to the source compute subtree but carries the $\mathit{rid}$ of the newly delegated resource capability:
\[
\operatorname{rid}(h_c^{\mathsf{src}})
=
rid_{\mathsf{del}}.
\]

Hence,
\[
h_c^{\mathsf{src}}
\xmapsto{R}
c'_r.
\]

The delegating process receives
\[
h_p^{\mathsf{src}}
=
\left\langle
p,
r'_r,
v'_r,
\mathtt{handle},
0,
\operatorname{NID}_c(cc),
0,
cid_h
\right\rangle,
\]
giving
\[
h_p^{\mathsf{src}}
\xmapsto{H}
h_c^{\mathsf{src}}.
\]

The inter-node revocation-control path is therefore
\[
h_p^{\mathsf{src}}
\xmapsto{H}
h_c^{\mathsf{src}}
\xmapsto{R}
c'_r.
\]

The resource-facing destination root capability and the resource-facing fields of the source compute revocation handle are authenticated by the resource controller. The source-local $\mathit{cid}$ is used for compute-side handle resolution and need not be part of the resource-controller authentication domain because the resource controller does not use it for revocation authorization.

The state is updated with
\[
c'_r\in\Sigma'_{RC}(rc),
\]
\[
c_c^{\mathsf{dst}}\in\Sigma'_{CC}(cc'),
\qquad
c_p^{\mathsf{dst}}\in\Sigma'_P(p'),
\]
\[
h_c^{\mathsf{src}}\in\Sigma'_{CC}(cc),
\qquad
h_p^{\mathsf{src}}\in\Sigma'_P(p).
\]

Neither handle conveys usable resource-access authority.

\end{definition}

\begin{definition}[Escaped Inter-Node Revocation Targets]
\label{def:escaped-targets}

Let $c_a$ be a usable compute capability in
\[
T_c(cc)
=
\langle
V_c(cc),E_c(cc),\bot
\rangle.
\]

The inter-node compute handles structurally dependent on $c_a$ are
\[
\operatorname{Esc}(c_a,\Sigma)
=
\left\{
h_c\in V_c(cc)
\;\middle|\;
\begin{aligned}
&\operatorname{type}(h_c)=\mathtt{handle}
\\
&\land\;
\operatorname{Reach}_{T_c(cc)}(c_a,h_c)
\end{aligned}
\right\}.
\]

The corresponding live resource-side targets are
\[
\operatorname{EscR}(c_a,\Sigma)
=
\left\{
c_r
\;\middle|\;
\begin{aligned}
&\exists h_c\in
\operatorname{Esc}(c_a,\Sigma):
\\
&h_c\xmapsto{R}c_r
\land
c_r\in V_r(\operatorname{subj}(c_r))
\end{aligned}
\right\}.
\]

Each such handle represents authority that escaped the local compute subtree through an earlier inter-node delegation.

Because compute-tree derivation preserves $\mathit{rnid}$, all compute handles reachable from one usable compute subtree refer to capabilities maintained by the same authoritative resource controller, although the handles may contain different $\mathit{rid}$ values.

\end{definition}

\begin{definition}[Resource-Side Revocation Authorization]
\label{def:resource-revocation-authorization}

Let
\[
req_{\mathsf{rev}}^{R}
\]
be a controller-to-controller revocation request received by resource controller
\[
rc\in\mathbb{RC}
\]
with compute revocation handle
\[
h_c\in\mathbb{C}_c.
\]

For a live resource capability
\[
c_r\in V_r(rc),
\]
the request is authorized to target $c_r$ if
\[
\begin{aligned}
\mathrm{RevokeAccept}_{rc}
(req_{\mathsf{rev}}^{R},h_c,c_r,& \Sigma)
\iff\\
& h_c\xmapsto{R}c_r
\\
&\land\;
\operatorname{rnid}(c_r)
=
\operatorname{NID}_r(rc)
\\
&\land\;
\operatorname{Authentic}_{rc}(h_c)
\\
&\land\;
\operatorname{Origin}(req_{\mathsf{rev}}^{R})
=
\operatorname{cnid}(h_c).
\end{aligned}
\]

The compute handle is therefore the resource-side revocation credential. Knowledge of an $\mathit{rid}$ alone does not authorize installation of a resource fence.

The origin check binds use of the handle to the source compute node for which the revocation handle was created. The targeted resource capability may itself be issued to a different compute node, namely the destination of the inter-node delegation.

The predicate does not require $c_r$ to be active. Consequently, a repeated revocation request can be treated idempotently if the target is already fenced. If no live capability with the referenced $\langle\mathit{rnid},\mathit{rid}\rangle$ remains, the resource controller may acknowledge that the authority has already been reclaimed.

\end{definition}

\begin{definition}[Owner-Initiated Dis-allocation]
\label{def:disallocation}

Owner-initiated dis-allocation disables authority controlled directly by a controller. Structural removal occurs only after the corresponding fence takes effect.

For a resource controller $rc$, let
\[
T_r(c_0)
=
\langle
V_r(rc),E_r(rc),c_0
\rangle.
\]

A non-root resource capability
\[
c_r\in
V_r(rc)\setminus\{c_0\}
\]
is first fenced:
\[
F'_{RC}(rc)
=
F_{RC}(rc)
\cup
\{c_r\}.
\]

For every
\[
c\in V_r(rc),
\]
\[
\operatorname{Reach}_{T_r(c_0)}(c_r,c)
\Rightarrow
\neg\mathrm{Active}_r(c,\Sigma').
\]

For a compute controller $cc$, let
\[
T_c(cc)
=
\langle
V_c(cc),E_c(cc),\bot
\rangle.
\]

A compute capability may be owner-disallocated directly only if it is a usable root-anchored capability:
\[
c_c\in V_c(cc),
\]
\[
(\bot,c_c)\in E_c(cc),
\]
and
\[
\operatorname{type}(c_c)
=
\mathtt{usable}.
\]

The compute controller first installs
\[
F'_{CC}(cc)
=
F_{CC}(cc)
\cup
\{c_c\}.
\]

For every
\[
c\in V_c(cc),
\]
\[
\operatorname{Reach}_{T_c(cc)}(c_c,c)
\Rightarrow
\neg\mathrm{Active}_c(c,\Sigma').
\]

If
\[
\operatorname{Esc}(c_c,\Sigma)
\neq
\emptyset,
\]
the fenced compute subtree contains authority previously delegated to other compute nodes. For every
\[
h_c\in
\operatorname{Esc}(c_c,\Sigma)
\]
whose resource-side target has not already been secured, the compute controller sends a controller-to-controller revocation request
\[
req_{\mathsf{rev}}^{R}
\]
carrying $h_c$ to the resource controller identified by
\[
\operatorname{rnid}(h_c).
\]

The resource controller installs the corresponding resource fence only after
\[
\mathrm{RevokeAccept}_{rc}
(req_{\mathsf{rev}}^{R},h_c,c_r,\Sigma)
\]
holds for the matching live resource capability $c_r$.

The local compute fence takes effect independently of this propagation. Compute handles representing escaped authority remain available until their resource-side targets have been secured.

The presence of $\mathtt{exclusive}$ does not prevent dis-allocation or revocation. It prohibits delegation, not withdrawal of authority by its authorized owner or controller.

\end{definition}

\begin{definition}[Handle-Based Revocation]
\label{def:handle-revocation}

Handle-based revocation allows a delegating process to revoke authority that it previously delegated.

Let
\[
req_{\mathsf{rev}}^{P}
\]
be a handle-based revocation request presented to compute controller
\[
cc\in\mathbb{CC},
\]
and let
\[
p
=
\operatorname{Requester}(req_{\mathsf{rev}}^{P})
\]
denote the requesting process.

Let
\[
h_p\in\mathbb{C}_p
\]
be the process revocation handle carried by the request.

The compute controller requires
\[
\operatorname{type}(h_p)
=
\mathtt{handle},
\]
\[
\operatorname{Authentic}_{cc}(h_p),
\]
\[
\operatorname{subj}(h_p)
=
\operatorname{Requester}(req_{\mathsf{rev}}^{P}),
\]
and a live compute-side target
\[
h_p\xmapsto{H}c_a,
\qquad
c_a\in V_c(cc).
\]

Thus, an authentic handle presented by a different process does not authorize revocation: the handle subject must match the trusted process identity observed by the compute-side request path.

The operation then follows one of two paths.

\paragraph{Compute-local revocation.}

For an intra-node delegation,
\[
\operatorname{type}(c_a)
=
\mathtt{usable}.
\]

The compute controller first installs
\[
F'_{CC}(cc)
=
F_{CC}(cc)
\cup
\{c_a\}.
\]

For every
\[
c\in V_c(cc),
\]
\[
\operatorname{Reach}_{T_c(cc)}(c_a,c)
\Rightarrow
\neg\mathrm{Active}_c(c,\Sigma').
\]

Thus, local revocation takes effect when the compute fence is installed; structural reclamation is not required for denial.

If
\[
\operatorname{Esc}(c_a,\Sigma)
\neq
\emptyset,
\]
the subtree also contains inter-node delegations. For every
\[
h_c\in
\operatorname{Esc}(c_a,\Sigma)
\]
whose resource-side target has not yet been secured, the compute controller sends a controller-to-controller revocation request
\[
req_{\mathsf{rev}}^{R}
\]
carrying $h_c$ to the authoritative resource controller.

For every matching live resource capability $c_r$, the resource controller installs
\[
F'_{RC}(rc)
=
F_{RC}(rc)
\cup
\{c_r\}
\]
only after
\[
\mathrm{RevokeAccept}_{rc}
(req_{\mathsf{rev}}^{R},h_c,c_r,\Sigma)
\]
holds.

The compute handles remain in the fenced local subtree until their resource-side targets have been secured.

The process handle may be consumed after the local compute fence is durably established, because the corresponding compute handles remain in trusted controller state and preserve any outstanding remote revocation obligations:
\[
\Sigma'_P(p)
=
\Sigma_P(p)
\setminus
\{h_p\}.
\]

\paragraph{Resource-rooted revocation.}

For an inter-node delegation,
\[
c_a=h_c
\]
and
\[
\operatorname{type}(h_c)
=
\mathtt{handle}.
\]

The handle identifies the delegated resource authority through
\[
h_c
\xmapsto{R}
c_r.
\]

The source compute controller sends a controller-to-controller revocation request
\[
req_{\mathsf{rev}}^{R}
\]
carrying $h_c$ to the authoritative resource controller.

The resource controller identifies the source compute node from the trusted request path as
\[
\operatorname{Origin}(req_{\mathsf{rev}}^{R}).
\]

If the matching resource capability remains live and
\[
\mathrm{RevokeAccept}_{rc}
(req_{\mathsf{rev}}^{R},h_c,c_r,\Sigma)
\]
holds, the resource controller installs
\[
F'_{RC}(rc)
=
F_{RC}(rc)
\cup
\{c_r\}.
\]

For every resource capability $c$ in the affected resource subtree,
\[
\operatorname{Reach}_{T_r(c_0)}(c_r,c)
\Rightarrow
\neg\mathrm{Active}_r(c,\Sigma').
\]

By \Cref{lem:cross-layer-revocation}, accesses dependent on the fenced resource authority fail resource-side acceptance even if stale process or compute authority remains active at a destination compute controller.

No compute fence is installed on $h_c$. A compute handle conveys no usable authority and cannot serve as a fence anchor.

The distinction between $req_{\mathsf{rev}}^{P}$ and $req_{\mathsf{rev}}^{R}$ separates the trusted identity of the process exercising $h_p$ from the trusted origin of the compute controller exercising $h_c$.

The source process handle and compute handle are retained until the resource controller acknowledges that the target resource authority has been fenced or has already been reclaimed. After that acknowledgment, they may be retired:
\[
\Sigma'_{CC}(cc)
=
\Sigma_{CC}(cc)
\setminus
\{h_c\},
\]
\[
E'_c(cc)
=
E_c(cc)
\setminus
\left\{
(a,b)
\mid
a=h_c
\lor
b=h_c
\right\},
\]
and
\[
\Sigma'_P(p)
=
\Sigma_P(p)
\setminus
\{h_p\}.
\]

Because a compute handle cannot serve as a parent for compute delegation, direct retirement of the handle cannot orphan another compute-tree node.

Resource-rooted revocation takes effect for resource access when the resource controller installs the authoritative fence. Retirement of source handles and reclamation of stale destination-side state are subsequent cleanup actions.

\end{definition}

\begin{definition}[Resource-Side Security of a Compute Handle]
\label{def:handle-resource-secured}

Let
\[
h_c\in\mathbb{C}_c
\]
be a compute revocation handle, and let $rc$ be the unique resource controller satisfying
\[
\operatorname{NID}_r(rc)
=
\operatorname{rnid}(h_c).
\]

The resource-side authority identified by $h_c$ is secured in state $\Sigma$ if no matching live resource capability remains active:
\[
\begin{aligned}
&\mathrm{ResourceSecured}(h_c,\Sigma)
\iff\\
&\qquad\qquad\qquad \forall c_r\in V_r(rc):\,
h_c\xmapsto{R}c_r
\Rightarrow
\neg\mathrm{Active}_r(c_r,\Sigma).
\end{aligned}
\]

The predicate also holds when the corresponding resource capability has already been reclaimed. Since resource-controller-local identifiers are never reassigned, absence of a matching live capability cannot later cause the same handle to identify different authority.

$\mathrm{ResourceSecured}$ is a specification predicate over global state. A compute controller does not evaluate it by reading resource-controller state directly. Operationally, the condition is established by an authenticated acknowledgment from the authoritative resource controller after the target has been fenced or found already
reclaimed.

\end{definition}

\begin{definition}[Reclamation]
\label{def:reclamation}

Reclamation removes controller-resident capability state only after the corresponding usable authority has been disabled by fencing.

For a compute tree
\[
T_c(cc)
=
\langle
V_c(cc),E_c(cc),\bot
\rangle,
\]
let
\[
V_c^{\mathsf{rec}}
\subseteq
V_c(cc)
\]
be the nodes removed by one reclamation transition.

Every reclaimed compute-tree node must be covered by an active compute
fence:
\[
\forall c\in V_c^{\mathsf{rec}},\quad
\exists f\in F_{CC}(cc):
\operatorname{Reach}_{T_c(cc)}(f,c).
\]

The reclaimed set is descendant-closed:
\[
\forall c\in V_c^{\mathsf{rec}},
\forall d\in V_c(cc),\quad
\operatorname{Reach}_{T_c(cc)}(c,d)
\Rightarrow
d\in V_c^{\mathsf{rec}}.
\]

Every inter-node compute handle removed as part of reclamation must have a secured resource-side target:
\[
\forall h_c\in V_c^{\mathsf{rec}},\quad
\operatorname{type}(h_c)=\mathtt{handle}
\Rightarrow
\mathrm{ResourceSecured}(h_c,\Sigma).
\]

The compute-controller state becomes
\[
\Sigma'_{CC}(cc)
=
\Sigma_{CC}(cc)
\setminus
V_c^{\mathsf{rec}},
\]
and
\[
E'_c(cc)
=
E_c(cc)
\setminus
\left\{
(a,b)
\mid
a\in V_c^{\mathsf{rec}}
\lor
b\in V_c^{\mathsf{rec}}
\right\}.
\]

For a resource tree
\[
T_r(c_0)
=
\langle
V_r(rc),E_r(rc),c_0
\rangle,
\]
let
\[
V_r^{\mathsf{rec}}
\subseteq
V_r(rc)\setminus\{c_0\}.
\]

Every reclaimed resource capability must be covered by an active resource fence:
\[
\forall c\in V_r^{\mathsf{rec}},\quad
\exists f\in F_{RC}(rc):
\operatorname{Reach}_{T_r(c_0)}(f,c).
\]

The resource reclamation set is descendant-closed:
\[
\forall c\in V_r^{\mathsf{rec}},
\forall d\in V_r(rc),\quad
\operatorname{Reach}_{T_r(c_0)}(c,d)
\Rightarrow
d\in V_r^{\mathsf{rec}}.
\]

The resource-controller state becomes
\[
\Sigma'_{RC}(rc)
=
\Sigma_{RC}(rc)
\setminus
V_r^{\mathsf{rec}},
\]
and
\[
E'_r(rc)
=
E_r(rc)
\setminus
\left\{
(a,b)
\mid
a\in V_r^{\mathsf{rec}}
\lor
b\in V_r^{\mathsf{rec}}
\right\}.
\]

Process-held capability tokens need not be removed during reclamation. Once their corresponding controller-resident authority has been reclaimed, they no longer resolve to active authority and cannot authorize subsequent operations.

Reclamation never decrements or rewinds $\operatorname{ctr}_{CC}$ or $\operatorname{ctr}_{RC}$. Identifiers associated with reclaimed authority are therefore never reassigned.

A fence remains active until all capability-tree nodes covered by that fence have been structurally removed. For a fence anchor $f$, let 
\[
\operatorname{Desc}_{T_x}(f)
=
\left\{
c\in V_x
\mid
\operatorname{Reach}_{T_x}(f,c)
\right\}
\]
denote its subtree in the pre-transition state.

The fence may be removed only when
\[
\operatorname{Desc}_{T_x}(f)
\cap
V'_x
=
\emptyset.
\]

\end{definition}

For the following invariant, let
\[
F_x
=
F_{CC}(cc)
\]
when $x=c$, and
\[
F_x
=
F_{RC}(rc)
\]
when $x=r$.

\begin{invariant}[Fence Before Reclamation]
\label{inv:fence-before-reclamation}

For every reachable state, structural reclamation of usable authority is enabled only for capability-tree nodes covered by an active fence.

For every node removed by a reclamation transition,
\[
c\in V_x^{\mathsf{rec}}
\Rightarrow
\exists f\in F_x:
\operatorname{Reach}_{T_x}(f,c).
\]

Every reclamation set is descendant-closed, so removing an internal capability cannot leave an orphaned descendant in the capability tree.

For compute-tree reclamation, an inter-node compute revocation handle may be removed as part of a fenced subtree only after
\[
\mathrm{ResourceSecured}(h_c,\Sigma)
\]
holds. Thus, local reclamation cannot erase the retained reference to escaped inter-node authority before the corresponding resource-side authority has been disabled.

Direct retirement of an inter-node compute handle following an authoritative resource-side acknowledgment is not reclamation of usable authority and therefore does not require a compute fence.

A fence anchor may be removed only after every capability-tree node reachable from that anchor in the pre-transition tree has been removed.

\end{invariant}

\begin{remark}[Fence Effect Without Descendant Enumeration]
\label{rem:no-enumeration}

Installing a fence does not require enumeration of descendant capabilities to make affected authority inactive at the controller where the fence is installed.

The controller records the fence anchor in
\[
F_{CC}
\]
or
\[
F_{RC},
\]
and subsequent activity checks reject usable capabilities whose capability-tree path contains that anchor. The security effect of the local fence therefore does not wait for structural traversal or reclamation.

Cleanup may subsequently traverse the fenced subtree. If a fenced compute subtree contains inter-node compute handles, those handles must also be processed so that the corresponding resource-side authority is secured. This propagation is separate from the local fence effect and does not delay denial at the controller where the fence has already been installed.

\end{remark}